\documentclass[draftcls,onecolumn]{IEEEtran}

\usepackage[utf8]{inputenc}
\usepackage[T1]{fontenc}
\usepackage{geometry}
\usepackage{amsmath,amssymb,amsthm,mathtools}
\usepackage{bbm}
\usepackage{microtype}
\usepackage[colorlinks=true, allcolors=blue]{hyperref}
\usepackage{color}

\newtheorem{definition}{Definition}

\newtheorem{theorem}{Theorem}

\newtheorem{lemma}{Lemma}
\newtheorem{corollary}{Corollary}
\newtheorem{assumption}{Assumption}

\title{Group Testing with Selectable Thresholds}
\author{Trung-Khang Tran, Daniel McMorrow, and Jonathan Scarlett\thanks{The authors are with the Department of Computer Science, National University of Singapore (NUS).  J.~Scarlett is also with the Department of Mathematics and Institute for Data Science at NUS.  This work is supported by the National
University of Singapore under the Presidential Young Professorship grant scheme. }}
\date{\today}
\allowdisplaybreaks

\newcommand{\bX}{\mathbf{X}}

\newcommand{\by}{\mathbf{y}}
\newcommand{\bM}{\mathbf{M}}
\newcommand{\bbR}{\mathbb{R}}
\newcommand{\sd}{\preceq_{\rm st}}

\begin{document}
    
\maketitle

\begin{abstract}
    We consider the problem of group testing, in which one seeks to identify a subset of defective items of size $k$ from a larger set of $n$ items based on pooled tests.  We introduce a \emph{selectable threshold model}, in which each test has an associated threshold that can be chosen, such that the test outcome is 1 if and only if the number of defectives in the test is no smaller than that threshold.  In settings with a large or unbounded maximum threshold, we establish conditions under which high-probability recovery can be attained with a rate (i.e., the asymptotic ratio of $\log_2{n \choose k}$ to the number of tests) approaching its maximum possible value of 1.  Moreover, in the case of a fixed maximum threshold, we establish an achievable number of tests using simple and computationally efficient decoding methods, and a converse that holds under suitable regularity conditions on the test design, with the two coinciding in the dense limit (i.e., $\theta$ approaching one in the scaling $k = \Theta(n^{\theta})$).
\end{abstract}

\section{Introduction}

The group testing problem, in which one seeks to identify a subset of ``defective'' items using pooled tests, has been widely studied since its introduction in the 1940s, and has far-reaching applications including medical testing \cite{aldridge2021pooled}, DNA testing \cite{curnow1998pooling, gille1991pooling}, database systems \cite{cormode2005whats}, communication protocols \cite{hayes1978adaptive,wolf1985born}, and more.  The majority of attention has been paid to the standard binary ``OR'' model, in which a test outcome is positive if and only if it includes at least one defective item, but a number of alternative models have also been proposed, including the following:
\begin{itemize}
    \item In \emph{quantitative group testing} (QGT) \cite{wang2018optimal,feige2020quantitative}, the test outcome reveals the number of defectives in the test, and is thus a non-binary outcome;
    \item In \emph{semi-quantitative group testing} (SQGT) \cite{emad2014semiquantitative,cheraghchi2021semiquantitative}, the test outcome can be viewed as a quantized form of the QGT outcome, so that the outcome reveals a certain \emph{range} that the number of defectives lies in (but not the exact value);
    \item In the simplest form of \emph{threshold group testing} (TGT) \cite{damaschke2006threshold,chen2009nonadaptive,chan2013stochastic}, the test outcome is again binary, and indicates whether or not the number of defectives is greater than or equal to a certain threshold $\gamma$; setting $\gamma=1$ gives regular group testing.
\end{itemize}
In this paper, we study a setting that builds on TGT but is new to the best of our knowledge: For each test, we can not only select which items are included in the test, but also \emph{select the threshold} $\gamma$ such that the test outcome is $Y = \mathbbm{1}\{ \text{\#defectives in test} \ge \gamma\}$ (subject to suitable constraints such as $\gamma \le \gamma_{\max}$).  We call this problem \emph{group testing with selectable thresholds} (GT-ST). 

From an application viewpoint, this setting is potentially relevant to scenarios where the testing device has a ``sensitivity parameter'' that dictates how many defectives are required to produce a positive outcome, though we are starting with an idealized theoretical model in which the threshold can be specified exactly without noise.  From a theoretical viewpoint, GT-ST is a natural generalization of TGT, and its study provides insight on \emph{what kinds of 1-bit test outcomes are useful}, and how the flexibility of selectable thresholds can simplify algorithm design.

We proceed to describe the problem formally, and then outline the related work and our contributions.

\subsection{Problem Setup} \label{sec:setup}

There are $n$ items, among which some subset $S \subseteq[n] \coloneqq \{1,\dotsc,n\}$ of size $|S|=k$ is \emph{defective}.  We focus on the widely-studied sublinear sparsity regime in which $k = \Theta(n^{\theta})$ for some $\theta \in (0,1)$. The number of tests is denoted by $T$, and the $t$-th test is specified by an indicator vector $x_t\in\{0,1\}^n$, where an entry of 1 indicates an item being included.  The test outcome (described below) is denoted by $y_t \in \{0,1\}$, and we collect the test vectors into a matrix $\bX \in \{0,1\}^{T \times n}$ and the test outcomes into a vector $\by \in \{0,1\}^T$.  We let $\psi_t=\sum_{i\in S}x_{ti}$ denote the number of defectives in test $t$.

In standard noiseless group testing (GT), we observe $y_t=\mathbbm{1}\{\psi_t\ge 1\}$.  In noiseless threshold group testing (TGT), this generalizes to $y_t=\mathbbm{1}\{\psi_t\ge \gamma\}$ for some fixed threshold $\gamma$.  In this paper, we are interested in the following variant:
\begin{itemize}
    \item Each test is specified not only by which items are included (i.e.,  $x_t\in\{0,1\}^n$), but also by a \emph{choice} of a positive integer threshold $\gamma^{(t)}$.  Unless stated otherwise, we assume that the threshold is subject to a constraint of the form $\gamma^{(t)} \le \gamma_{\max}$ for some maximum value $\gamma_{\max}$.
    \item The corresponding test result is $y_t=\mathbbm{1}\{\psi_t\ge \gamma^{(t)}\}$.
\end{itemize}
We are interested in conditions under which this added flexibility can help in designing algorithms, e.g., by utilizing known group testing methods (with $\gamma^{(t)} = 1$) and combining them with higher-threshold tests.

In an \emph{adaptive} testing strategy, each test (including its choice of threshold) may be designed based on previous outcomes, whereas in a \emph{non-adaptive} testing strategy, all tests must be chosen in advance.  For reasons we will mention shortly, we will focus entirely on non-adaptive strategies.  Given $(\bX,\by)$, a decoding algorithm produces an estimate $\hat{S}$ of $S$, and the error probability is given by 
\begin{equation}
    P_e \coloneqq  \Pr(\hat{S} \ne S),
\end{equation}
where the randomness is over a uniformly random defective set $S$ of size $k$, as well as any randomness used in the test design or decoding rule.

Before proceeding, we make a simple observation that the well-known \emph{counting bound} remains true in our setting, as we formally state in the following.

\begin{theorem} \label{thm:counting}
    {\em (Counting Bound)} 
    In the preceding GT-ST model, if the number of tests satisfies $T \le (1-\eta)\log_2 {n \choose k}$ for some fixed $\eta > 0$, then it holds that $P_e \to 1$ as $n \to \infty$ (with $k = \Theta(n^{\theta})$ for some $\theta \in (0,1)$).
\end{theorem}
\begin{proof}
    This result is well-known for the standard GT model, and the proof in \cite{baldassini2013capacity} only makes use of the uniform prior on $S$ and the fact that the test results are binary.  Hence, an identical proof holds in our setting.
\end{proof}

Motivated by this result and the extensive developments in standard GT \cite{aldridge2019group}, we adopt the convention of measuring the number of tests using the \emph{rate}: A sequence of strategies (indexed by $n$) is said to achieve rate $R$ if it holds that $P_e \to 0$ and
\begin{equation}
    \liminf_{n \to \infty} \frac{\log_2 {n \choose k}}{T} \ge R.
\end{equation}
In other words, the number of tests is at most $\frac{1}{R} \log_2 {n \choose k}$ up to lower-order asymptotic terms (or equivalently $\frac{1}{R} k \log_2 \frac{n}{k}$ since $k = o(n)$), and $R$ represents the amount of information learned per test.  The counting bound reveals that the rate can never exceed 1.

In the rest of the paper, we focus entirely on non-adaptive strategies, in which all tests must be decided in advance.  For adaptive strategies, there is little to be gained, since well-known standard group testing methods asymptotically match the counting bound, and they can still be used in our setting as long as the threshold $\gamma^{(t)} = 1$ is allowed.  We highlight the following two test designs from the standard GT literature, which we will use in certain subroutines in our own algorithms.

\begin{definition} \label{def:Bernoulli}
    (Bernoulli Design)
    Under the {\em Bernoulli test design} with parameter $\nu > 0$, each item is placed in each test independently with probability $\frac{\nu}{k}$.
\end{definition}

\begin{definition} \label{def:NCC}
    (Near-Constant Column Weight Design) 
    Under the {\em near-constant column weight} (NCC) design  with parameter $\nu > 0$, also known as the \emph{near-constant tests-per-item design}, each item is independently placed in $L = \frac{\nu T}{k}$ tests\footnote{The choice of rounding here is inconsequential, so is omitted.} chosen uniformly at random with replacement.  (If some test gets chosen multiple times by some item, then it still only gets placed once, resulting in a column weight strictly less than $L$.)
\end{definition}

{\bf Notation.} The function $\log(\cdot)$ has base $e$, and we use the standard asymptotic notation $O(\cdot)$, $o(\cdot)$, $\Omega(\cdot)$, $\omega(\cdot)$, and $\Theta(\cdot)$, often with the implicit understanding of the implied constants being non-negative.  For two real-valued random variables $U,V$ we use the notation $U \sd V$ to mean that $U$ is \emph{stochastically dominated} by $V$, i.e., $\Pr(U \ge x) \le \Pr(V \ge x)$ for all $x \in \mathbb{R}$ (and thus $V$ is ``larger'' than $U$ in a probabilistic sense).

\subsection{Related work} \label{sec:related_work}

{\bf Standard group testing.} The information-theoretic and algorithmic limits of (noiseless) standard GT are now very well-understood; we provide a brief summary here, and refer the reader to \cite{aldridge2019group} for a detailed overview.  Some of the main milestones included rates of simple algorithms \cite{aldridge2014group}, information-theoretic achievability \cite{scarlett2016phase}, improved rates via near-constant column weight designs \cite{johnson2018performance,coja2020information}, and finally, optimal rates \cite{coja2020optimal} (along with a polynomial-time strategy based on spatial coupling).  We highlight two prominent standard GT algorithms (e.g., see \cite[Ch.~2]{aldridge2019group}) that we will use as subroutines:
\begin{itemize}
    \item The \emph{Combinatorial Orthogonal Matching Pursuit} (COMP) algorithm declares all items in negative tests as non-defective, and the rest as defective.
    \item The \emph{Definite Defectives} (DD) algorithm first forms the \emph{possibly defective} (PD) set as the set of items appearing in no negative tests.  It then declares an item as defective if it is PD and lies in at least one test with no other PD items, and declares the item as non-defective otherwise.
\end{itemize}

The optimal rates of standard GT with \emph{approximate recovery} were also established in \cite{scarlett2016phase,scarlett2017little}, with a two-sided error guarantee allowing both $o(k)$ false positives and $o(k)$ false negatives.  More relevant to our work is the recent work \cite{mcmorrow2026optimal} studying \emph{one-sided criteria}.  Specifically, they established the optimal rates in the case of either $o(k)$ false positives (SUPERSET) \emph{or} $o(k)$ false negatives (SUBSET), but not both.  We will use their results for the SUBSET problem as a stepping stone towards some of our GT-ST results.

{\bf Threshold group testing.} An early mathematical study of threshold group testing (TGT) was given by Damaschke~\cite{damaschke2006threshold}.  The model therein was more general than that of Section \ref{sec:setup} in that it allowed a \emph{gap}: For some parameters $(\ell,u)$, the test result is always $0$ with up to $\ell$ defectives, always $1$ with at least $u$ defectives, and is arbitrary otherwise.  Exact recovery turns out to be impossible when the ``gap'' $g = u-\ell-1$ is positive, but one can still attain recovery to within $g$ false positives and $g$ false negatives.

Follow-up works on TGT include \cite{ahlswede2011bounds,cheraghchi2013improved,bui2021improved} for the setting with a gap, and \cite{de2017subquadratic,bui2019efficiently,bui2024efficient} for the gap-free setting.  All of the preceding works focus on \emph{zero-error} (combinatorial) recovery, meaning that recovery must be deterministically guaranteed for all defective sets of a given size.  Moreover, their focus is on scaling laws rather than precise constants.  Thus, their settings and results are not comparable to  ours, and we avoid going into significant detail.

Our focus is on \emph{small-error} (probabilistic) recovery, allowing $o(1)$ error probability with respect to a random defective set and/or a random test design.  Threshold GT has historically received less attention in such settings, with a notable exception being the work of Laarhoven \cite{laarhoven2015asymptotics} providing exact threshold in the very sparse regime $k = O(1)$ (which is known to be significantly simpler than the general sparse regime $k = \Theta(n^{\theta})$).  However, two highly relevant papers very recently arose in concurrent work, which we now discuss in detail.

\textbf{Comparison to concurrent work.} In the late stages of our work, two closely related papers on (noiseless) threshold group testing appeared:  The work of van der Hofstad \emph{et al.}~\cite{van2026tgt} establishes the exact information-theoretic limits of the NCC design (subject to a certain analytic conjecture holding for $\gamma \ge 3$), and the work of Coja-Oghlan \emph{et al.}~\cite{coja2026tgt} establishes that the same rate can be achieved with polynomial-time decoding under a spatially coupled design (without needing the analytic conjecture).  Thus, these works on TGT provide precise counterparts to standard GT results from \cite{coja2020information} and \cite{coja2020optimal} respectively.  While there is some clear overlap with our work, we highlight the following key aspects in which our work differs:
\begin{itemize}
    \item We introduce the \emph{selectable threshold} setting, whereas \cite{van2026tgt,coja2026tgt} focus on the case of a fixed non-selectable threshold.
    \item Within our setting, we propose an approach that is conceptually simple, easy to analyze (though sometimes relying on existing results), and can be computationally efficient.  The idea is to identify most of the defectives using standard GT tests, and then resolve the remaining uncertainty via tests that use the highest threshold.  While the resulting rates are not as precise as \cite{van2026tgt,coja2026tgt}, their simplicity and connections to the more well-studied standard GT problem make them potentially desirable.
    \item We establish a converse that matches that of \cite{van2026tgt} but holds in greater generality.   Specifically, our converse holds even with selectable thresholds, and more importantly, even within a class of \emph{deterministic designs}, rather than only the random NCC design of Definition \ref{def:NCC}.  (We show that under the NCC design, our deterministic conditions hold with high probability.)  This level of generality also leads to largely distinct proof details, though still following the same high-level idea of identifying ``uninformative'' tests.
\end{itemize}

\subsection{Contributions}

Our main contributions are outlined as follows:
\begin{itemize}
    \item We introduce the GT-ST problem, as outlined above.
    \item We provide various algorithms based on first finding a near-complete subset of defectives using standard group testing queries and then resolving the remaining by additional tests that use $\gamma^{(t)} = \gamma_{\max}$:
    \begin{itemize}
        \item When the first step achieves the optimal threshold for the SUBSET problem \cite{mcmorrow2026optimal} (albeit only currently proved for a computationally inefficient method), we establish an achievable rate that approaches 1 as $\gamma_{\max}$ increases (with $O\big(\frac{1}{\sqrt{\gamma_{\max}}}\big)$ convergence speed), thus approaching the counting bound.
        \item When the first step instead uses a combination of the simple and computationally efficient COMP and DD algorithms \cite[Ch. 2]{aldridge2019group}, we obtain an achievable rate for the finite $\gamma_{\max}$ setting, and observe that it can be strictly higher than the case $\gamma_{\max} = 1$ (corresponding to standard group testing).
    \end{itemize}
    \item Beyond the simple counting bound, we establish a refined converse for the finite $\gamma_{\max}$ under suitable regularity conditions on the test matrix.  Among other things, this establishes that our achievable number of tests is tight under these regularity conditions in the dense limit $\theta \to 1$, albeit with gaps remaining for $\theta < 1$.  As noted above, this converse generalizes the converse from concurrent work \cite{van2026tgt}, which was proved therein for TGT under the random NCC design.
\end{itemize}
For the fixed-$\gamma_{\max}$ regime, an example rate plot (with $\gamma_{\max} = 10$) is shown in Figure \ref{fig:Gamma10}, and demonstrates two notable features: (i) As $\theta \to 1$, the achievability and converse match in the sense of having a matching slope (viewed differently, the bounds match asymptotically in terms of the ratio $\frac{\#{\rm tests}}{k \log k}$, whereas the ``inverse rate'' $\frac{\#{\rm tests}}{k \log \frac{n}{k}}$ approaches $\infty$), and (ii) For $\theta$ sufficiently close to 1, the achievability curve is strictly higher than the optimal threshold for standard group testing, thus showing that (selectable) threshold queries are strictly more powerful.

\begin{figure*}[!t] 
	{\centering \includegraphics[width=0.45\textwidth]{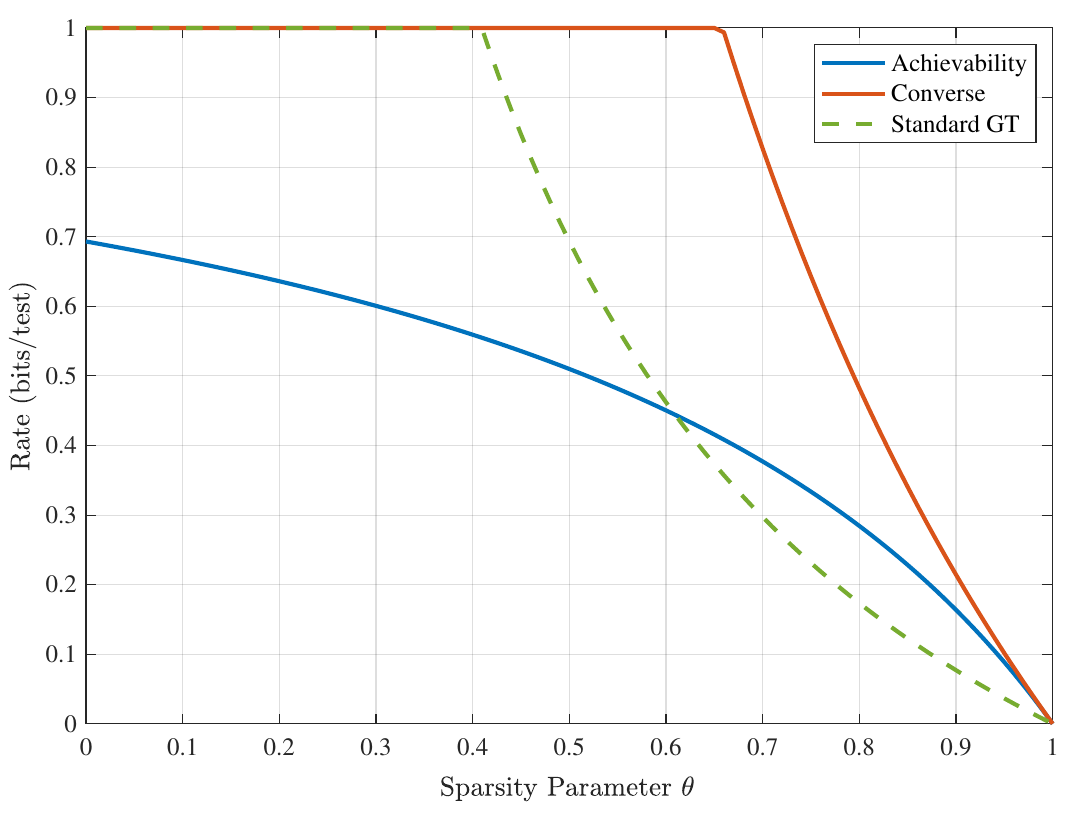} \par}
	\caption{Comparison of rates with $\gamma_{\max} = 10$.  The achievability curve is from Theorem \ref{thm:fixed}, the converse from Theorem \ref{thm: multi} (combined with Lemma~\ref{cor: bestdenom}), and the standard GT curve from \cite{coja2020optimal}.} \label{fig:Gamma10}
\end{figure*}

\section{Achievability Results} \label{sec:achievability}

\subsection{Warm-Up: Stronger Query Models}

In this section, we discuss more general 1-bit models, as they will provide useful insights that we can build on.  We broadly refer to a ``query'' as any yes/no question about a given test, e.g., our model in Section \ref{sec:setup} is based on threshold queries of the form ``are there at least $\gamma^{(t)}$ defectives in the test?''.

\subsubsection{General 1-bit Queries}

We first note that the counting bound (Theorem \ref{thm:counting}) applies for arbitrary 1-bit queries; any group testing strategy whose error probability is bounded away from one as $n \to \infty$ must have a number of tests $T$ satisfying $T \ge (1-\eta) \log_2 {n \choose k}$, where $\eta > 0$ is arbitrarily small. 

General 1-bit queries, in which $y_t$ is an arbitrary 1-bit function of $\{ i \,:\, i \in S, x_{ti} = 1 \}$ are ``too powerful'' in the sense that they make the problem trivial to solve.  Specifically, we can let the test include all items, in which case the preceding set becomes $S$ itself, and then we can trivially ask queries of the following form: ``Represent each $S$ by a length-$\log_2 {n \choose k}$ binary vector: Is the first bit 1?  Is the second bit 1? (etc.)''.  This trivially succeeds using $\log_2 {n \choose k}$ 1-bit queries, and matches the counting bound.

To reduce the power of the 1-bit queries, a natural starting point is to require that they return a function of the number of defectives in the test (i.e., $\psi_t$).  Group testing models satisfying this property are termed ``only defects matter'' in \cite{aldridge2019group}.

\subsubsection{Two-Stage Template via \textsc{Subset} and Quantitative Queries} \label{sec:qgt}

In this subsection, we momentarily suppose that we have access to a limited number of \emph{quantitative} queries that return $\psi_t$ itself, thus violating the condition of 1-bit queries.  After this subsection, we will only consider 1-bit queries.

We propose a strategy that mostly makes standard (1-bit) group testing queries, and then makes a small number of quantitative queries.  Accordingly, the tests are performed in two batches:
\begin{itemize}
    \item \textbf{Batch~1 (\textsc{Subset}).} In this batch, we perform only standard group testing queries.  We use these to run a non-adaptive routine \textsc{Subset} that returns $\widehat{S}_1 \subseteq S$ of size $|\widehat{S}_1| \ge k-k/\log^4 k$ (i.e., it misses at most $k/\log^4 k$ defectives).  This can be done with a number of tests meeting the counting bound \cite[Thm. 2]{mcmorrow2026optimal}.
    \item \textbf{Batch~2 (Reduced GT).} Design tests for a standard group testing problem with $k/\log^4 k$ defectives, but instead of observing the standard results, observe the quantitative results ($\psi_t$).  Subtract the contribution of items in $\widehat{S}_1$ from each measurement, then change any non-zero result to $1$ (i.e., an ``OR'' operation), so that the outcomes correspond to standard GT outcomes on the reduced defective set $S \setminus \widehat{S}_1$.  Run a standard GT decoder such as \textsc{COMP} or \textsc{DD} (see Section \ref{sec:related_work} for their descriptions) to recover the defectives that were not identified in batch 1.
\end{itemize}
Batch~2 then needs only $O\big((k/\log^4 k)\log n\big) = o(k \log \frac{n}{k})$ tests, and the overall rate is $1$.  However, we would like to avoid the use of quantitative queries; we now give a simple variation for this purpose.

\subsubsection{Adapting to Generalized 1-bit Queries}

Observe that in batch 2 above, the test outcomes lie in $\{0,1,\dotsc,k\}$.  As a result, we can trivially adapt the above strategy to only use 1-bit queries by repeating each of batch 2's test placements $\log_2(k+1)$ times, but with each repetition using a different choice of 1-bit query, namely, one for each of the $\log_2(k+1)$ bits in the binary expansion of $\psi_t$ (or similarly with suitable rounding to an integer).  This only increases the number of tests in batch 2 by an $O(\log k)$ factor to a total of $O\big((k/\log^3 k)\log n\big)$, and all tests are now 1-bit.  Observe that the number of tests in batch 1 dominates, and so this strategy achieves a rate of $1$.

\subsubsection{Adapting to Unbounded Threshold Queries} \label{sec:adapted_threshold_queries}

To conclude this ``warm-up'' subsection, we adapt the above strategy to the case that only threshold queries are allowed but there are no constraints on the selectable threshold (i.e., $\gamma_{\max} = \infty$, which is equivalent to $\gamma_{\max} = k$).  We outline the relevant analysis below and defer the full details to Appendix \ref{app:unbounded_proof}.

A naive strategy would be to learn $\psi_t$ by testing all thresholds $\gamma \in \{1,2,\dotsc,k\}$, but this would require $k$ queries to learn each value of $\psi_t$.  This can be reduced by noticing that the number of defectives (including those in $\widehat{S}_1$) in each test of batch 2 is $\Theta\big( \log^4 k \big)$ with high probability; however, testing all thresholds $\gamma \in \{1,2,\dotsc,\Theta(\log^4 k)\}$ would still lead to $O(k \log n)$ tests in batch 2.  The solution is to further narrow the required range using a concentration argument.

Specifically, consider a Bernoulli test design with  inclusion probability $p=\frac{\log^4 k}{k}$ so that $ \mu_\psi \coloneqq \mathbb{E}[\psi_t]=\log^4 k$. By the multiplicative Chernoff bound and a union bound over all tests in the second batch, it can be shown that the deviation of $\psi_t$ from its mean $\mu_\psi$ is at most $\Theta(\log^{5/2}k)$ for all $t$ with probability $1-o(1)$.
Hence one can emulate a quantitative result with only $O(\log^{5/2} k)$ threshold queries per test, so the total scales as
\begin{equation}
    O\left(\frac{k}{\log^4 k}\log n\right)\times O(\log^{5/2} k)=o(k\log n).
\end{equation}
Thus, batch 1 again dominates, and we meet the counting bound.  Formally, we have the following.

\begin{theorem} \label{thm:unbounded}
    {\em (Unconstrained Threshold Queries)}
    In the GT-ST problem with no constraints on the selectable threshold (i.e., $\gamma_{\max} = \infty$), there exists a choice of non-adaptive test design $\bX$ and decoder $\widehat{S}$ satisfying $P_e \to 0$ with a number of tests satisfying $T = \big( k \log_2\frac{n}{k} \big)(1+o(1))$, i.e., a rate of 1 bit/test is achievable.
\end{theorem}
See Appendix \ref{app:unbounded_proof} for a detailed proof of this result.

\subsection{The Large $\gamma_{\max}$ Limit}

In this subsection, we seek to sharpen the understanding from Theorem \ref{thm:unbounded} by studying the case that $\gamma_{\max}$ is ``bounded but large''.  Specifically, we seek to understand how quickly the rate approaches 1 as $\gamma_{\max} \to \infty$, where this limit is taken \emph{after} having already taken $n \to \infty$ with $k = \Theta(n^{\theta})$.  We outline the relevant analysis below and defer the full details to Appendix \ref{app:large_proof}.

We adopt a similar strategy to the previous subsection, but with modified parameters so that $\psi_t$ in batch 2 concentrates around $\gamma_{\max} / 2$ rather than around $\log^4 k$.  
In batch 1, we again apply \textsc{Subset} to identify $k-k/\log^4 k$ defectives.  In batch 2, we apply Bernoulli testing with inclusion probability $p=\frac{\gamma_{\max}}{2k}$ and keep only those tests whose number of defectives lies in the interval $(\gamma^-, \gamma^+) \coloneqq(\frac{\gamma_{\max}}{2}-2\sqrt{\gamma_{\max}},\,\frac{\gamma_{\max}}{2}+2\sqrt{\gamma_{\max}})$ (see below for how this is detected). Similar to the case with unbounded $\gamma_{\max}$, we repeat each test $O(\sqrt{\gamma_{\max}})$ times by initially using threshold $\lfloor \gamma^-\rfloor $, and repeating the test with incrementing thresholds until reaching $\lceil \gamma^+ \rceil$. This allows us to learn $\psi_t$ if it lies within $(\gamma^-, \gamma^+)$, or learn that $\psi_t$ lies outside of the interval, in which case it is discarded. The multiplicative Chernoff bound (Lemma \ref{lem:chernoff}) implies that the probability of $\psi_t$ falling outside of this interval is upper bounded by an absolute constant.
Thus, the window captures a constant fraction of tests, so discarding the rest loses only a constant factor (with high probability). 

Observe that the ``residual'' problem (with the \textsc{Subset} defectives removed) has up to $\frac{k}{\log^4 k}$ unknown defectives and inclusion probability $p = \frac{\gamma_{\max}}{2k}$. Using Bernoulli's inequality, the probability of a specific test containing none of these defectives (which we term being \emph{residually negative}) is lower bounded by
\begin{equation}
    \left(1-\frac{\gamma_{\max}}{2k}\right)^{k/\log^4 k} \geq  1-\frac{\gamma_{\max}}{2\log^4 k} = 1-o(1).
\end{equation}
Now let us momentarily assume that the remaining residually negative tests are quantitative (i.e., $\psi_t$ is observed), and that there are $T_2'$ such tests (among those that were not discarded). Then, we can apply the same steps as in Section \ref{sec:qgt} (i.e., subtract the contribution from $\widehat{S}_1$, change non-zero tests to $1$) to convert it to a classical group testing problem. We can then recover the remaining defective items from these ``converted'' tests using these residually negative tests. A fixed non-defective fails to appear in any kept residually negative test with probability $\big(1 - \frac{\gamma_{\max}}{2k}\big)^{T_2'} \leq \exp\big(-\frac{\gamma_{\max}}{2k}T_2'\big)$. Thus, if $T_2'=\Theta\big(\frac{k\log n}{\gamma_{\max}}\big)$ with a large enough implied constant, then every non-defective will appear in at least one such negative test with probability $1-o(1)$, implying successful recovery.

Each ``quantitative test'' corresponds to evaluating $O(\sqrt{\gamma_{\max}})$ threshold queries (i.e., the above-mentioned interval length), and hence, the overall number of 1-bit queries (to non-discarded, residually negative tests) sufficient for accurate recovery is
\begin{equation}
    T''_2=\Theta\left(\frac{k\log n}{\sqrt{\gamma_{\max}}}\right). \label{eq:T2_gamma_max}
\end{equation}
To ensure that there are $T''_2$ non-discarded and residually negative 1-bit tests, we set the actual number of conducted tests to be
\begin{equation}
    \label{eq:actual_T} T_2 = \Theta(T''_2) = \Theta\left(\frac{k\log n}{\sqrt{\gamma_{\max}}}\right).
\end{equation}
Since a constant fraction of $T_2$ will be discarded, and a further $\gamma_{\max}/(2\log^4 k) = o(1)$ fraction will not be residually negative, we can set the implied constant in \eqref{eq:actual_T} sufficiently large to attain \eqref{eq:T2_gamma_max} with high probability (by binomial concentration).  Adding the $\big(k\log_2\frac{n}{k}\big)(1+o(1))$ tests from batch 1, we see that the rate attained is $\big(1+O(\gamma_{\max}^{-1/2})\big)^{-1}=1-O(\gamma_{\max}^{-1/2})$.  Formally, we have the following.

\begin{theorem} \label{thm:large}
    {\em (The Large-$\gamma_{\max}$ Limit)}
    For the GT-ST problem with maximum threshold $\gamma_{\max}$, there exists a choice of non-adaptive test design and decoder satisfying $P_e \to 0$ with a number of tests satisfying the following as $n \to \infty$:
    \begin{equation}
        T = \bigg( k \log_2\frac{n}{k} \bigg)(1+o(1)) + O\bigg( \frac{k\log n}{\sqrt{\gamma_{\max}}} \bigg), \label{eq:T_large}
    \end{equation}
    Hence, the rate $1 - O\big( \frac{1}{\sqrt{\gamma_{\max}}} \big)$ is achieved as $\gamma_{\max} \to \infty$, where the limit is after already having taken $n \to \infty$ with $k = \Theta(n^{\theta})$ in accordance with the definition of an achievable rate.
    
\end{theorem}
The detailed proof of this result can be found in Appendix \ref{app:large_proof}.

\subsection{The Finite $\gamma_{\max}$ Regime} \label{sec:finite_gamma}

We now look at attaining a closer understanding when $\gamma_{\max}$ is a fixed constant, starting with $\gamma_{\max} = 2$.  We outline the relevant analysis below and defer the full details to Appendix \ref{app:fixed_proof}.

\subsubsection{The Case $\gamma_{\max}=2$}

We start with the case $\gamma_{\max} = 2$, and again adopt a two-batch approach.  This means that each test result can reveal one of four things: $\psi_t = 0$, $\psi_t \ge 1$, $\psi_t \in \{0,1\}$, or $\psi_t \ge 2$ (only the former two if $\gamma = 1$, and only the latter two if $\gamma = 2$). 

Different from before, we only use one threshold in each batch: $\gamma = 1$ (i.e., standard GT) in batch 1, and $\gamma = 2$ in batch 2. Throughout, we will adopt separate NCC designs (Definition \ref{def:NCC}) in each batch; in the first batch we set $L_1 = \frac{T_1 \log 2}{k}$, and in the second we set $L_2 = \frac{T_2\nu}{k}$ for some parameter $\nu > 0$ that can be optimized.

\noindent \textbf{Batch~1:} In the first batch, we use the threshold $\gamma^{(t)} = 1$ and adopt a design achieving rate $\log 2$, so that $T_1 = (\frac{1}{\log2} k\log_2(n/k)) (1 + o(1))$. By applying the approximate recovery versions of \textsc{COMP} and \textsc{DD} decoding, we can attain (with $1-o(1)$ probability) a subset $S_{\mathrm{DD}}\subseteq S$ of size $k(1-o(1))$ and a superset $S_{\mathrm{COMP}}\supseteq S$ of size $k(1+ o(1))$ \cite[App. D]{mcmorrow2026optimal}.

\noindent\textbf{Batch~2:} We use the threshold $\gamma^{(t)} = 2$. For each $i\in S_{\mathrm{COMP}}\setminus S_{\mathrm{DD}}$, we say that a test \emph{resolves $i$} if it contains $i$, exactly one element of $S_{\mathrm{DD}}$, and no other element of $S_{\mathrm{COMP}}\setminus S_{\mathrm{DD}}$. Such a test definitively classifies $i$ from its outcome. Thus, to successfully recover the remaining defectives in $S \setminus S_{\rm DD}$, it suffices to resolve the status of all items in $S_{\rm COMP} \setminus S_{\rm DD}$. For a given test, the choice of test design implies that probability of exactly 1 item from $S_{\rm DD}$ being included in a given test is $\nu e^{-\nu}(1+o(1))$. Thus, the expected number of such tests is $T_2 \, \nu e^{-\nu}(1+o(1))$. It can then be shown that the actual number of such tests is tightly concentrated around its mean with high probability due to McDiarmid's inequality (Lemma \ref{lem:mcdiarmid}).

Conditioned on this concentration event, and fixing an item $i \in S_{\rm COMP} \setminus S_{\rm DD}$, it follows that each test drawn for $i$ to be included in contains exactly one item from $S_{\rm DD}$ with probability at least $\nu e^{-\nu}(1-o(1))$. Thus, the probability that it is not resolved by a test in \emph{any} of the $L_2$ draws with replacement is at most $(1 - \nu e^{-\nu}(1-o(1)))^{L_2}$, which can be expressed as $\exp\big(\frac{T_2 \nu}{k} \log(1 - \nu e^{-\nu})(1-o(1))\big)$. 
    Following a union bound over at most $k$ uncertain items, it follows that $T_2 =  \frac{k\log k}{\nu\,\log \frac{1}{1-\nu e^{-\nu}}}(1+o(1))$ tests suffice to ensure every $i \in S_{\rm COMP} \setminus S_{\rm DD}$ is resolved with high probability. Numerically optimizing $\nu$ gives a denominator of roughly $0.636$, improving on $(\ln 2)^2\approx 0.48$ (e.g., appearing in the analysis of DD for standard GT in \cite{johnson2018performance}). Summing the contribution from batch 1 and batch 2, we obtain the following (see Theorem \ref{thm:fixed} for a general and formal version):
\begin{equation}
    T \approx \frac{1}{\log2} k\log_2(n/k) + \frac{k\log k}{0.636}. \label{eq:T_NCC_gamma2}
\end{equation}
Writing $\log k=\frac{\theta \log 2}{1-\theta}\log_2\frac{n}{k}$, the normalized rate becomes roughly $\big(\frac{1}{\log 2} +\frac{\theta\log 2}{0.636(1-\theta)}\big)^{-1}$. Note that the second term dominates for $\theta$ close to one, which will be important when comparing to the converse in Section \ref{sec:converse}.

\subsubsection{Extension to general $\gamma_{\max} = r$}

For the general case, we adopt the same two-batch approach, alongside the same test designs.  Batch 1 proceeds in exactly the same manner as before, with the only changes arising in the second batch. 

In the second batch, each item $i \in S_{\rm COMP} \setminus S_{\rm DD}$ is now resolved in a test if it contains exactly $r-1$ items from $S_{\rm DD}$, and $i$ is the only item included from $S_{\rm COMP} \setminus S_{\rm DD}$. The probability that a test contains exactly $r-1$ items from $S_{\mathrm{DD}}$ can be determined via a Poisson approximation to the binomial (along with $|S_{\rm DD}| = k(1-o(1))$): 
\begin{equation}
    \label{eq:S_DD_poisson} \Pr\left({\rm Bin}\left(\lvert S_{\rm DD}\rvert, \frac{\nu}{k}(1+o(1))\right) = r-1\right) = \frac{\nu^{r-1}e^{-\nu}}{(r-1)!}(1+o(1)).
\end{equation}
As in the previous case, the number of tests that contain exactly $r-1$ items from $S_{\rm DD}$ is tightly concentrated around its mean $T_2 \,\nu^{r-1}e^{-\nu}/(r-1)!$ with high probability due to McDiarmid's inequality. Conditioned on this, each uncertain item has $L_2=\frac{T_2 \nu}{k}$ placements that each ``succeed'' with probability $\nu^{r-1}e^{-\nu}/(r-1)! \cdot (1+o(1))$, meaning the ``overall failure'' probability is given by $\exp\big(\tfrac{T_2\nu}{k}\log(1-\nu^{r-1}e^{-\nu}/(r-1)!)(1-o(1))\big)$. A union bound then implies that
\begin{equation}
    T_2 = \frac{k\log k}{\nu\,\log\Big(\frac{1}{1-\nu^{r-1}e^{-\nu}/(r-1)!}\Big)}(1+o(1))
\end{equation}
tests suffice for all uncertain items to be resolved with high probability. 

The above findings are summarized as follows.

\begin{theorem} \label{thm:fixed}
    {\em (The Fixed-$\gamma_{\max}$ Regime)} For the GT-ST problem with maximum threshold $\gamma_{\max}$, there exists a choice of non-adaptive test design and a polynomial-time decoder satisfying $P_e \to 0$ with a number of tests satisfying
    \begin{equation}
        T = \bigg( \frac{1}{\log 2} k \log_2\frac{n}{k} + \min_{\nu > 0}\frac{k\log k}{\nu\,\log\Big(\frac{1}{1-\nu^{\gamma_{\max}-1}e^{-\nu}/(\gamma_{\max}-1)!}\Big)} \bigg)(1+o(1)).
    \end{equation}
In particular, for $\gamma_{\max} = 2$, this becomes
    \begin{equation}
        T = \bigg( \frac{1}{\log 2} k \log_2\frac{n}{k} + \min_{\nu > 0}\frac{k\log k}{\nu\,\log\Big(\frac{1}{1-\nu e^{-\nu}}\Big)} \bigg)(1+o(1)).
    \end{equation}
\end{theorem}
The detailed proof of this result can be found in Appendix \ref{app:fixed_proof}.

\section{Converse Results} \label{sec:converse}

\subsection{Useful Definitions}

We introduce some useful notation in addition to that of Section \ref{sec:setup}.  For each test $t \in [T]$, we let $w_t+1$ denote its \emph{weight},\footnote{We define the number of items in a test as $w_t + 1$ purely for notational convenience in the proof of Theorem \ref{thm: sing}, as it avoids the extensive inclusion of $-1$ terms if $w_t$ were the weight.} i.e., the number of items included in test $t$, and for each item $i \in [n]$, we denote by $d_i$ its \emph{degree}, i.e., the number of tests in which item $i$ participates.  We let $\Gamma$ denote the set of thresholds that the algorithm is allowed to choose; in Section \ref{sec:setup} we adopted $\Gamma = \{1,2,\dotsc,\gamma_{\max}\}$, but we will consider the more general case of an \emph{arbitrary set of thresholds}:
\begin{equation}
\Gamma = \{\gamma_1, \dots, \gamma_\ell\}, \label{eq:Gamma}
\end{equation}
where $\ell$ is the number of unique thresholds.  
We do this because this generalization comes ``for free''.

As with converse bounds for standard group testing \cite{coja2020information,bay2022optimal}, we will consider two prior models for the defective set $S \subset [n]$:
\begin{itemize}
    \item All of our main results will be stated for the \emph{combinatorial prior} in which $S$ is uniform on the $n \choose k$ subsets of $[n]$ of size $k$, as we assumed in Section \ref{sec:setup}. 
    \item As a stepping stone in the analysis, we will also consider the \emph{i.i.d.~prior}, in which each item is independently defective with some probability $q$.
\end{itemize}



Lastly, we introduce a notion of items being \emph{masked} (or \emph{disguised}) that generalizes such notions from standard GT  \cite{aldridge2018individual,coja2020optimal,bay2022optimal}. Let $i \in [n]$ and let $t \in [T]$ be a test with threshold $\gamma^{(t)}$.  We say that the test $t$ is \emph{informative} for item $i$ if $i$ participates in $t$ and, among the remaining items in the test (excluding $i$), there are exactly $\gamma^{(t)} - 1$ defectives.  This is significant because it means that the outcome of test $t$ changes when $i$ is changed from defective to non-defective or vice versa. An item $i$ is said to be \emph{masked} if there are no tests that are informative for $i$.  This is a random event due to the randomness in $S$, and we let $D_i$ denote the event that $i$ is masked (or $i$ is \emph{disguised}, hence the symbol $D$).

\subsection{The Single-Threshold Setting}

We begin with the simple setting in which the threshold set contains only a single value, i.e., $\Gamma = \{\gamma\}$.  This recovers the TGT problem, and it captures most of the ideas for the GT-ST problem while simplifying notation and proofs.

\subsubsection{Discussion: Failure of Standard GT Techniques}

It is tempting to seek a converse via the same techniques as standard GT \cite{coja2020optimal}, but (selectable) threshold group testing faces major technical barriers when doing so, which we now proceed to outline.  

If we momentarily consider an i.i.d.~defectivity prior in which each item is defective independently with some probability $q$ (as is done in \cite{coja2020optimal}), then a simple calculation reveals that for a given test $t$ containing $i$, the probability of the test being uninformative for $i$ is $\Phi_{t,i} \coloneqq 1 - \Pr(\mathrm{Bin}(w_t,q)=\gamma-1)$, where $w_t = \sum_{j : j \ne i} x_{tj}$ is the number of items in test $t$ excluding $i$.  Moreover, the binomial probability can readily be characterized tightly using a Poisson approximation (when $q = o(1)$ and $\gamma$ is fixed).

However, the major technical obstacle is that we \emph{cannot directly conclude that} $\Pr(D_i) \ge \prod_{t : x_{ti} = 1} \Phi_{t,i}$.  When $\gamma = 1$, such a claim follows from the FKG inequality and the fact that the ``uninformative'' events are \emph{increasing} -- changing any item from non-defective to defective can only expand (or keep unchanged) the set of tests that are uninformative for $i$.  For threshold GT with $\gamma \ge 2$, this is no longer the case; for instance, a given test may have contained item $i$ and a total of $\gamma - 2$ other defectives, and changing a non-defective to defective could increase that to $\gamma - 1$, thus changing the test from uninformative to informative.

Our results turn out to match what \emph{would have been obtained} if the FKG inequality were to remain valid, but we do so via a distinct approach that separates the tests whose number of defectives is ``at most $\gamma - 2$'' vs.~``at least $\gamma$''.  In order to control the dependencies between the relevant events, we turn out (at least for our proof) to require some additional assumptions, but we argue that these are reasonable, and show that they are indeed satisfied by the NCC random design.

\subsubsection{Assumptions}

The regularity assumptions that we adopt are formally stated as follows.

\begin{assumption} \label{asm:regularity}
    The number of tests scales as $T = \Theta(k \log n)$, each item $i \in [n]$ appears in $d_i = \Theta(\log n)$ tests, and each test $t \in [T]$ contains $w_t + 1 = \Theta(n/k)$ items. Moreover, the following \emph{near-constant degree and weight} conditions hold:
    \begin{align}
    &\frac{\max_{i \in [n]}d_{i}}{\min_{i \in [n]}d_{i}} = 1+o(1), \quad \frac{\max_{t \in [T]} w_t}{\min_{t \in [T]} w_t} = 1+o(1).
    \end{align}
    For convenience, we define $\zeta = \frac{\max_{i \in [n]}d_{i}}{\min_{i \in [n]}d_{i}}$ (and thus $\zeta = 1+o(1)$), and let $\Delta$ denote the constant such that such that $w_t = \frac{\Delta n}{k} (1 \pm o(1))$ (for all $t \in [T]$).
\end{assumption} 

\begin{assumption} \label{asm:sparsity}
    There exists $c = o\big( \frac{n}{k \log n} \big)$ such that any two distinct tests intersect in at most $c$ items.  Combined with Assumption \ref{asm:regularity}, this implies that 
    \begin{equation}
     \frac{c \cdot \max_{i \in [n]}d_{i}}{\min_{t \in [T]} w_{t}} =  o(1).
    \end{equation}
\end{assumption} 

Assumption \ref{asm:regularity} states that all items have the same degree to within a $1+o(1)$ factor, and all tests have the same weight to within a $1+o(1)$ factor.  The assumption $T = \Theta(k \log n)$ is without loss of generality for proving a converse, since the counting bound prevents any smaller scaling (when $k = \Theta(n^{\theta})$), and all of our achievability results match this scaling.  Given the near-constant weight condition, the assumption of row weight $\Theta(n/k)$ is also easily justified, because any other scaling would yield $o(1)$ entropy per test, thus requiring $T = \omega(k \log n)$ via Fano's inequality (e.g., see \cite{atia2012boolean}).  Moreover, given a near-constant row weight of $\Theta(n/k)$ and $T = \Theta(k \log n)$, a simple double counting argument reveals that the column weight must be $\Theta(\log n)$.

The interpretation of Assumption \ref{asm:sparsity} is less immediately obvious.  It ensures that the tests do not overlap too much, so that the information contributed by two different tests remains largely ``independent''.  (Specifically, the fraction of overlap is constrained to scale as $o\big( \frac{1}{\log n} \big)$.)  Intuitively, this would be a preferable property to have anyway, but dropping it as a formal assumption here still appears to be challenging.

In Appendix \ref{app:verify}, we show that the NCC design (Definition \ref{def:NCC}) with $\nu = \Theta(1)$ and $T = \Theta(k \log n)$ satisfies both assumptions with probability $1-o(1)$.  Moreover, the analysis shows that the above $o(1)$ terms come out to be much smaller than required (e.g., polynomially small in $n$), meaning that the assumptions hold with a ``significant margin'' for this design.  

For the Bernoulli design (Definition \ref{def:Bernoulli}), the degrees $d_i$ turn out to fail to be near-constant, since they do not concentrate sharply enough.  While this is actually the cause of the design being \emph{weaker} than NCC design (e.g., see the discussion in \cite[Sec.~2.7]{aldridge2019group}), this observation may still limit ``how mild'' Assumption \ref{asm:regularity} can be considered to be.  While our analysis can be adapted to handle $\zeta > 1$, the resulting lower bound on $T$ becomes worse as $\zeta$ increases.

\subsubsection{Statement of Converse Result}

We are now ready to state our converse result for the single-threshold setting, which we prove in Appendix \ref{app:pf_conv_single}.

\begin{theorem}\label{thm: sing}
    {\em (Converse for the Single-Threshold Setting)} 
   In the TGT problem with $k = \Theta(n^\theta)$ defectives and threshold $\gamma$, in order to have $\Pr(\hat{S}\neq S) \nrightarrow 1$, it is necessary that
   \begin{equation}
       \label{eq:sing_T_lb} T \geq \min_{\Delta > 0} \frac{k \log k}{-\displaystyle\Delta \log \left(1 - \frac{\Delta^{\gamma - 1}e^{-\Delta}}{(\gamma - 1)!}\right)}(1-o(1))
   \end{equation}
   for any test design $\bX$ satisfying Assumptions \ref{asm:regularity} and \ref{asm:sparsity}, and any decoding algorithm.
\end{theorem}

\subsubsection{Discussion}

Theorem \ref{thm: sing} matches the ``non-entropy term'' for the exact threshold of the near-constant column weight design established in concurrent work \cite{van2026tgt}.  As discussed in the introduction, our converse holds in greater generality, since it is not restricted to the random design of Definition \ref{def:NCC}, but instead holds for general deterministic designs satisfying Assumptions \ref{asm:regularity} and \ref{asm:sparsity}.  Under these assumptions, the entropy term from \cite{van2026tgt} can also readily be established via Fano's inequality (e.g., similar to \cite{atia2012boolean}), at least when considering the requirement $\Pr(\hat{S} \ne S) \nrightarrow 0$ rather than $\Pr(\hat{S} \ne S) \to 1$ (i.e., a ``weak converse'' rather than a ``strong converse'').  We omit the details, since it was observed in \cite{van2026tgt} that the tighter of the counting bound and Theorem \ref{thm: sing} is correct for all $\theta$ except a very narrow ``transition region''.  For instance, when $\gamma = 2$, this region occurs at $\theta \approx 0.4785$ and has a width of less than $0.005$; see \cite[Fig.~2]{van2026tgt}.  (The transition region for $\gamma = 10$ therein is also imperceptible in a zoomed out plot.)


\subsection{The Selectable-Threshold Setting}

We now turn to the general setting in which tests may have different thresholds, i.e., $\Gamma = \{\gamma_1, \dots, \gamma_\ell\}$, where $\ell$ is the number of distinct thresholds.  The near-constant row weight part of Assumption \ref{asm:regularity} may be overly restrictive in this setting, since it may be highly suboptimal to keep the same row weight across different thresholds.  Accordingly, we adopt a natural generalization (stated below) allowing each threshold to have a different weight (but still near-constant among all tests using that threshold).  

\subsubsection{Assumptions}

We keep Assumption \ref{asm:sparsity} unchanged, but generalize Assumption \ref{asm:regularity} to the following.

\begin{assumption} \label{asm:sub-matrix_regularity}
    For each $r \in [\ell]$, the number of tests with threshold $\gamma_r$ scales as $\Theta(k \log n)$, each item $i \in [n]$ appears in $d_{i,r} = \Theta(\log n)$ such tests, and the tests $\mathcal{T}_r \subseteq [T]$ with threshold $\gamma_r$ each contain $w_t + 1 = \Theta(n/k)$ items.  Moreover, the following \emph{near-constant degree and weight} conditions hold for all $r \in [\ell]$:
    \begin{align}
    &\frac{\max_{i \in [n]}d_{i,r}}{\min_{i \in [n]}d_{i,r}} = 1+o(1), \quad \frac{\max_{t_r \in \mathcal{T}_r} w_{t_r}}{\min_{t_r \in \mathcal{T}_r} w_{t_r}} = 1+o(1),
    \end{align}
    where $\mathcal{T}_r \subseteq [T]$ is the set of test indices with threshold $\gamma_r$.  
    For convenience, we define $\zeta_r = \frac{\max_{i \in [n]}d_{i,r}}{\min_{i \in [n]}d_{i,r}}$ (and thus $\zeta_r = 1+o(1)$), and let $\Delta_r$ denote the constant such that $w_{t_r} = \frac{\Delta_r n}{k} (1 \pm o(1))$ for all $t_r \in \mathcal{T}_r$.
\end{assumption} 

In short, this means that up to re-ordering tests, $\bX$ is the vertical concatenation of $\ell$ matrices satisfying Assumption \ref{asm:regularity}, each with a different threshold $\gamma_r$ and possibly different degrees and weights.  While not every sequence of tests designs necessarily consists of $\Theta(k \log n)$ tests for \emph{every} threshold, assuming this is convenient and very mild; intuitively, for the purpose of proving a converse, each threshold could be given $\epsilon k \log n$ additional tests ``for free'', with an arbitrarily small impact on the rate as $\epsilon$ decreases (since $\ell$ is constant).



\subsubsection{Statement of Converse Results}

Given that the test design satisfies Assumptions \ref{asm:sparsity} and \ref{asm:sub-matrix_regularity} above, we proceed to state our main converse bound for the GT-ST problem.  We initially suppose that the fraction of tests with threshold $\gamma_r$ is some fixed quantity $\alpha_r \in (0,1)$ (i.e., $\lvert \mathcal{T}_r\rvert = \alpha_r T$), for each $r \in [\ell]$, and then provide a corollary in which we optimize over $\alpha_r$.  The main theorem is proved in Appendix \ref{app:pf_conv_multi}, and the subsequent two lemmas are proved in Appendices \ref{app:pf_max_f} and \ref{app:pf_large_conv}.

\begin{theorem}\label{thm: multi}
    {\em (Converse for the Selectable-Threshold Setting)} 
    In the GT-ST problem with $k=\Theta(n^\theta)$ defectives and thresholds $\Gamma = \{\gamma_1, \dots, \gamma_\ell\}$, when tests with threshold $\gamma_r$ comprise an $\alpha_r \in (0,1)$ fraction of the $T$ tests for each $r \in [\ell]$ (up to rounding), in order to have $\Pr(\hat{S}\neq S) \nrightarrow 1$, it is necessary that
    \begin{align} \label{eq:selec-thres}
    T \geq \min_{(\Delta_1,\dots \Delta_\ell) \in \bbR_+^\ell} \frac{k \log k}{\displaystyle-\sum_{r \in [\ell]}\alpha_r \Delta_r \log \left(1 - \frac{\Delta_r^{\gamma_r - 1}e^{-\Delta_r}}{(\gamma_r - 1)!}\right)}(1-o(1))
    \end{align}
    for any test design $\bX$ satisfying Assumptions \ref{asm:sparsity} and \ref{asm:sub-matrix_regularity}, and any decoding algorithm.
\end{theorem}

We now consider the optimization over the threshold proportions $(\alpha_1, \dots, \alpha_\ell)$ and the threshold-dependent weight parameters $(\Delta_1,\dotsc,\Delta_r)$, both of which may be chosen to minimize the number of tests.  This optimization leads to the bound in Theorem~\ref{thm: sing} with $\gamma = \gamma_{\max}$, as we will discuss further below.  Note that while Theorem \ref{thm: multi} assumes $\alpha_r > 0$ for all $r$, we may safely include $\alpha_r \in \{0,1\}$ in the optimization since this simply amounts to removing one or more thresholds, i.e., shrinking the set $\Gamma$ in \eqref{eq:Gamma} (Theorem \ref{thm: multi} holds for any such set).


\begin{lemma}\label{cor: bestdenom}
    {\em (Optimization of Parameters)} 
   Let $\boldsymbol{\alpha} = (\alpha_1, \dots, \alpha_\ell)$, $\boldsymbol{\Delta} = (\Delta_1, \dots, \Delta_\ell)$, and $\boldsymbol{\gamma} = (\gamma_1, \dots, \gamma_\ell)$. Define
    \begin{equation}
    f(\boldsymbol{\alpha}, \boldsymbol{\boldsymbol{\Delta}}, \boldsymbol{\gamma})
    = \sum_{r \in [\ell]} \alpha_r \, h(\Delta_r, \gamma_r),
    \end{equation}
    where 
    \begin{equation}
        h(\Delta, \gamma) = -\Delta \log\left(1 - \dfrac{\Delta^{\gamma-1}}{(\gamma-1)!} e^{-\Delta}\right).
    \end{equation}
    For each $r \in [\ell]$, let $\Delta_r^* = \arg\max_{\Delta > 0} h(\Delta, \gamma_r)$, and define $r^* = \arg\max_{r \in [\ell]} \gamma_r$ (and thus $\gamma_{r^*} = \gamma_{\max}$). Then, the maximum value of $f(\boldsymbol{\alpha}, \boldsymbol{\Delta}, \boldsymbol{\gamma})$ over all choices of $\boldsymbol{\alpha}$ and $\boldsymbol{\Delta}$ is given by
    \begin{equation}
    \max_{\boldsymbol{\alpha}, \boldsymbol{\Delta}}
    f(\boldsymbol{\alpha}, \boldsymbol{\Delta}, \boldsymbol{\gamma})
    = h(\Delta_{r^*}^*, \gamma_{r^*}),
    \end{equation}
    where $\boldsymbol{\alpha}$ is constrained to the probability simplex, and $\boldsymbol{\Delta}$ is constrained to be entry-wise non-negative.
\end{lemma}

Lastly, we investigate the asymptotic behavior of $\max_{\Delta > 0} h(\Delta,\gamma)$ in the large-threshold limit $\gamma \to \infty$, as characterized in the following lemma. 
\begin{lemma}\label{cor: denomconv}
    {\em (The Large-$\gamma$ Limit)} 
    The following holds as $\gamma \to \infty$, uniformly over all $\Delta=\Delta(\gamma)$ satisfying $\Delta \in [\gamma-1/3,\gamma]$:
    \begin{align}
    - \Delta \log\left(1 - \frac{\Delta^{\gamma-1}}{(\gamma-1)!} e^{-\Delta}
    \right)
    = \sqrt{\frac{\gamma}{2\pi}} (1+o(1)).
    \end{align}
    Moreover, the condition $\Delta \in [\gamma - 1/3, \gamma]$ holds for the optimal $\Delta$, i.e., for $\Delta^* = \Delta^*(\gamma)$ chosen according to $\Delta^* \in \underset{\Delta >0}{\arg\max} \big(-\Delta\log\big(1-\frac{\Delta^{\gamma-1}}{(\gamma-1)!}e^{-\Delta}\big)\big)$.
\end{lemma}

Combining these lemmas, we find that the maximum value of $f(\boldsymbol{\alpha},\boldsymbol{\Delta},\boldsymbol{\gamma}) =\sum_{r\in[\ell]}\alpha_r h(\Delta_r,\gamma_r)$
also satisfies
    \begin{equation}
        \max_{\boldsymbol{\alpha},\boldsymbol{\Delta}}
    f(\boldsymbol{\alpha},\boldsymbol{\Delta},\boldsymbol{\gamma})
    \sim \sqrt{\frac{\gamma_{r^*}}{2\pi}}
    \qquad \text{as } \gamma_{r^*}\to\infty,
    \end{equation}
where $r^*=\arg\max_{r\in[\ell]}\gamma_r$ and thus $\gamma_{r^*} = \gamma_{\max}$.  Combined with the counting bound from Theorem \ref{thm:counting} and the fact that $k \log k = \frac{\theta\log 2}{1-\theta} k\log_2\frac{n}{k}$, this gives the following corollary.

\begin{corollary} \label{cor:large_conv}
    For the GT-ST problem, any test design satisfying Assumptions \ref{asm:sparsity} and \ref{asm:sub-matrix_regularity} achieves a rate of at most $\min\big\{1, \frac{1-\theta}{\theta \log 2} \sqrt{\frac{\gamma_{\max}}{2 \pi}} (1+o(1)) \big\}$ as $\gamma_{\max} \to \infty$, where the limit is after already having taken $n \to \infty$ with $k = \Theta(n^{\theta})$ in accordance with the definition of an achievable rate.
\end{corollary}


\subsubsection{Discussion}

Lemma \ref{cor: bestdenom} reveals that the largest denominator in \eqref{eq:selec-thres} is obtained when all of the tests use a threshold of $\gamma_{\max}$, thus simplifying Theorem \ref{thm: multi} to the single-threshold result of Theorem \ref{thm: sing}.  This suggests that the flexibility of selectable thresholds may have limited (if any) benefit when it comes to the information-theoretic limits.  
However, as we saw in Section \ref{sec:achievability}, the notion of selectable thresholds can lead to distinct algorithmic ideas and the ability to incorporate methods from the standard GT setting.

We observe that Theorem \ref{thm: sing} (with $\gamma_{\max}$) exactly matches the second term in Theorem \ref{thm:fixed}.  When $\theta$ is very close to 1, the remaining $O\big(k \log\frac{n}{k} \big)$ becomes insignificant, so when expressing the results as a coefficient to $k \log k$ (rather than the usual $k \log \frac{n}{k}$), we find that the achievability and converse bounds match as $\theta \to 1$.  However, for fixed $\theta \in (0,1)$ a gap remains, as seen in Figure \ref{fig:Gamma10}.

Finally, one interpretation of Corollary \ref{cor:large_conv} is that the number of tests must always have a $k \log k$ term whose coefficient decays no faster than $\Theta\big(\frac{1}{\sqrt{\gamma_{\max}}}\big)$.  This implies that in our achievable number of tests in \eqref{eq:T_large}, the $\frac{1}{\sqrt{\gamma_{\max}}}$ dependence in the second term cannot be improved.  (While it uses $k \log n$ instead of $k \log k$, the two are equivalent to within a constant factor except when $\theta \to 0$, and they are equivalent to within $1+o(1)$ as $\theta \to 1$.)  However, this does not rule out the possibility of improving \eqref{eq:T_large} in other ways; in particular, the sum would ideally be replaced by a maximum, and the precise hidden constant in $O(\cdot)$ could be sought.

\section{Conclusion}

We have introduced the problem of Group Testing with Selectable Thresholds (GT-ST), and established several achievability and converse bounds summarized as follows depending on the maximum threshold $\gamma_{\max}$:
\begin{itemize}
    \item For unbounded $\gamma_{\max}$, the optimal rate is 1.
    \item For large $\gamma_{\max}$, a rate of the form $1 - O\big(\frac{1}{\sqrt{\gamma_{\max}}}\big)$ can be achieved, and under the regularity conditions in our converse, the coefficient to $k \log k$ cannot decay any faster than $\frac{1}{\sqrt{\gamma_{\max}}}$.
    \item For fixed $\gamma_{\max}$, our achievability and converse results coincide in the limit $\theta \to 1$, but there is a gap between the two for fixed $\theta \in (0,1)$. 
\end{itemize}
Some immediate directions for future research include the following:
\begin{itemize}
    \item {\bf Other algorithmic benefits of selectable thresholds.}  We have seen that allowing selectable thresholds does not improve the optimal rate compared to only allowing the highest threshold.  However, it could still be of interest to further explore the extent to which the added flexibility can simplify algorithm design, possibly using more than two thresholds (unlike our approach that only uses the smallest and largest thresholds).
    \item {\bf Converse for arbitrary designs.} We expect that our main converse results should hold for arbitrary test designs, but formally proving this appears to be difficult.
\end{itemize}

\appendices

\section{Technical Lemmas}
Here we present some useful technical lemmas that are used throughout the paper.
\begin{lemma}[Multiplicative Chernoff Bound {\cite[Corollary 4.6]{mitzenmacher2017probability}}] \label{lem:chernoff}
    Let $X_1,\dots, X_n$ be a sequence of i.i.d.~${\rm Bernoulli}(p)$ random variables, and define $\mu = np$. Then, for any $\delta \in (0,1)$ it holds that
    \begin{equation}
        \label{eq:chernoff} \Pr\left( \left \lvert \sum_{i=1}^n X_i -  \mu \right\rvert \geq \delta \mu\right) \leq 2\exp\left(-\frac{\delta^2 \mu}{3}\right).
    \end{equation}
\end{lemma}

\begin{lemma}[McDiarmid's Inequality {\cite{mcdiarmid1989on}}] \label{lem:mcdiarmid}
    Fix an arbitrary set $\mathcal{X}$, and suppose that $X_1,\dots,X_n$ is a sequence of independent random variables supported on $\mathcal{X}$. Fix any function $f: \mathcal{X}^n \to \mathbb{R}$ satisfying the \emph{bounded differences property}, in the sense that there exist constants $c_1,\dots, c_n$ such that the following holds for all $i =1,\dots,n$ and all $(x_1,\dotsc,x_n) \in \mathcal{X}^n$:
    \begin{equation}
        \label{eq:bounded_diff} \sup_{x'_i \in \mathcal{X}} \, \lvert f(x_1,\dots, x_{i-1}, x_i, x_{i+1},\dots x_n) - f(x_1,\dots, x_{i-1}, x_i', x_{i+1},\dots x_n) \rvert \leq c_i.
    \end{equation}
    Then, for any $\delta > 0$, it holds that
    \begin{equation}
        \label{eq:mcdiarmid} \Pr\big(\big\lvert f(X_1,\dots,X_n) - \mathbb{E}[f(X_1,\dots,X_n)] \big\rvert \geq \delta \big) \leq 2\exp \left(-\frac{2\delta^2}{\sum_{i=1}^nc_i^2}\right).
    \end{equation}
\end{lemma}

\begin{lemma}[Poisson Approximation to the Binomial {\cite[Theorem 5.5]{mitzenmacher2017probability}}] \label{lem:poisson_approx}
    Let $n > 0$ and $p \in (0,1)$ be such that $\lim_{n \to \infty} np = \lambda$ for some fixed $\lambda > 0$. Then, for any fixed positive integer $x$, it holds that
    \begin{equation}
        \label{eq:poisson_approx} \lim_{n \to \infty} \Pr({\rm Bin}(n,p) = x)  = \frac{\lambda^xe^{-\lambda}}{x!}.
    \end{equation}
\end{lemma}

To state the next lemma, we require the following definition:
\begin{definition} \label{def:inc_dec_rvs}
     Let $\omega, \omega' \in \{0,1\}^n$, and equip $\{0,1\}^n$ with the coordinate-wise partial order defined by $\omega \leq \omega'$ whenever $\omega_i \leq \omega_i'$ for all $i \in [n]$. An event $A \subseteq \{0,1\}^n$ is called \emph{increasing} if $\omega \in A$ and $\omega \leq \omega'$ jointly imply that $\omega' \in A$. 
    Under the same setup, $A$ is said to be \emph{decreasing} if $\omega \in A$ and $\omega \geq \omega'$ jointly imply that $\omega' \in A$.
\end{definition}

\begin{lemma}[FKG Inequality {\cite[Theorem 2.4]{grimmett1999percolation}}] \label{lem:FKG}
    Let $A_1,\dots, A_m \subseteq \{0,1\}^n$ be a collection of events that are either all increasing, or all decreasing, and let $\Pr_\nu$ be the product measure on $\{0,1\}^n$ induced by the ${\rm Bernoulli}(\nu)$ measure on $\{0,1\}$. Then, for any $\nu \in [0,1]$, it holds that
    \begin{equation}
        \label{eq:FKG} \Pr\nolimits_\nu\left(\bigcap_{i=1}^m A_i \right) \geq \prod_{i=1}^m \Pr \nolimits_\nu(A_i).
    \end{equation}
\end{lemma}

\begin{lemma}[Bounds on the Binomial Coefficient {\cite[Lemma 4.7.1]{ash1990information}}] \label{lem:binom_bounds}
    Fix a positive integer $n \in \mathbb{N}$. For any $p \in (0,1)$ such that $np$ is an integer, it holds that
    \begin{equation}
        \label{eq:binom_bounds} \frac{\exp({nH(p)})}{2\sqrt{2np(1-p)}}  \leq {n \choose np} \leq \frac{\exp(nH(p))}{\sqrt{2\pi np(1-p)}},
    \end{equation}
    where $H(p) \coloneqq -p\log p - (1-p) \log (1-p)$ is the binary entropy function measured in nats.
\end{lemma}

\section{Achievability Proofs} \label{app:ach_proofs}

\subsection{Proof of Theorem \ref{thm:unbounded} (Achievability for Unbounded $\gamma_{\max})$} \label{app:unbounded_proof}
As discussed in Section \ref{sec:achievability}, our algorithm consists of two batches of tests: the first batch uses $T_1 = (k \log_2 \frac{n}{k})(1+o(1))$ tests to run the \textsc{Subset} algorithm from \cite{mcmorrow2026optimal}, which yields an estimate $\widehat{S}_1 \subseteq S$ of size $\lvert \widehat{S}_1 \rvert \geq k - k /\log^4k$ with probability $1-o(1)$. Throughout the rest of the proof, we will implicitly condition on this first batch being successful. In the second batch, we aim to resolve the status of the remaining $\lvert S \setminus \widehat{S}_1 \rvert \leq k/\log^4k$ tests by doing the following:
\begin{enumerate}
    \item Design $T_2'= (\frac{ek}{\log^4k} \log n)(1+o(1))$ tests for a standard group testing problem with $k/\log^4 k$ defectives by using a Bernoulli design (Definition \ref{def:Bernoulli}) with inclusion probability $p = \frac{\log^4k}{k}$.
    \item Repeat each of the $T_2'$ tests $R$ times (for a total of $T_2 = R \cdot T_2'$ tests) for some $R > 0$ to be determined later. In the $r$-th repetition of each test, the chosen defectivity threshold is $\gamma_r = \tau + r-1$ 
    for some $\tau > 0$ also to be determined later.
    \item For each test $t \in [T_2']$, do the following: (i) if all of its repetitions return the same outcome, then declare $\psi_t$ as undetermined and declare an error; (ii) otherwise, identify $\psi_t$ as the highest threshold $\gamma_r$ for which the test outcome is 1.
    \item For each test $t \in [T_2']$, subtract the number of items from $\widehat{S}_1$ in the test from $\psi_t$. Convert this to a standard group testing observation by setting $Y_t = 1$ if the residual number of defectives is non-zero (otherwise $Y_t = 0$).
     \item Using the test design for the $T_2'$ tests and the standard group testing outcomes $Y_1,\dots, Y_{T_2'}$, run the COMP algorithm to recover the remaining $\lvert S \setminus \widehat{S}_1 \rvert \leq k/\log^4k$ defectives. \label{enum:step4}
\end{enumerate}
Suppose that we can show (with high probability) that: 
\begin{itemize}
    \item[(i)] There exist values of $R,\tau$ such that all values of $\psi_t$ lie within the set
    $\{\tau + 1, \tau + 2,\dots, \tau + R-2\}$, and $T_2 = R\cdot T_2' = o(k \log n)$;
    \item[(ii)] $T_2'$ tests suffice to recover the remaining $\lvert S \setminus \widehat{S}_1 \rvert \leq k/\log^4k$ defectives from the standard GT outcomes.
\end{itemize} 
Then, a rate of $1$ is achievable in the GT-ST problem with unbounded $\gamma_{\max}$, since all defectives are recovered, and the number of tests across both stages is
\begin{equation}
    \label{eq:unbounded_both_stages_tests} T = \Big(k \log_2 \frac{n}{k}\Big)(1+o(1)) + o(k \log n) = \Big(k \log_2 \frac{n}{k}\Big)(1+o(1)).
\end{equation}
We now proceed to show each of these statements hold:

{\bf Establishing Statement (i).} To find appropriate values of $R,\tau$, observe that if $\psi_t$ is known to lie (with high probability) in an interval $(a,b)$ for all $t$, the choices $\tau = \lfloor a \rfloor , \; R = \lceil b\rceil  - \lfloor a \rfloor +1$ suffice to ensure that $\psi_t \in \{\tau + 1, \tau +2,\dots, \tau + R-2\}$, since $\tau \leq a$ and $\tau + R-1 \geq b$. Therefore, determining $R,\tau$ reduces to proving the existence of such an interval of size $o(\log^4 k)$, which can be shown via a concentration argument. Specifically, we wish to show that the deviation of $\psi_t$ from its mean is $o(\log^4k)$ for all $t$ with high probability, which will give us the desired interval. Since the tests follow a Bernoulli design with parameter $p = \frac{\log^4k}{k}$, the total number of defectives $\psi_t$ in a test is a ${\rm Bin}(k, p)$ random variable with mean $\mu_\psi \coloneqq \mathbb{E}[\psi_t] = \log^4 k$. Thus, the multiplicative Chernoff bound (Lemma \ref{lem:chernoff}) and a union bound over all $T_2'$ tests imply that
\begin{equation}
    \label{eq:pos_test_conc} \Pr(\exists t \in \{1,\dots, T_2'\} :\lvert \psi_t - \mu_\psi \rvert \geq \delta \mu_\psi) \leq 2T_2' \exp\Big(-\frac{\delta^2\mu_\psi}{3}\Big)
\end{equation}
for any $\delta \in (0,1)$. We then choose $\delta = \sqrt{\frac{3}{\mu_\psi}\log (2kT_2')}$, which ensures that 
\begin{equation}
    \label{eq:unbounded_conc_interval} \psi_t \in  \big((1-\delta)\mu_\psi, (1+\delta)\mu_\psi\big), \ \forall t \in \{1,\dots, T_2'\}
\end{equation}
with probability $1 - o(1)$. In accordance with our previous discussion, we can then set 
\begin{equation}
    \label{eq:choice_R_phi} \tau = \lfloor (1-\delta)\mu_\psi\rfloor , \quad R = \lceil (1+\delta)\mu_\psi\rceil - \lfloor (1-\delta)\mu_\psi\rfloor + 1 = O(\delta\mu_\psi)
\end{equation}
to ensure that $\psi_t \in \{\tau + 1,\tau + 2, \dots, \tau + R - 2\}$ for all $t$ with probability $1-o(1)$. Finally, to obtain a bound on the scaling of $R$, we observe from the choice of $\delta$ that
\begin{align}
    \label{eq:delta_mu_scaling_1} \delta\mu_\psi &= \sqrt{3\mu_\psi \log (2kT_2')} \\
    \label{eq:delta_mu_scaling_2} &= \log^2k \sqrt{3 \log(2kT_2')} \\
    \label{eq:delta_mu_scaling_3} &= O(\log^{5/2} k),
\end{align}
where \eqref{eq:delta_mu_scaling_3} uses the fact that $T_2' = O(k\log n/\log^4 k)$. Hence, $R = O(\log^{5/2}k)$, implying that the number of tests required to learn the thresholds $\psi_t$ is 
\begin{equation}
    \label{eq:T2} T_2 = R \cdot T_2' = O(\log^{5/2} k ) \times O\left(\frac{k}{\log^4 k}\log n\right) = o(k \log n),
\end{equation}
as desired. This implies that statement (i) holds with probability $1-o(1)$.

{\bf Establishing Statement (ii).} Conditioned on statement (i) holding, we now show that $T_2' = (\frac{ek}{\log^4k}\log n)(1+o(1))$ tests suffice to recover the remaining $\lvert S \setminus \widehat{S}_1 \rvert \leq k/\log^4k$ defectives using the COMP algorithm. While this claim follows directly from standard results for COMP \cite[Thm. 4]{chan2011non}, we formally prove it here for completeness, and since the argument is relatively simple. 

Recall that the COMP algorithm fails if there is at least one non-defective item that does not appear in any negative tests. Further recalling that the tests follow a Bernoulli design, the probability of a non-defective item appearing in a given test, alongside the test being negative is $p(1-p)^{\lvert S \setminus \widehat{S}_1\rvert}$. Therefore, the probability that any given non-defective appears in no negative tests is $(1 - p(1-p)^{\lvert S \setminus \widehat{S}_1\rvert})^{T_2'}$. Using the fact that $\lvert S \setminus \widehat{S}_1 \rvert \leq k/\log^4k$ and recalling that $p = \frac{\log^4k}{k}$, this probability can be upper bounded by $(1-p(1-p)^{1/p})^{T_2'}$, which we can further upper bound by $\exp(-T_2' \, p(1-p)^{1/p})$. A union bound over all non-defective items then implies that 
\begin{equation}
    \label{eq:comp_error} \Pr \left( \bigcup_{i \in [n] \setminus S} \{i \text{ is not in any negative tests} \}\right) \leq n\exp\left(-T_2' \,p(1-p)^{1/p}\right),
\end{equation}
where we have upper bounded the number of non-defectives by $n$. Choosing
\begin{equation}
    \label{eq:T_2'_choice} T_2' = \frac{\log n}{p(1-p)^{1/p}}\Big(1 + \frac{1}{\log \log n}\Big)
\end{equation}
ensures that the error probability is upper bounded by $n^{-1/\log \log n} = o(1)$. Finally, applying $(1-p)^{1/p} = e^{-1}(1-o(1))$ and substituting in the value of $p$ implies that the number of tests sufficient for COMP to succeed with probability $1-o(1)$ is $T_2' = (\frac{ek}{\log^4k} \log n)(1+o(1))$, which implies that statement (ii) holds with probability $1-o(1)$.

The proof of Theorem \ref{thm:unbounded} is completed by combining \eqref{eq:unbounded_both_stages_tests} with the fact that both statements (i) and (ii) hold with probability $1-o(1)$.

\subsection{Proof of Theorem \ref{thm:large} (Achievability for Large $\gamma_{\max}$)} \label{app:large_proof}

In the case that $\gamma_{\max}$ is ``bounded but large'', we use a similar two-batch template to the one used in the proof of Theorem \ref{thm:unbounded}. The first batch proceeds in the exact same way as in Theorem \ref{thm:unbounded}, so we do not repeat the details. For the second batch (conditioned on batch 1 succeeding), we do the following: 
\begin{enumerate}
    \item Design $T_2' = (\frac{8k}{\gamma_{\max}}\log n)(1+o(1))$ tests using a Bernoulli design with inclusion probability $p = \frac{\gamma_{\max}}{2k}$, and calculate $\psi_{t,1}$, the number of \textsc{Subset}-identified defectives in test $t$ for each $t \in [T_2']$.
    \item Define the interval $(\gamma^-, \gamma^+) \coloneqq (\frac{\gamma_{\max}}{2} - 2\sqrt{\gamma_{\max}}, \,\frac{\gamma_{\max}}{2} + 2\sqrt{\gamma_{\max}})$, and repeat each of the $T_2'$ tests with incrementing thresholds $\lfloor \gamma^-\rfloor, \lfloor\gamma^-\rfloor + 1, \dots, \lceil \gamma^+ \rceil$, for a total of $T_2 = \gamma_{{\rm dif}}\cdot T_2'$ tests, where we define $\gamma_{{\rm dif}} \coloneqq \lceil \gamma^+\rceil - \lfloor\gamma^-\rfloor + 1$.
    \item For each test $t \in [T_2']$, observe the test outcomes $\{Y_{t,r}\}_{r=0}^{\gamma_{{\rm dif}}-1}$, with the $r$-th test outcome being given by $Y_{t,r} =\mathbbm{1} \{\psi_t \geq \lfloor \gamma^-\rfloor + r\}$, and discard those whose test outcomes $\{Y_{t,r}\}_{r=0}^{\gamma_{\rm dif}-1}$ are all identical (as this implies these tests' values of $\psi_t$ lie outside of $(\gamma^-, \gamma^+)$).
    \item For the remaining $\widetilde{T}_2 \leq T_2'$ tests, use the test outcomes $\{Y_{t,r}\}_{r=0}^{\gamma_{\rm dif} -1 }$ to find $\psi_t$, and calculate the number of non \textsc{Subset}-identified defectives $\widetilde{\psi}_t = \psi_t - \psi_{t,1}$ in the test.
    \item Defining a test $t \in[\widetilde{T}_2]$ to be \emph{residually negative} if $\widetilde{\psi}_t = 0$, return the subset of items from $[n] \setminus \widehat{S}_1$ that do not appear in any residually negative tests (i.e., declare them to be defective).
\end{enumerate}
To prove Theorem \ref{thm:large}, it suffices to show that out of the non-discarded $\widetilde{T}_2$ tests, there are sufficiently many residually negative tests, so that all non-defectives are included in at least one with high probability. To this end, we first derive a (high probability) lower bound on $\widetilde{T}_2$.

In any given test $t \in [T_2']$, each item is included in it independently with probability $p = \frac{\gamma_{\max}}{2k}$. Thus, $\psi_t$ has a ${\rm Bin}(k, p)$ distribution with mean $\mu_\psi \coloneqq \mathbb{E}[\psi_t] =\frac{\gamma_{\max}}{2}$. The multiplicative Chernoff bound (Lemma \ref{lem:chernoff}) thus implies that
\begin{equation}
    \label{eq:ell_t_conc} \Pr(\lvert \psi_t - \mu_\psi \rvert \geq 2\sqrt{\gamma_{\max}}) \leq 2 \exp\left(-\frac{8}{3}\right) \approx 0.139.
\end{equation}
Thus, the number of defectives $\psi_t$ lies within $(\gamma^-, \gamma^+)$ with probability at least $1 - 2\exp(-8/3) \approx 0.861$. Since the number of defectives is independent across tests, it follows that the number of non-discarded tests $\widetilde{T}_2$ stochastically dominates a ${\rm Bin}(T_2', 1 - 2\exp(-8/3))$ random variable. Using this and the multiplicative Chernoff bound, it follows that
\begin{equation}
    \label{eq:T2_tilde_conc} \Pr\left(\widetilde{T}_2 \leq \frac{1}{2}T_2'\right) \leq \Pr\left({\rm Bin}(T_2', 1 - 2\exp(-8/3)) \leq \frac{1}{2}T_2'\right) = o(1).
\end{equation}
Thus, $\widetilde{T}_2 \geq \frac{1}{2}T_2'$ with probability $1-o(1)$.

We now show that a constant fraction of the non-discarded tests are \emph{residually negative}, i.e., contain no item from $S \setminus \widehat{S}_1$.  To do so, we momentarily return to considering the entire sequence of tests indexed by $t \in [T'_2]$, and let $B$ denote the number of such tests containing at least one item from $S \setminus \widehat{S}_1$.  Using $\lvert S \setminus \widehat{S}_1 \rvert \leq k/\log^4k$ and Bernoulli's inequality, a given test contains an item from $S \setminus \widehat{S}_1$ with probability at most
\begin{equation}
    \label{eq:residual_negative_lb} 1 - \left(1 - \frac{\gamma_{\max}}{2k}\right)^{k/\log^4k} \leq \frac{\gamma_{\max}}{2\log^4k}.
\end{equation}
Since the tests are independent, this implies that $B \sd{\rm Bin}(\widetilde{T}_2, \frac{\gamma_{\max}}{2\log^4 k})$. Thus, the Chernoff bound and the fact that $\frac{\gamma_{\max}}{2\log^4 k} = o(1)$ imply that $B \le \frac{1}{4}T_2'$ with probability $1-o(1)$.  Combined with the previous paragraph, this implies that the number of non-discarded residually negative tests is at least $\widetilde{T}_2 - \frac{1}{4}T_2' \geq \frac{1}{4}T_2'$ with probability $1-o(1)$.

Fix a non-defective item $i \in [n] \setminus S$. Conditioned on there being at least $\frac{1}{4}T_2'$ residually negative tests, the probability that $i$ appears in none of them (thus being incorrectly classified as defective) is upper bounded by $(1 - \frac{\gamma_{\max}}{2k})^{T_2'/4}$, which can be further upper bounded by $\exp(-\frac{T_2'\gamma_{\max}}{8k})$. Taking a union bound over all non-defective items (of which there are at most $n$), the probability of \emph{any} non-defective not appearing in any residually negative test is upper bounded by $n\exp(- \frac{T_2'\gamma_{\max}}{8k})$. The choice $T_2' = (\frac{8k}{\gamma_{\max}}\log n)(1+\frac{1}{\log \log n})$ therefore implies that this probability is upper bounded by $n^{-1/\log \log n} = o(1)$. Since each of the designed $T_2'$ tests is conducted a total of $\gamma_{\rm dif} \coloneqq \lceil \gamma^+ \rceil - \lfloor \gamma^- \rfloor + 1 = O(\sqrt{\gamma_{\max}})$ times, the total number of tests conducted is
\begin{align}
    \label{eq:T_2_total_1} T_2 &= \gamma_{\rm dif} \cdot T_2' \\
    \label{eq:T_2_total_2} &= O(\sqrt{\gamma_{\max}}) \times \left(\frac{8k}{\gamma_{\max}}\log n \right)\left(1 + \frac{1}{\log \log n}\right) \\
    \label{eq:T_2_total_3} &= O\left(\frac{k}{\sqrt{\gamma_{\max}}}\log \frac{n}{k}\right),
\end{align}
where \eqref{eq:T_2_total_3} follows since $\log n = \Theta(\log \frac{n}{k})$, as $k = \Theta(n^\theta)$ for some $\theta \in (0,1)$. We conclude by noting that the number of tests across both batches is
\begin{equation}
    \label{eq:large_T_total} T = \Big(k \log_2 \frac{n}{k}\Big)(1+o(1)) + O\left(\frac{k}{\sqrt{\gamma_{\max}}}\log \frac{n}{k}\right),
\end{equation}
thus completing the proof.

\subsection{Proof of Theorem \ref{thm:fixed} (Achievability for Fixed $\gamma_{\max}$)} \label{app:fixed_proof}
In the case that $\gamma_{\max} = r$ for some fixed $r > 0$, we again adopt a two-batch approach. In the first batch, we apply a standard group testing algorithm with $T_1$ tests, similarly to the approaches taken in Theorems \ref{thm:unbounded} and \ref{thm:large}. Unlike these approaches, however, we incorporate a different test design and decoder, which we now outline: 
\begin{itemize}
    \item Rather than designing tests using a Bernoulli design (Definition \ref{def:Bernoulli}) with inclusion probability $p = \frac{\log 2}{k}$, the tests in batch 1 are now designed according to a near-constant column weight (NCC) design (Definition \ref{def:NCC}), with each item being independently placed in $L_1 = \frac{T_1 \log 2}{k}$ tests chosen uniformly at random with replacement.
    \item Rather than using $(k \log_2 \frac{n}{k})(1+o(1))$ tests in the first batch, we now let this number be $T_1 = (\frac{k}{\log 2}\log_2 \frac{n}{k})(1+o(1))$.
    \item Rather than using the \textsc{Subset} algorithm from \cite{mcmorrow2026optimal} as our decoder, we apply both COMP and DD decoders (Section \ref{sec:related_work}) to the test outcomes, and rely on their approximate recovery guarantees.
\end{itemize}
The reason for switching to the NCC design is that it has strictly better constants than the Bernoulli design when paired with COMP or DD decoding.  Moreover, the reason for switching to COMP and DD is (i) these are very simple and computationally efficient, and (ii) the SUBSET and SUPERSET problems are handled separately in \cite{mcmorrow2026optimal}, whereas the use of COMP and DD provides a simple way to obtain a subset and a superset simultaneously.

It is well known (see e.g., \cite[Ch. 2]{aldridge2019group}) that both COMP and DD can be run in polynomial time. Additionally, the analysis in \cite[App. D]{mcmorrow2026optimal} reveals that under our choice of $T_1$, the COMP and DD algorithms produce sets $S_{\rm COMP} \supseteq S$ and $S_{\rm DD} \subseteq S$ of size $k(1+o(1))$ and $k(1-o(1))$ respectively, with probability $1-o(1)$. 

Conditioned on batch 1 ``succeeding'' (i.e., the sets $S_{\rm COMP}$ and $S_{\rm DD}$ satisfying the above properties), we then do the following in batch 2:
\begin{enumerate}
    \item Design 
    \begin{equation}
        \label{eq:fixed_T_2} T_2 = \frac{k \log k}{\nu \log\Big(\frac{1}{1 - \nu^{r-1}e^{-\nu}/(r-1)!}\Big)}(1+o(1))
    \end{equation}
    tests according to the NCC design with parameter $L_2 = \frac{T_2 \nu}{k}$, where $\nu >0$ will be optimized later.
    \item Conduct each of the $T_2$ tests with defectivity threshold $\gamma = \gamma_{\max} = r$.
    \item For each $i \in S_{\rm COMP} \setminus S_{\rm DD}$, find a test $t \in [T_2]$ such that (i) there are exactly $r-1$ items from $S_{\rm DD}$ included in the test, and (ii) the only item from $S_{\rm COMP} \setminus S_{\rm DD}$ included in the test is $i$. If there exists some $i \in S_{\rm COMP} \setminus S_{\rm DD}$ having no such tests, then declare an error.
    \item Classify each $i \in S_{\rm COMP} \setminus S_{\rm DD}$ as defective if $Y_{t_i} = 1$, and non-defective otherwise, where $t_i \in [T_2]$ is a test satisfying conditions (i) and (ii) above for $i$.
    \item Declare the defective set as the union of $S_{\rm DD}$ and the set of classified defectives from $S_{\rm COMP} \setminus S_{\rm DD}$.
\end{enumerate}
It is straightforward to verify that this procedure runs in polynomial time, e.g., by first forming a length-$L_2$ list of which tests (of batch 2) each item is included in, one can readily verify $O(nk \log n)$ runtime.


If a test $t$ satisfies properties (i) and (ii) in step 3 above for an item $i \in S_{\rm COMP} \setminus S_{\rm DD}$, we say that test $t$ \emph{resolves} item $i$, since the status of $i$ can be determined from the test outcome. Therefore, if there is at least one test that resolves $i$ for all $i \in S_{\rm COMP} \setminus S_{\rm DD}$, successful recovery is possible, since all items in $S_{\rm COMP} \supseteq S$ can be correctly classified (meaning $S$ can be correctly recovered). Thus, to show that $S$ is successfully recovered (with high probability), it suffices to show that the choice of $T_2$ leads to all $i \in S_{\rm COMP} \setminus S_{\rm DD}$ being resolved by at least one test.

To do so, we first consider the probability that exactly $r-1$ items from $S_{\rm DD}$ are included in a test. The NCC design implies that each item is included in a test with probability $1 -(1 - \frac{1}{T_2})^{L_2}$. Using the definition of $L_2$, and the fact that $\frac{\nu}{k} = o(1)$, this probability can be expressed as 
\begin{align}
    \label{eq:i_in_t_1} \Pr({\rm item} \ i \ {\rm in \ test } \ t) &= 1 - \left(1 - \frac{1}{T_2}\right)^{L_2} \\
    \label{eq:i_in_t_2} &= 1 - \exp\left(- \frac{\nu}{k}\right)(1-o(1)) \\
    \label{eq:i_in_t_3} &= \frac{\nu}{k}(1+o(1)).
\end{align}
Denote the number of items from $S_{\rm DD}$ in test $t$ by $D_t$. The above calculations, and the fact that items are placed in tests independently of one another, imply that $D_t \sim {\rm Bin}\big(\lvert S_{\rm DD}\rvert, \frac{\nu}{k}(1+o(1)\big)$. Since $\lvert S_{\rm DD}\rvert \cdot \frac{\nu}{k}(1+o(1)) \to \nu$, the Poisson approximation to the binomial (Lemma \ref{lem:poisson_approx}) implies that 
\begin{equation}
    \label{eq:D_t_poisson_approx} q_{\nu, r} \coloneqq \Pr(D_t = r-1) = \frac{\nu^{r-1}e^{-\nu}}{(r-1)!}(1+o(1)).
\end{equation}
Thus, each test contains $r-1$ defectives from $S_{\rm DD}$ with probability approaching $\nu^{r-1}e^{-\nu}/(r-1)!$. Now, let
\begin{equation}
    \label{eq:defn_V} V = \sum_{t=1}^{T_2} \mathbbm{1}\{D_t = r-1\}
\end{equation}
be the number of tests that include exactly $r-1$ items from $S_{\rm DD}$. By \eqref{eq:D_t_poisson_approx} and the linearity of expectation, it follows that $\mu_V \coloneq \mathbb{E}[V] = q_{\nu, r} T_2$. To obtain a high probability lower bound on $V$, define $Z_{\ell i} \in [T_2]$ as the $\ell$-th test drawn for item $i \in S_{\rm DD}$ to be included in, and note that there exists a function $f:\mathbb{R}^{L_2\cdot \lvert S_{\rm DD}\rvert} \to \mathbb{R}$ such that $V = f(\{Z_{\ell i}\}_{(\ell , i) \in [L_2] \times S_{\rm DD}})$. Moreover, changing the value of any $Z_{\ell i}$ can lead to the value of $V$ changing by at most two, since in the worst case, it could be removed from a test containing exactly $r-1$ defectives from $S_{\rm DD}$ including $i$ (leading to $V$ decreasing by one), and be added to a different test containing exactly $r-1$ defectives from $S_{\rm DD}$ \emph{not} including $i$ (leading to $V$ further decreasing by one). Thus, $f$ satisfies the \emph{bounded differences property} \eqref{eq:bounded_diff} with $c_{\ell i} =2$ for all $(\ell, i ) \in [L_2] \times S_{\rm DD}$, meaning we can apply McDiarmid's inequality (Lemma \ref{lem:mcdiarmid}) to obtain 
\begin{align}
    \label{eq:V_conc_mcdiarmid_1} \Pr \left( \lvert V - \mu_V \rvert \geq \delta \right) &\leq 2\exp\left(- \frac{\delta^2}{2L_2 \lvert S_{\rm DD}\rvert}\right) \\
    \label{eq:V_conc_mcdiarmid_2} &= 2\exp\left(- \frac{\delta^2}{2\nu T_2} (1+o(1))\right)
\end{align}
for any $\delta > 0$, where we applied $L_2 = \frac{\nu T_2}{k}$ and $|S_{\rm DD}| = k(1-o(1))$ in \eqref{eq:V_conc_mcdiarmid_2}. Choosing $\delta = q_{\nu, r} T_2/\log n$ gives
\begin{equation}
    \label{eq:V_conc_delta_choice} \Pr \left(\lvert V - \mu_V \rvert \geq \frac{q_{\nu, r} T_2}{\log n}\right) \leq 2 \exp\left(- \frac{q_{\nu, r}^2T_2}{2\nu\log^2 n} (1+o(1))\right) = o(1),
\end{equation} 
since $T_2 = \Omega(k)$. In particular, this implies that 
\begin{align}
    \label{eq:V_lb_1} V &\geq q_{\nu, r}T_2\Big(1 - \frac{1}{\log n}\Big)  
\end{align}
with probability $1-o(1)$.  

Now, fix an item $i \in S_{\rm COMP} \setminus S_{\rm DD}$.  Let $U_i$ be the number of tests that include any $j \in S_{\rm COMP} \setminus S_{\rm DD}$ with $j \ne i$.  Since each item is placed in (at most) $L_2 = \Theta(\log n)$ tests, and since $|S_{\rm COMP}  \setminus S_{\rm DD}| = o(k)$, we have $U_i = o(k \log n) = o(T_2)$ (with probability one).  Hence, conditioned on \eqref{eq:V_lb_1} holding, each of the $L_2$ tests drawn for $i$ to be included in resolves it with probability at least
\begin{equation}
    \label{eq:prob_resolve_single_test} 
    \frac{V - U_i}{T_2} = q_{\nu, r} - o(1).
\end{equation}
Thus, the probability $i$ is not resolved in \emph{any} of the tests it is drawn to be included in is at most
\begin{align}
    \label{eq:prob_resolve_no_tests_1} \left(1 - q_{\nu, r} + o(1) \right)^{L_2} &= \exp\left(L_2 \log \left(1 - q_{\nu, r} + o(1)\right)\right) \\
    &= \exp\left(- L_2 \log  \left(\frac{1}{1 - q_{\nu, r} + o(1)} \right)\right).
\end{align}
Substituting $L_2 = \frac{\nu T_2}{k}$, taking a union bound over all $i \in  S_{\rm COMP} \setminus S_{\rm DD}$, and trivially upper bounding the size of $S_{\rm COMP} \setminus S_{\rm DD}$ by $k$, the probability of \emph{any} $i \in  S_{\rm COMP} \setminus S_{\rm DD}$ not being resolved by any test is upper bounded by
\begin{equation}
    \label{eq:prob_not_resolved_ub} k \exp\left(- \frac{\nu T_2}{k} \log  \left(\frac{1}{1 - q_{\nu, r} + o(1)}\right)\right). 
\end{equation}
Therefore, choosing
\begin{equation}
    \label{eq:fixed_gamma_T_choice} T_2 = \frac{k\log k}{\nu \log \big(\frac{1}{1 - q_{\nu, r} + o(1)}\big)}\left(1 + \frac{1}{\log \log k}\right)
\end{equation}
suffices to ensure that \eqref{eq:prob_not_resolved_ub} is upper bounded by $k^{-1/\log \log k} = o(1)$. Substituting the expression for $q_{\nu, r}$ in \eqref{eq:D_t_poisson_approx}, this simplifies to
\begin{align}
    \label{eq:T_2_q_sub_2} T_2 &= \frac{k \log k}{\nu \log \big(\frac{1}{1 - \nu^{r-1}e^{-\nu}/(r-1)!}\big)}(1+o(1)).
\end{align}
Optimizing over $\nu$, substituting $r = \gamma_{\max}$, and summing the number of tests across both batches, it follows that
\begin{equation}
    \label{eq:large_T_final} T = \Big(\frac{k}{\log 2}\log_2 \frac{n}{k}\Big)(1+o(1)) + \min_{\nu > 0} \frac{k \log k}{\nu \log \big(\frac{1}{1 - \nu^{\gamma_{\max}-1}e^{-\nu}/(\gamma_{\max}-1)!}\big)}(1+o(1))
\end{equation}
tests suffice to ensure $P_{\rm e} \to 0$, as claimed. 

\section{Converse Proofs} 
\subsection{Roadmap} \label{sec:roadmap}

As discussed in Section \ref{sec:converse}, the proofs of Theorems \ref{thm: sing} and \ref{thm: multi} revolve around the concepts of \emph{informativeness} and \emph{masking}, which are defined as follows:
\begin{definition} \label{def:informative_masked}
    A test $t \in [T]$ with threshold $\gamma^{(t)}$ is called \emph{informative} for an item $i \in [n]$ if $i$ is included in the test, and there are exactly $\gamma^{(t)} -1$ other defective items (not including $i$) in the test.
    If there are no tests $t \in [T]$ that are informative for $i \in [n]$, we say that $i$ is \emph{masked}.
\end{definition}
Such a concept has been used extensively in the classical group testing literature to prove non-adaptive converse bounds \cite{aldridge2018individual, coja2020optimal, bay2022optimal, mcmorrow2026optimal}, so it is natural to adopt such a strategy in the TGT and GT-ST settings we consider. Intuitively, an item being masked means that it is difficult for any decoder to correctly classify, since its defectivity status is conditionally independent of the test outcomes. Thus, if there are many masked defectives and non-defectives, then any decoder would effectively have to ``guess'' the status of each of these items, which will fail with high probability. While this is only an intuitive argument, it is formalized in Section \ref{sec:step6_single}. 

Thus, we seek to show that for values of $T$ below that stated in Theorems \ref{thm: sing} and \ref{thm: multi}, there are many masked defectives and non-defectives with high probability. In pursuit of this goal, we make the following simplifying assumptions on $S$ and $k$ (to be dropped later) that will be useful in the analysis:
\begin{itemize}
    \item We assume that $S$ is generated according to the \emph{i.i.d.~prior}, where each item is defective independently with probability $q = \frac{k}{n}$.
    \item Additionally, we assume that the sparsity parameter $\theta$ is \emph{arbitrarily close to $1$}; specifically, $k = n^{1-\xi}$, where $\xi > 0$ is arbitrarily small.
\end{itemize}

Under these assumptions, we demonstrate that there are many masked defectives and non-defectives with high probability via the following high-level roadmap. We will first state it for the simpler proof of Theorem \ref{thm: sing} (TGT), and outline the changes made for proving Theorem \ref{thm: multi} (GT-ST) in Section \ref{sec:modifications}.
\begin{enumerate}
    \item Establish the existence of a set $\mathcal{I} \subset [n]$ of size $\lvert \mathcal{I}\rvert \geq n^{1-3\xi}$ whose \emph{masking events are independent}. This can be found using a greedy construction analogous to that of the Gilbert-Varshamov construction \cite{gilbert1952comparison}. \label{enum:step1_roadmap}
    \item For a given $i \in \mathcal{I}$ and $s \in \{0,\dots,d_i\}$ (where we recall that $d_i$ is the number of tests $i$ is included in), we seek to obtain a lower bound on the probability of $L_{s,i}$, where $L_{s,i}$ is the event that the following both occur: (i) exactly $s$ of the tests $i$ is included in contain at most $\gamma - 2$ other defectives; and (ii) the remaining $d_i - s$ tests that $i$ is included in contain at least $\gamma$ other defectives. \label{enum:step2_roadmap}
    \item Let $D_i$ denote the event that $i \in \mathcal{I}$ is masked. The probability of $D_i$ is then equal to the probability that $L_{s,i}$ occurs for at least one $s \in \{0,\dots, d_i\}$. Additionally, the events $\{L_{s,i}\}_{s=0}^{d_i}$ are mutually disjoint, meaning $\Pr(D_i)$ can be written as $\sum_{s=0}^{d_i} \Pr(L_{s,i})$, which can subsequently be lower bounded by the aforementioned bounds on $\Pr(L_{s,i})$, for each $s \in \{0,\dots,d_i\}$. \label{enum:step3_roadmap}
    \item We are then able to use the lower bounds on $\Pr(D_i)$ and the fact that the masking events for each $i \in \mathcal{I}$ are independent (under the i.i.d.~prior) to show that there are many masked items with high probability via a concentration argument. Given this, we can also show that among these masked items, there is at least one masked defective and $\omega(1)$ masked non-defectives with high probability. \label{enum:step4_roadmap}
\end{enumerate}
Given these masked items, we will be in a position to show that this implies that successful recovery is not possible, while simultaneously dropping the simplifying assumptions:
\begin{enumerate}
    \item[5.] We remove the first assumption by proving, via a conditioning argument, that the high-probability bound on the number of masked items under the i.i.d.~prior extends to the combinatorial prior in the ``high-$\theta$ regime''. \label{enum:step5_roadmap}
    \item [6.] Given this, we can show that having at least one masked defective and $\omega(1)$ masked non-defective implies a success probability of $o(1)$ for the MAP decoder, which implies the same for \emph{any} decoder. \label{enum:step6_roadmap}
    \item[7.] Finally, we transfer this high-$\theta$ result to any $\theta \in (0,1)$ via the contrapositive statement that achievability for smaller $\theta$ implies achievability for larger $\theta$ under suitable conditions, importantly showing that this reduction preserves Assumptions \ref{asm:regularity} and \ref{asm:sparsity}. \label{enum:step7_roadmap}
\end{enumerate}
We now follow this roadmap to give a detailed proof of Theorem \ref{thm: sing}.

\subsection{Proof of Theorem \ref{thm: sing} (Converse for the Single-Threshold Setting)} \label{app:pf_conv_single}

\subsubsection{Step 1 (Constructing $\mathcal{I}$)} \label{sec:step1_single}
We begin by constructing a set $\mathcal{I} \subset [n]$ such that the following properties hold: (i) $\lvert \mathcal{I}\rvert \geq n^{1-3\xi}$; and (ii) the masking events for all $i \in \mathcal{I}$ are mutually independent. To this end, we fix an arbitrary non-adaptive test design $\bX$, and represent $\bX$ as a \emph{bipartite graph} $G = ([n]\cup [T], E)$, where the edge set $E \subseteq [n] \times [T]$ is the set consisting of all pairs $(i,t) \in [n] \times [T]$ such that $X_{ti} = 1$. That is, there is an edge between item $i \in [n]$ and test $t \in [T]$ if and only if $i$ is included in $t$. Additionally, define the (bipartite) \emph{distance} between $i,j \in [n]$ to be minimum number of edges that need to be traversed to reach $j$ from $i$. To construct the set $\mathcal{I}$ satisfying the desired properties, we adopt the following greedy iterative procedure, which follows the ones used in \cite{coja2020optimal, bay2022optimal} (but with different analysis details to follow):
\begin{enumerate}
    \item Initialize $\mathcal{U} = [n]$ and $\mathcal{I} = \emptyset$.
    \item Select an arbitrary item $i \in \mathcal{U}$, add it to $\mathcal{I}$, and remove all items from $\mathcal{U}$ whose distance from $i$ is at most $4$ (including $i$ itself).
    \item Repeat Step 2 until $\mathcal{U} = \emptyset$, and return the set $\mathcal{I}$.
\end{enumerate}
The following lemma gives a lower bound on the size of the set returned by this procedure, alongside the minimum bipartite distance between items:
\begin{lemma}\label{lem:distance4_packing}
Consider a test design satisfying Assumptions \ref{asm:regularity} and \ref{asm:sparsity}. Then, if $k = n^{1-\xi}$, there exists a subset $\mathcal{I} \subset [n]$ such that the following two statements hold:
\begin{enumerate}
    \item For any two distinct items $i, j \in \mathcal{I}$, the distance between $i$ and $j$ in the bipartite graph exceeds $4$.
    \item The cardinality of $\mathcal{I}$ satisfies $|\mathcal{I}| \ge n^{1-3 \xi}$.
\end{enumerate}
\end{lemma}
\begin{proof}
The first statement follows immediately from the construction of $\mathcal{I}$, so it only remains to prove the second statement. To do so, fix an arbitrary item $i \in [n]$. We bound the number of items whose bipartite distance from $i$ is at most $4$:
\begin{itemize}
    \item \textbf{Distance $0$:} The item $i$ itself contributes $1$.
    \item \textbf{Distance $2$:}  
    The item $i$ participates in at most $d_{\max} \coloneqq \max_{i \in [n]} d_i$ tests. Each such test contains at most $w_{\max}$ additional items, where $w_{\max} \coloneqq \max_{t \in [T]} w_t$ (recalling our convention that $1+w_t$ is the test weight). Hence the number of items at distance $2$ from $i$ is at most $d_{\max} w_{\max}$.
    \item \textbf{Distance $4$:}  
    Each distance-$2$ item participates in at most $d_{\max}-1$ additional tests (excluding the test already used to reach it). Each such test contains at most $w_{\max}$ additional items. Therefore, the number of items at distance $4$ is bounded by
    $d_{\max}w_{\max}(d_{\max}-1)w_{\max} = d_{\max}(d_{\max}-1)w_{\max}^2$.
\end{itemize}
Combining these bounds, the total number of items within distance at most $4$ of any fixed item is bounded by
\begin{equation}
B = 1 + d_{\max}w_{\max} + d_{\max}(d_{\max}-1)w_{\max}^2 \le 2(d_{\max}w_{\max})^2.
\end{equation}
Thus at each iteration, at most $B$ items (including $i$ itself) are removed from $\mathcal{U}$. Since the removed sets are disjoint across iterations, the procedure selects at least
\begin{equation}
|\mathcal I| \ge \frac{n}{B} \ge \frac{n}{2(d_{\max} w_{\max})^2}
\end{equation}
items. Recalling that $d_{\max} = \Theta(\log n)$ and $w_{\max} = \Theta(n/k)$ due to Assumption \ref{asm:regularity}, it follows that $w_{\max} \leq \frac{n}{k}\log n$ for sufficiently large $n$. Using this and substituting in $k = n^{1-\xi}$, the size of $\mathcal{I}$ is lower bounded by 
\begin{align}
    \lvert \mathcal{I} \rvert &\geq \frac{n}{2(d_{\max} \cdot \frac{n}{k}\log n)^2} \\
    &= \frac{n^{1-2\xi}}{2 (d_{\max}\log n)^2} \\
    \label{eq:cardinality_I_ub} &\geq n^{1-3\xi},
\end{align}
where \eqref{eq:cardinality_I_ub} holds for sufficiently large $n$. This proves the second statement.
\end{proof}
Having established that property (i) holds for $\mathcal{I}$, it now remains to argue that the ``minimum distance property'' exhibited by the items in $\mathcal{I}$ implies that their masking events are independent. The key observation is that the masked status of any item $i \in [n]$ depends only on the set of items that share at least one test with $i$ (i.e., items with bipartite distance $2$ of $i$). Since any pair of distinct items $j, j' \in \mathcal{I}$ has distance exceeding $4$ by construction, the set of items impacting the masking event for $j$ (i.e., sharing at least one test with $j$) is disjoint from those impacting the masking event for $j'$.   This therefore implies that the masking events for all items in $\mathcal{I}$ are mutually independent, meaning property (ii) also holds for $\mathcal{I}$.

\subsubsection{Step 2 (Lower Bounding $\Pr(L_{s,i})$)} \label{sec:step2_single}

We first introduce some notation. For a given item $i \in [n]$, we define the \emph{co-occurrence matrix} $\bM^{(i)} \in \{0,1\}^{d_i\times(n-1)}$ to be the sub-matrix of $\bX$ containing the rows $t \in [T]$ such that $X_{ti} = 1$, and all columns except for the $i$-th one. Thus, the rows of $\bM^{(i)}$ correspond to the tests $i$ is included in, and the columns denote whether each of the remaining items are included in these tests. Additionally, we define the $t$-th \emph{defective weight} of $\bM^{(i)}$ as
\begin{equation}
    \label{eq:defective_weight} U_t \coloneqq \sum_{j \in [n] \setminus \{i\}} \mathbbm{1}\{M_{tj}^{(i)} = 1 \wedge j \in S\}
\end{equation}
for each $t \in [d_i]$, which gives the number of defective items (excluding $i$) that are in the $t$-th test of $\bM^{(i)}$. In accordance with Definition \ref{def:informative_masked}, an item is therefore masked if and only if $U_t \neq \gamma -1$ for all $t \in [d_i]$.

Fix an item $i \in [n]$ and integer $s \in \{0,\dots,d_i\}$, and recall that the event $L_{s,i}$ occurs if exactly $s$ of the tests that $i$ is included in contain at most $\gamma - 2$ defectives, and the remaining $d_i - s$ tests contain at least $\gamma$ defectives. Using the above definitions, we can equivalently express $L_{s,i}$ as the event that there is at least one partition $(A, A^c)$ of $[d_i]$ with $\lvert A \rvert = s$ such that $U_t \leq \gamma - 2$ for all $t \in A$ and $U_t \geq \gamma$ for all $t \in A^c$.  Letting $\mathcal{A}_{s,i}$ denote the collection of all subsets of $[d_i]$ of size $s$, and defining the events
\begin{equation}
    \label{eq:L_plus_minus} L_{A}^- \coloneqq \bigcap_{t \in A} \{U_t \leq \gamma - 2 \}, \qquad   L_A^+ \coloneqq \bigcap_{t \in A^c} \{U_t \geq \gamma \}
\end{equation}
for each $A \in \mathcal{A}_{s,i}$, we can thus write $\Pr(L_{s,i})$ as
\begin{align}
    \Pr(L_{s,i}) &= \Pr \left( \bigcup_{A \in \mathcal{A}_{s,i}} \{ L_{A}^- \cap L_{A}^+ \}\right) \\
    \label{eq:L_si_new_defn} &= \sum_{A \in \mathcal{A}_{s,i}} \Pr ( L_{A}^- \cap L_{A}^+ ),
\end{align}
where \eqref{eq:L_si_new_defn} follows since the events $\{  L_{A}^- \cap L_{A}^+  \}_{A \in \mathcal{A}_{s,i}}$ are mutually disjoint. Note that the test design is fixed, so the probability is only over the randomness in the defective set (under the i.i.d.~prior). To derive a lower bound on $\Pr(L_{s,i})$, we thus begin by lower bounding $\Pr(L_{A}^- \cap L_{A}^+)$ for each $A \in \mathcal{A}_{s,i}$. Writing $\Pr(L_A^- \cap L_A^+) = \Pr(L_A^-) \Pr(L_A^+ \mid L_A^-)$, we proceed to bound each term separately.

{\bf Lower Bounding $\Pr(L_A^-)$.} We begin by lower bounding $\Pr(L_A^-)$. To do so, we note that the events $\{U_t \leq \gamma - 2\}_{t \in A}$ are \emph{decreasing} with respect to the items' defectivity indicator random variables (see Definition \ref{def:inc_dec_rvs}). Thus, the FKG inequality (Lemma \ref{lem:FKG}) implies that 
\begin{equation}
    \label{eq:FKG_minus} \Pr(L_A^-) = \Pr\left(\bigcap_{t \in A} \{U_t \leq \gamma - 2 \}\right) \geq \prod_{t \in A} \Pr(U_t \leq \gamma - 2).
\end{equation}
Since $S$ is generated according to the i.i.d.~prior (for now), and since test $t$ includes $w_t +1$ items, it follows that $U_t \sim {\rm Bin}(q, w_t)$, where we recall that $q = \frac{k}{n}$. Thus, $\Pr(L_A^-)$ can be lower bounded by
\begin{align}
    \label{eq:L_minus_lb_1} \Pr(L_A^-) &\geq \prod_{t \in A} \Pr ({\rm Bin}(w_t,q) \leq \gamma - 2) \\
    \label{eq:L_minus_lb_2} &= \prod_{t \in A} \sum_{j=0}^{\gamma - 2} {w_t \choose j} q^j (1-q)^{w_t -j}.
\end{align}

{\bf Lower Bounding $\Pr(L_A^+ \mid L_A^-).$} We now turn to deriving a lower bound on $\Pr(L_A^+ \mid L_A^-)$. Doing so directly seems difficult, since the conditioning on $L_A^-$ means that items in each test may no longer be defective independently. The solution is to only consider the defectivity status of the items that \emph{only appear in tests} $t \in A^c$. Clearly, if there are at least $\gamma$ defectives among these items, the same applies for the entire test pool. Moreover, since any such item is only included in tests $t \in A^c$, its defectivity status is independent of $L_A^-$, which only depends on the status of items included in tests $t \in A$. Combining these assertions thus implies that
\begin{equation}
    \label{eq:L_plus_lb_init} \Pr( L_A^+ \mid L_A^- ) = \Pr \left(\bigcap_{t \in A^c} \{ U_t \geq \gamma \} \mid L_A^- \right) \geq \Pr\left(\bigcap_{t \in A^c} \{ \widetilde{U}_{t,A} \geq \gamma\}\right),
\end{equation}
where $\widetilde{U}_{t,A} \leq U_t$ is the number of defective items (excluding $i$) that are in test $t \in A^c$ of $\bM^{(i)}$, and that do not participate in any tests in $A$. We can now use a similar argument as in the previous case to further lower bound \eqref{eq:L_plus_lb_init}. The main difference here is that events $\{\widetilde{U}_{t,A} \geq \gamma \}_{t \in A^c}$ are now \emph{increasing} (Definition \ref{def:inc_dec_rvs}). Thus, the FKG inequality (Lemma \ref{lem:FKG}) can once again be applied to obtain
\begin{equation}
    \label{eq:FKG_plus} \Pr\left(\bigcap_{t \in A^c} \{\widetilde{U}_{t,A} \geq \gamma  \}\right) \geq \prod_{t \in A^c} \Pr(\widetilde{U}_{t,A} \geq \gamma).
\end{equation}
Denote the number of items in test $t \in A^c$ that do not appear in any tests in $A$ by $z_{t,A}$. Analogous arguments to those in previous section then imply that $\widetilde{U}_{t,A} \sim {\rm Bin} (z_{t,A}, q)$, which after combining \eqref{eq:L_plus_lb_init} and \eqref{eq:FKG_plus} gives
\begin{align}
    \label{eq:L_plus_lb_final_1} \Pr(L_A^+ \mid L_A^-) &\geq \prod_{t \in A^c} \Pr({\rm Bin}(z_{t,A}, q) \geq \gamma) \\
    \label{eq:L_plus_lb_final_2} &= \prod_{t \in A^c} \left(1- \sum_{j=0}^{\gamma-1} {z_{t,A} \choose j}q^{j}(1-q)^{z_{t,A} - j} \right).
\end{align}


{\bf Combining Terms.} We  obtain a lower bound on $\Pr(L_A^- \cap L_A^+)$ by combining \eqref{eq:L_minus_lb_2} and \eqref{eq:L_plus_lb_final_2}, which yields
\begin{equation}
    \label{eq:L_plus_minus_lb} \Pr(L_A^- \cap L_A^+) \geq \left[\prod_{t \in A} \sum_{j=0}^{\gamma - 2} {w_t \choose j} q^j (1-q)^{w_t -j} \right] \cdot \left[\prod_{t \in A^c} \left(1- \sum_{j=0}^{\gamma-1} {z_{t,A} \choose j}q^{j}(1-q)^{z_{t,A} - j} \right)\right].
\end{equation}
{\bf Summing Over $A \in \mathcal{A}_{s,i}$.} With this lower bound in place, we are now in a position to lower bound $\Pr(L_{s,i})$. Taking the sum over all $A \in \mathcal{A}_{s,i}$, we obtain
\begin{equation}
    \label{eq:Pr_L_si_lb_init} \Pr(L_{s,i}) \geq \sum_{A \in \mathcal{A}_{s,i}} \left[\prod_{t \in A} \sum_{j=0}^{\gamma - 2} {w_t \choose j} q^j (1-q)^{w_t -j} \right] \cdot \left[\prod_{t \in A^c} \left(1- \sum_{j=0}^{\gamma-1} {z_{t,A} \choose j}q^{j}(1-q)^{z_{t,A} - j} \right)\right],
\end{equation}
which can equivalently be written as\footnote{If $\gamma = 1$ then the first term should be omitted in order to avoid taking the logarithm of zero.  The subsequent equations from \eqref{eq:final_L_si_lb} onward still hold for $\gamma = 1$ by following the analysis with the first term omitted (or more simply relying on the standard GT analysis in \cite{coja2020optimal}). \label{foot:r1_note}}
\begin{align}
    \Pr(L_{s,i}) \geq \lvert \mathcal{A}_{s,i}\rvert\frac{1}{\lvert \mathcal{A}_{s,i}\rvert}\sum_{A \in \mathcal{A}_{s,i}} &\exp \Bigg[\sum_{t \in A} \log\Bigg(\sum_{j=0}^{\gamma - 2} {w_t \choose j} q^j (1-q)^{w_t -j}\Bigg) \nonumber \\
    \label{eq:equiv_Pr_L_si} &\qquad + \sum_{t \in A^c} \log \Bigg(1- \sum_{j=0}^{\gamma-1} {z_{t,A} \choose j}q^{j}(1-q)^{z_{t,A} - j}\Bigg)\Bigg].
\end{align}
Using the convexity of the exponential function and Jensen's inequality, this can be further lower bounded by
\begin{align}
    \Pr(L_{s,i}) \geq \lvert \mathcal{A}_{s,i} \rvert &\exp\Bigg[\frac{1}{\lvert \mathcal{A}_{s,i}\rvert} \sum_{A \in \mathcal{A}_{s,i}} \sum_{t \in A} \log\Bigg(\sum_{j=0}^{\gamma - 2} {w_t \choose j} q^j (1-q)^{w_t -j}\Bigg) \nonumber \\
    \label{eq:L_si_jensen}&\qquad + \frac{1}{\lvert \mathcal{A}_{s,i}\rvert}\sum_{A \in \mathcal{A}_{s,i}} \sum_{t \in A^c} \log \Bigg(1- \sum_{j=0}^{\gamma-1} {z_{t,A} \choose j}q^{j}(1-q)^{z_{t,A} - j}\Bigg)\Bigg].
\end{align}
To further lower bound \eqref{eq:L_si_jensen}, we define $z_t \coloneqq w_t - c(d_{\max}-1)$, and observe the following:
\begin{itemize}
    \item By Assumptions \ref{asm:regularity} and \ref{asm:sparsity}, each pair of tests shares at most $c$ items. Consequently, there are at least $w_t - c(d_i-1) \geq w_t - c(d_{\max} - 1) \eqqcolon z_t$ items (excluding $i$) in $t \in A^c$ that do not appear in any of the remaining $d_i-1$ tests $i$ is included in, which implies the same lower bound for $z_{t,A}$ for all $A \in \mathcal{A}_{s,i}$ and $s \in \{0,\dots, d_i\}$.
    \item The inner sum in the second line of \eqref{eq:L_si_jensen} equals $\Pr({\rm Bin}(z_{t,A}, q) \leq \gamma-1)$. This can therefore be upper bounded by $\Pr({\rm Bin}(z_t, q) \leq \gamma - 1)$, since the fact that $z_{t} \leq z_{t,A}$ implies that ${\rm Bin}(z_t, q) \sd {\rm Bin}(z_{t,A},q)$ (recall that $\sd$ means ``stochastically dominated by'').
\end{itemize}
Combining these assertions implies that \eqref{eq:L_si_jensen} can be lower bounded by
\begin{align}
     \Pr(L_{s,i}) \geq \lvert \mathcal{A}_{s,i} \rvert &\exp\Bigg[\frac{1}{\lvert \mathcal{A}_{s,i}\rvert} \sum_{A \in \mathcal{A}_{s,i}} \sum_{t \in A} \log\Bigg(\sum_{j=0}^{\gamma - 2} {w_t \choose j} q^j (1-q)^{w_t -j}\Bigg) \nonumber \\
    \label{eq:L_si_zt}&\qquad + \frac{1}{\lvert \mathcal{A}_{s,i}\rvert}\sum_{A \in \mathcal{A}_{s,i}} \sum_{t \in A^c} \log \Bigg(1- \sum_{j=0}^{\gamma-1} {z_{t} \choose j}q^{j}(1-q)^{z_{t} - j}\Bigg)\Bigg].
\end{align}

Now, to simplify each of the above double sums, we use the following counting argument. For the first double sum, fix a test $t \in [d_i]$, and consider the number of sets $A \in \mathcal{A}_{s,i}$ such that $t \in A$. Since all sets in $\mathcal{A}_{s,i}$ are of size $s$, and since $t \in A$ by assumption, there are ${d_i - 1 \choose s - 1}$ ways of choosing the remaining items. For the second double sum, a similar argument reveals that there are ${d_i -1 \choose s}$ sets $A \in \mathcal{A}_{s,i}$ such that $t \in A^c$ for each $t \in [d_i]$. Using this, we can therefore express \eqref{eq:L_si_zt} as
\begin{align}
    \Pr(L_{s,i}) &\geq \lvert \mathcal{A}_{s,i} \rvert \exp\Bigg[\frac{1}{\lvert \mathcal{A}_{s,i}\rvert} {d_i - 1 \choose s - 1} \sum_{t \in [d_i]} \log\Bigg(\sum_{j=0}^{\gamma - 2} {w_t \choose j} q^j (1-q)^{w_t -j}\Bigg) \nonumber \\
    \label{eq:L_si_double_count_1}&\qquad \qquad \qquad + \frac{1}{\lvert \mathcal{A}_{s,i}\rvert}{d_i -1 \choose s} \sum_{t \in [d_i]} \log \Bigg(1- \sum_{j=0}^{\gamma-1} {z_t \choose j}q^{j}(1-q)^{z_t - j}\Bigg)\Bigg] \\
     \label{eq:L_si_double_count_2} &=  {d_i \choose s} \exp \Bigg[\frac{s}{d_i} \sum_{t \in [d_i]} \log\Bigg(\sum_{j=0}^{\gamma - 2} {w_t \choose j} q^j (1-q)^{w_t -j}\Bigg) \nonumber \\
     &\hspace{105pt}+ \frac{d_i - s}{d_i} \sum_{t \in [d_i]}\log \Bigg(1- \sum_{j=0}^{\gamma-1} {z_t \choose j}q^{j}(1-q)^{z_t - j}\Bigg)\Bigg],
\end{align}
where \eqref{eq:L_si_double_count_2} uses $\lvert \mathcal{A}_{s,i}\rvert = {d_i \choose s}$. Finally, basic algebraic manipulations give a final lower bound of
\begin{equation}
    \label{eq:final_L_si_lb} \Pr(L_{s,i}) \geq {d_i \choose s} \left[\prod_{t \in [d_i]} \sum_{j=0}^{\gamma-2}{w_t \choose j} q^j (1-q)^{w_t -j}\right]^{\frac{s}{d_i}}  \left[\prod_{t \in [d_i]} \left(1 - \sum_{j=0}^{\gamma-1} {z_t \choose j}q^{j}(1-q)^{z_t - j} \right)\right]^{\frac{d_i - s}{d_i}}\hspace{-15pt}.
\end{equation}

\subsubsection{Step 3 (Lower Bounding $\Pr(D_i)$)} \label{sec:step3_single}
We are now in a position to lower bound $\Pr(D_i)$ (i.e., the probability that $i$ is masked). Recalling Definition \ref{def:informative_masked}, an item $i \in [n]$ is masked if there are no tests that $i$ is included in that contain exactly $\gamma -1$ defectives (excluding $i$). It follows that $i$ is masked if and only if $L_{s,i}$ occurs for at least one $s \in \{0,\dots, d_i\}$. Using this and the fact that the events $\{L_{s,i}\}_{s=0}^{d_i}$ are mutually disjoint, $\Pr(D_i)$ can be written as $\Pr(D_i) = \sum_{s=0}^{d_i} \Pr(L_{s,i})$. Substituting in the lower bound on $\Pr(L_{s,i})$ for each $s \in \{0,\dots,d_i\}$ and using the binomial formula, this can be lower bounded by
\begin{align}
    \label{eq:Pr_Di_lb_1} \Pr(D_i) &\geq \sum_{s=0}^{d_i} {d_i \choose s} \left[\prod_{t \in [d_i]} \sum_{j=0}^{\gamma-2}{w_t \choose j} q^j (1-q)^{w_t -j}\right]^{\frac{s}{d_i}}  \left[\prod_{t \in [d_i]} \left(1 - \sum_{j=0}^{\gamma-1} {z_t \choose j}q^{j}(1-q)^{z_t - j} \right)\right]^{\frac{d_i - s}{d_i}} \\
    \label{eq:Pr_Di_lb_2} &= \left[ \left(\prod_{t \in [d_i]} \sum_{j=0}^{\gamma-2}{w_t \choose j} q^j (1-q)^{w_t -j}\right)^{\frac{1}{d_i}}  + \left(\prod_{t \in [d_i]} \left(1 - \sum_{j=0}^{\gamma-1} {z_t \choose j}q^{j}(1-q)^{z_t - j} \right)\right)^{\frac{1}{d_i}}\right]^{d_i}\hspace{-5pt}.
\end{align}
To further lower bound this, we focus our attention on each of the sums in \eqref{eq:Pr_Di_lb_2}. Namely, we wish to lower bound the summation in the first term, and upper bound the summation in the second term, for each $t \in [d_i]$. We begin with the former. To do so, note that this sum equals $\Pr({\rm Bin}(w_t, q) \leq \gamma -2)$, which implies that it is lower bounded by $\Pr({\rm Bin}(w_{\max}, q )\leq \gamma -2 )$, since ${\rm Bin}(w_t, q) \sd {\rm Bin}(w_{\max}, q)$ for all $t \in [d_i]$. Thus, it follows that
\begin{equation}
    \label{eq:first_sum_lb} \sum_{j=0}^{\gamma - 2}{w_t \choose j}q^j(1-q)^{w_t - j} \geq \sum_{j=0}^{\gamma - 2} {w_{\max} \choose j}q^j(1 -q)^{w_{\max} - j}
\end{equation}
for all $t \in [d_i]$. Using an analogous argument for the second sum in \eqref{eq:Pr_Di_lb_2} gives an upper bound of
\begin{equation}
    \label{eq:second_sum_lb} \sum_{j=0}^{\gamma-1} {z_t \choose j}q^{j}(1-q)^{z_t - j} \leq \sum_{j=0}^{\gamma-1} {z_{\min} \choose j}q^{j}(1-q)^{z_{\min} - j},
\end{equation}
for all $t \in [d_i]$ where $z_{\min} := w_{\min} - c(d_{\max}-1)$ (with $w_{\min} \coloneqq \min_{t \in [T]}w_t$) is a uniform lower bound on $z_t$. Substituting these bounds back into \eqref{eq:Pr_Di_lb_2} and simplifying, alongside using $d_i \leq d_{\max}$, we obtain
\begin{equation}
    \label{eq:Pr_Di_lb_3} \Pr(D_i) \geq \left[\sum_{j=0}^{\gamma-2}{w_{\max} \choose j}q^j(1 -q)^{w_{\max} - j} + 1 - \sum_{j=0}^{\gamma-1} {z_{\min} \choose j}q^{j}(1-q)^{z_{\min}- j}\right]^{d_{\max}}\hspace{-15pt}.
\end{equation}
To obtain our final bound on $\Pr(D_i)$ for all $i \in [n]$, we derive an upper bound on $d_{\max}$ in terms of $T$. A double counting argument reveals that $\sum_{i=1}^n d_i = \sum_{t=1}^T (w_t+1)$, which in turn implies that $n d_{\min} \leq T (w_{\max}+1)$, with $d_{\min} \coloneqq \min_{i \in [n]}d_i$. Using Assumption \ref{asm:regularity}, we also have the bound $n d_{\min} \geq \frac{n d_{\max}}{\zeta}$. Combining these, we obtain a final lower bound for $\Pr(D_i)$ of
\begin{equation}
    \label{eq:Pr_Di_lb_final} \Pr(D_i) \geq \left[\sum_{j=0}^{\gamma-2}{w_{\max} \choose j}q^j(1 -q)^{w_{\max} - j} + 1 - \sum_{j=0}^{\gamma-1} {z_{\min} \choose j}q^{j}(1-q)^{z_{\min}- j}\right]^{\frac{\zeta T(w_{\max}+1)}{n}}\hspace{-25pt}.
\end{equation}

\subsubsection{Step 4 (Lower Bounding the Number of Masked Items, i.i.d.~Prior)} \label{sec:step4_single}
Given the set $\mathcal{I}$ and the lower bound on $\Pr(D_i)$ for all $i \in [n]$, we are now in a position to derive a high-probability lower bound on the number of masked items, which we denote by $M$. Recall that by construction, $\mathcal{I}$ has size at least $n^{1-3\xi}$, and the items in $\mathcal{I}$ have mutually independent masking events. Using this and the lower bound on $\Pr(D_i)$, it follows that ${\rm Bin}(n^{1-3\xi}, \Xi^{\frac{\zeta T (w_{\max}+1)}{n}}) \sd M$, where $\Xi \in [0,1]$ is the term inside the exponent in \eqref{eq:Pr_Di_lb_final}. The expected value of the binomial random variable is given by
\begin{equation}
    \label{eq:avg_num_binom} n^{1-3\xi} \Xi^{\frac{\zeta T (w_{\max}+1)}{n}} = n^{1-3\xi} \exp\left( \frac{\zeta T (w_{\max}+1)}{n} \log (\Xi)\right),
\end{equation}
which is at least $n^{2\xi}$ whenever 
\begin{equation}
    \label{eq:n_2xi_avg_T} T \leq \frac{n \log n}{-\zeta (w_{\max}+1) \log (\Xi)}(1 -  5\xi).
\end{equation}
Substituting the definition of $\Xi$ and writing $n = k^{\frac{1}{1-\xi}}$, this condition becomes
\begin{equation}
    \label{eq:Lambda_def_avg_T} T \leq \frac{n\log k}{\displaystyle-\zeta (w_{\max}+1) \log \left(\sum_{j=0}^{\gamma-2}{w_{\max} \choose j}q^j(1 -q)^{w_{\max} - j} + 1 - \sum_{j=0}^{\gamma-1} {z_{\min} \choose j}q^{j}(1-q)^{z_{\min}- j}\right)}(1-\varepsilon),
\end{equation}
where $\varepsilon > 0$ is an arbitrarily small constant (depending on arbitrarily small $\xi > 0$). To further simplify this, recall that $\frac{c \cdot d_{\max}}{w_{\min}} = o(1)$ via Assumption \ref{asm:sparsity}, which in turn implies that $z_{\min} = w_{\min}(1 - o(1))$. Additionally, we have due to Assumption \ref{asm:regularity} that $w_{\min}=\frac{\Delta n}{k}(1-o(1))$, $w_{\max}+1 = \frac{\Delta n}{k}(1+o(1))$ and $\zeta = 1+o(1)$, which combined with the fact that $q = \frac{k}{n}$ means we can apply a Poisson approximation (Lemma \ref{lem:poisson_approx}) to obtain
\begin{align}
    \label{eq:simplified_avg_T_1} T &\leq \frac{k \log k}{\displaystyle- \Delta\log \left(\sum_{j=0}^{\gamma - 2}\frac{\Delta^j e^{-\Delta}}{j!} + 1 - \sum_{j=0}^{\gamma - 1} \frac{\Delta^je^{-\Delta}}{j!}\right)}(1-\varepsilon)(1-o(1)) \\
    \label{eq:simplified_avg_T_2} &= \frac{k \log k}{\displaystyle- \Delta \log\left(1 - \frac{\Delta^{\gamma - 1}e^{-\Delta}}{(\gamma - 1)!}\right)}(1-\varepsilon)(1-o(1)),
\end{align}
since the sums in the denominator of \eqref{eq:Lambda_def_avg_T} each correspond to a binomial CDF. Absorbing the $1-o(1)$ term into $\varepsilon$ (when $n$ is large enough) and minimizing over $\Delta$, it follows that \eqref{eq:avg_num_binom} is at least $n^{2\xi}$ for \emph{all} $\Delta > 0$ as long as
\begin{equation}
    \label{eq:optimised_avg_T} T \leq \min_{\Delta > 0} \frac{k \log k}{\displaystyle- \Delta\log\left(1 - \frac{\Delta^{\gamma - 1}e^{-\Delta}}{(\gamma - 1)!}\right)}(1-\varepsilon).
\end{equation}

If \eqref{eq:optimised_avg_T} holds, then the Chernoff bound (Lemma \ref{lem:chernoff}) alongside the fact that $ {\rm Bin}(n^{1-3\xi}, \Xi^{\frac{\zeta T (w_{\max}+1)}{n}}) \sd M$ gives
\begin{equation}
    \label{eq:total_masked_lb_whp}\Pr\left(M \leq \frac{1}{2}n^{2\xi}\right) \leq  \Pr\left({\rm Bin}(n^{1-3\xi}, \Xi^{\frac{\zeta T (w_{\max}+1)}{n}}) \leq \frac{1}{2}n^{2\xi}\right) \leq \exp(-\Omega(n^{2\xi})).
\end{equation}

Conditioned on $M \geq \frac{1}{2}n^{2\xi}$, we now seek to obtain lower bounds on $M_0$ and $M_1$, the number of \emph{masked non-defective} and \emph{masked defective} items respectively. We begin with $M_0$. Since a given masked item's defectivity status is conditionally independent of the test outcomes (given the masking event), its posterior probability of being defective matches its prior (i.e., it remains defective with probability $q = \frac{k}{n}$), which can easily be seen via an application of Bayes' theorem (and is used in \cite{coja2020optimal,bay2022optimal}). Moreover, the independent masking events of each $i \in \mathcal{I}$ imply that each of the masked items in $\mathcal{I}$ are defective independently of one another. Thus, we obtain ${\rm Bin}(\frac{1}{2}n^{2\xi}, 1-q) \sd (M_0 \mid M \geq \frac{1}{2}n^{2\xi})$, which in turn implies via the Chernoff bound that
\begin{equation}
    \label{eq:M_0_lb} \Pr\left(M_0 \leq \frac{1}{4}n^{2\xi} \mid M \geq \frac{1}{2}n^{2\xi}\right) \leq \Pr\left({\rm Bin}\Big(\frac{1}{2}n^{2\xi}, 1 - q\Big) \leq \frac{1}{4}n^{2\xi}\right) \leq \exp(-\Omega(n^{2\xi})).
\end{equation}
Using a similar argument for $M_1$, the probability of there being no masked defectives can be upper bounded by
\begin{equation}
    \label{eq:Pr_no_masked_defs} \Pr\left(M_1 = 0 \mid M \geq \frac{1}{2}n^{2\xi}\right) \leq \left(1 - \frac{k}{n}\right)^{\frac{1}{2}n^{2\xi}} = \left(1 -\frac{1}{n^{\xi}}\right)^{\frac{1}{2}n^{2\xi}} \leq e^{-\frac{1}{2}n^{\xi}}.
\end{equation}
 Defining $\mathcal{E} = \{M_0 \leq \frac{1}{4}n^{2\xi}\} \cup \{M_1 = 0\}$, we can use the law of total probability and \eqref{eq:total_masked_lb_whp}, alongside a union bound, to obtain
\begin{align}
    \label{eq:Pr_E_lb_1} \Pr(\mathcal{E}) &\leq \Pr\left(\mathcal{E} \mid M \geq \frac{1}{2}n^{2\xi} \right) + \exp(-\Omega(n^{2\xi})) \\
    \label{eq:Pr_E_lb_2} &\leq \Pr\left(M_0 \leq \frac{1}{4}n^{2\xi} \mid M \geq \frac{1}{2}n^{2\xi}\right) + \Pr\left(M_1 = 0 \mid M \geq \frac{1}{2}n^{2\xi}\right) + \exp(-\Omega(n^{2\xi})) \\
    \label{eq:Pr_E_lb_3} &= \exp(-\Omega(n^\xi)).
\end{align} 

\subsubsection{Step 5 (Transferring from i.i.d.~Prior to Combinatorial Prior)} \label{sec:step5_single}
In the previous step, we showed that under the i.i.d.~prior, there were at least $\frac{1}{4}n^{2\xi}$ masked non-defectives and at least one masked defective with probability $1 - \exp(-\Omega(n^\xi))$. To obtain a similar guarantee under the combinatorial prior, we use the fact that for any event $\mathcal{F}$,
\begin{align}
    \label{eq:iid_comb_relation_1} \Pr\nolimits_{\rm c}(\mathcal{F}) &= \Pr\nolimits_{{{\rm i.i.d.}}}(\mathcal{F} \mid \{\lvert S \rvert = k\}) \\
    \label{eq:iid_comb_relation_2} &\leq \frac{\Pr_{\rm i.i.d.}(\mathcal{F})}{\Pr_{\rm i.i.d.}(\lvert S \rvert = k)},
\end{align}
where $\Pr_{\rm c}$ and $\Pr_{\rm i.i.d.}$ are the probability measures under the combinatorial and i.i.d.~priors respectively. Applying this with the event $\mathcal{E}$ and using its aforementioned upper bound under the i.i.d.~prior, its probability under the combinatorial prior is upper bounded by
\begin{equation}
    \label{eq:Pr_E_comb_ub_init} \Pr\nolimits_{\rm c}(\mathcal{E}) \leq \frac{\exp(-\Omega(n^\xi))}{\Pr_{\rm i.i.d.}(\lvert S \rvert = k)},
\end{equation}
which means that showing $\Pr_{\rm c}(\mathcal{E}) = o(1)$ amounts to showing that $\Pr_{\rm i.i.d.}(\lvert S \rvert = k)$ is ``not too small''. Under the i.i.d.~prior with $q=\frac{k}{n}$, this probability is given by
\begin{equation}
    \label{eq:Pr_iid_S_eq_k} \Pr\nolimits_{\rm i.i.d.}(\lvert S \rvert = k) = {n \choose k}\left(\frac{k}{n}\right)^{k}\left(1 - \frac{k}{n}\right)^{n-k} = {n \choose k} \exp(-nH(k/n)),
\end{equation}
where $H(\cdot)$ denotes the binary entropy function measured in nats. Using Lemma \ref{lem:binom_bounds}, this can be lower bounded by
\begin{align}
    \label{eq:S_eq_k_lb_1} \Pr\nolimits_{\rm i.i.d.}(\lvert S \rvert = k) &\geq \frac{\exp(n H(k/n))}{2\sqrt{2k(1-k/n)}} \exp(-nH(k/n)) \\
    \label{eq:S_eq_k_lb_2} &\geq \frac{1}{2 \sqrt{2k}}.
\end{align}
Thus, $\Pr_{\rm c}(\mathcal{E})$ can be upper bounded by
\begin{equation}
    \label{eq:Pr_E_comb_ub_final} \Pr\nolimits_{\rm c}(\mathcal{E}) \leq 2\sqrt{2k} \exp(-\Omega(n^\xi)) = o(1),
\end{equation}
which gives the desired result.

\subsubsection{Step 6 (Lower Bounding the Error Probability, High-$\theta$ Regime)} \label{sec:step6_single}
Having established that $\mathcal{E}^c$ holds with probability $1-o(1)$ under the combinatorial prior, we are now able to derive a lower bound on the error probability when $k = n^{1-\xi}$. To do so, we first introduce the notion of \emph{satisfying sets:}
\begin{definition}
    \label{def:satisfying} Fix a test design $\bX$ and its induced test outcome vector $\by$. A set $s \subset [n]$ of size $\lvert s \rvert = k$ is called \emph{satisfying} with respect to $(\bX, \by)$ if the observed test outcomes $\by$ match the test outcomes that would be generated by $\bX$ if $s$ were the true defective set.
\end{definition}

Before proceeding, we state some known facts:
\begin{itemize}
    \item Denoting $\mathcal{V}(\bX, \by)$ as the collection of all satisfying sets of size $k$, the posterior distribution of $S$ given $(\bX, \by)$ is uniform over $\mathcal{V}(\bX, \by)$ (e.g.,  see \cite[Fact 3.1]{coja2020optimal}).\footnote{While this is stated in \cite{coja2020optimal} for the classical group testing problem, the exact same argument applies for our setting, so we omit the details.}
    \item The \emph{maximum a posteriori} (MAP) decoder is optimal among all decoders; in view of the first dot point, it returns an arbitrary (or possibly randomly-chosen) element of $\mathcal{V}(\bX, \by)$, thus being correct with conditional probability $\frac{1}{|\mathcal{V}(\bX, \by)|}$ (given $\mathcal{V}(\bX, \by)$).
    \item It was also shown in \cite{coja2020optimal} that $\lvert \mathcal{V}(\bX, \by)\rvert \geq M_0 \cdot M_1$; this is because the true defective set $S$ is satisfying by definition, and a new satisfying set $S'$ can be formed by swapping out a single masked defective item with a masked non-defective item. 
    Since there are $M_0 \cdot M_1$ ways of performing such a swap, the claim follows.
\end{itemize}
Combining these observations, the success probability of the MAP decoder (and thus any decoder $\hat{S}$) is upper bounded by
\begin{equation}
    \label{eq:succ_prob_ub} \Pr(\hat{S} = S \mid (M_0,M_1) = (m_0,m_1)) \leq \frac{1}{\lvert \mathcal{V}(\bX, \by)\rvert } \leq \frac{1}{m_0 \cdot m_1}.
\end{equation}
Conditioned on $\mathcal{E}^c$, we have $M_1 \geq 1$ and $M_0 \geq \frac{1}{4}n^{2\xi}$ by definition. Combining this with law of total probability and the fact that $\Pr(\mathcal{E}) = o(1)$, the success probability of any decoder is upper bounded by
\begin{align}
    \label{eq:succ_prob_E_1} \Pr(\hat{S} = S) &= \Pr(\{\hat{S} = S\} \cap \mathcal{E}) + \Pr(\{\hat{S} = S\} \cap \mathcal{E}^c) \\
    \label{eq:succ_prob_E_2} &\leq  o(1) +\frac{4}{n^{2\xi}} \\
    \label{eq:succ_prob_E_3} &= o(1).
\end{align}
Thus, any decoder fails with probability approaching $1$ in the ``dense'' regime $k = n^{1-\xi}$ whenever $T$ satisfies \eqref{eq:optimised_avg_T}.

\subsubsection{Step 7 (Transferring from High $\theta$ to Low $\theta$)} \label{sec:step7_single}
In the previous step, we showed that under the combinatorial prior, any decoder fails with probability $1-o(1)$ when $k = n^{1-\xi}$ (i.e., when $\theta$ is arbitrarily close to one). To conclude the proof of Theorem \ref{thm: sing}, it remains to transfer this converse result from $\theta$ arbitrarily close to one to any $\theta \in (0,1)$. In what follows, we will denote quantities in the ``dense'' regime with a prime superscript (e.g., $n', \theta',\bX'$), and quantities in the ``sparse'' regime without this (e.g., $n, \theta, \bX$). Quantities that remain the same in both regimes (e.g., $k$) will retain their original notation. Thus, the test design $\bX$ in the previous steps is now written as $\bX'$, and the relationship between number of defectives $k$ and the population size in the dense regime is $k = (n')^{\theta'} = (n')^{1-\xi}$. To establish the reduction from the denser regime to the sparser regime, we follow the argument of \cite{coja2020optimal} (see also \cite[App. C]{mcmorrow2026optimal}), who proved the result via its contrapositive: an achievability result for a given $\theta$ implies the same for any larger $\theta' > \theta$, provided the two problems share the same $k$ but have different population sizes, and the number of tests scales as $T = c k \log k$ for some $c > 0$. This can be done by starting with a test design $\bX \in \{0,1\}^{T\times n}$ for the sparse regime with $n$ items and $k = \Theta(n^\theta)$ defectives, and constructing a test design $\bX' \in \{0,1\}^{T\times n'}$ for the dense regime with $n' \ll n$ items and $k = (n')^{\theta'}$ defectives for $\theta' > \theta$ as follows:
\begin{enumerate}
    \item Choose a subset $W \subset [n]$ of size $\lvert W \rvert = n'$ uniformly at random from the ${n \choose n'}$ possibilities, and declare the items in $[n] \setminus W$ to be (dummy) non-defectives.
    \item Map the items in $[n']$ to $W$ in an arbitrary manner (e.g., in numerical order).
    \item Construct $\bX' \in \{0,1\}^{T \times n'}$ by taking it as the sub-matrix of $\bX$ with the columns indexed by $W$.
    \item Perform the tests according to $\bX'$ and observe the induced test outcome result vector $\by'$.
\end{enumerate}
Observe that since the items in $[n] \setminus W$ are non-defective by construction, the test outcome vectors $\by$ and $\by'$ generated by $\bX$ and $\bX'$ respectively are identical. Additionally, the symmetry of the construction of $\bX'$ implies that if there are $k$ defectives in $[n']$ uniformly at random, the same applies to $[n]$. Thus, any decoder that achieves $\Pr(\hat{S} \neq S) \to 0$ using the test design $\bX$ also does so using $\bX'$, which gives the desired result. 

To conclude the proof, we show that constructing $\bX'$ in the manner described above leads to Assumptions \ref{asm:regularity} and \ref{asm:sparsity} still holding with probability $1-o(1)$ for $\theta'$ sufficiently close to $1$ (assuming $\bX$ satisfies them), as the previous steps relied on these assumptions. Since $\bX'$ is formed by taking a subset of the columns of $\bX$, the values of $d_i$ in the columns in $\bX'$ are equivalent to their corresponding columns in $\bX$, meaning that if the values of $d_i$ in $\bX$ satisfy Assumption \ref{asm:regularity}, the same applies for the corresponding values of $d_i$ in $\bX'$. Thus, to show that both assumptions still hold, we need to show (with high probability) that the test weights and pairwise test overlaps in $\bX'$ satisfy the conditions in Assumptions 
\ref{asm:regularity}-\ref{asm:sparsity}, which can be done via the following concentration arguments:

{\bf Concentration of Test Weights.} Fix a test $t \in \{1,\dots,T\}$ with test weight $w_t+1 = \frac{\Delta_t n}{k}$, where $\Delta_t \in [\Delta - o(1), \Delta + o(1)]$ due to Assumption \ref{asm:regularity}. Since $\bX'$ is formed by choosing $n'$ columns uniformly at random from $\bX$, the (random) test weight of the $t$-th test in $\bX'$ is distributed as $W_t' \sim {\rm HypGeo}(n, \frac{\Delta_t n}{k}, n')$, with mean $w_t' \coloneqq \mathbb{E}[W_t'] = \frac{\Delta_t n'}{k}$ (for notational convenience, we denote the test weight in the reduced test matrix as $W_t'$ rather than $W_t' + 1$ here; the shift by 1 is a lower-order asymptotic term). To obtain a high-probability bound on $W_t$, we use the fact that the Chernoff bound (Lemma \ref{lem:chernoff}) also holds for hypergeometric random variables to obtain
\begin{equation}
    \label{eq:w_t'_conc} \Pr(\lvert W_t' - w_t'\rvert \geq \delta w_t') \leq 2\exp\left(-\frac{\delta^2 w_t'}{3}\right)
\end{equation}
for any $\delta \in (0,1)$. Since $w_t' = n^{\Omega(1)}$, the choice $\delta = \frac{1}{\log n}$ ensures that $W_t' \notin (w_t'(1- o(1)), w_t'(1+o(1)))$ with probability $\exp(-n^{\Omega(1)})$. Since this holds for any $t \in [T]$, a union bound over all $T = O(k \log n)$ tests implies that the probability that $W_t' \notin (w_t'(1-o(1)), w_t'(1+o(1))$ for \emph{any} $t \in [T]$ is also $\exp(-n^{\Omega(1)})$. Substituting in the value of $w_t'$, we obtain
\begin{equation}
    \label{eq:W_t'_range_whp} W'_t \in \left(\frac{\Delta_t n'}{k}(1-o(1)), \frac{\Delta_t n'}{k}(1+o(1))\right), \ \forall t \in [T]
\end{equation}
with probability $1-o(1)$. Finally, the fact that $\Delta_t \in [\Delta  - o(1), \Delta + o(1)]$ for all $t$ implies that
\begin{equation}
    \label{eq:w_t'_min_max} \min_{t \in [T]} W_t'= \frac{\Delta n'}{k}(1-o(1)), \quad \max_{t \in [T]} W_t' = \frac{\Delta n'}{k}(1+o(1))
\end{equation}
with probability $1-o(1)$, meaning that Assumption \ref{asm:regularity} is satisfied for $\bX'$.

{\bf Concentration of Test Overlaps.} Fix a pair of tests $(i,j) \in [T]^2$, and denote by $c_{ij}$ the number of items appearing in both tests in $\bX$. By a similar argument to the previous case, the (random) number of items appearing in both tests in $\bX'$ is distributed as $C_{ij}' \sim {\rm HypGeo}(n, c_{ij}, n')$. Since Assumption \ref{asm:sparsity} gives $c_{ij} \leq c$ for all $(i,j) \in [T]^2$, it follows that $C_{ij}' \sd {\rm HypGeo}(n, c, n')$.  Recalling the assumption $c = o\big( \frac{n}{k \log n} \big)$, it follows that $\mu_{ij} := \mathbb{E}[C'_{ij}] \le \frac{c n'}{n} = o\big( \frac{n'}{k \log n} \big)$.  Since $k=(n')^{1-\xi}$ and $\log n' = \Theta(\log n)$, it follows that the quantity
\begin{equation}
    b_n:=\frac{n'}{k\log n'}=\frac{(n')^\xi}{\log n'}
\end{equation}
grows polynomially in \(n'\). We fix an arbitrary sequence \(a_n\) satisfying both $a_n\gg \log n$ and
\begin{equation}
    \max_{i,j}\mu_{ij}\ll a_n\ll b_n.
\end{equation}
Then, applying the Chernoff bound for the hypergeometric distribution, this time in the strengthened form
\begin{equation}
    \Pr(C'_{ij}\ge a_n)
    \le \left(\frac{e\mu_{ij}}{a_n}\right)^{a_n},
\end{equation} 
we obtain
\begin{equation}
    \Pr(C'_{ij}\ge a_n)\le \exp(-\omega(\log n)).
\end{equation}
Taking a union bound over \(O(T^2)=O(k^2\log^2 n)\) pairs of tests then gives the following for all $(i,j)$ with probability $1-o(1)$:
\begin{equation}
    \max_{i\neq j} C'_{ij}\le a_n = o\Big(\frac{n'}{k\log n'}\Big).
\end{equation}
Thus, Assumption \ref{asm:sparsity} is also satisfied for $\bX'$.

\subsubsection{Wrapping Up}
Combining the findings from Steps \hyperref[sec:step1_single]{1}-\hyperref[sec:step7_single]{7}, we can conclude that under the combinatorial prior, the error probability tends to $1$ for any decoder $\hat{S}$, sparsity parameter $\theta$, and test design $\bX$ satisfying Assumptions \ref{asm:regularity}-\ref{asm:sparsity}, as long as \eqref{eq:optimised_avg_T} holds. Thus, to ensure that $\Pr(\hat{S} \neq S) \nrightarrow 1$, we require that
\begin{equation}
    \label{eq:required_T_single} T \geq \min_{\Delta > 0} \frac{k \log k}{-\displaystyle\Delta \log \left(1 - \frac{\Delta^{\gamma - 1}e^{-\Delta}}{(\gamma - 1)!}\right)}(1-o(1)).
\end{equation}
This completes the proof of Theorem \ref{thm: sing}.

\subsection{Proof of Theorem~\ref{thm: multi} (Converse for the Selectable-Threshold Setting)} \label{app:pf_conv_multi}

\subsubsection{Outline of Modifications} \label{sec:modifications}
When handling the more general GT-ST problem, we can broadly use the same proof template, with many of the steps being exactly the same as in the single-threshold setting. In particular, we follow a similar 7-step roadmap to the one detailed in Section \ref{sec:roadmap}. Of these steps, we outline those that proceed in the same manner as before, and those that require modification:
\begin{itemize}
    \item Steps \hyperref[enum:step1_roadmap]{1} and \hyperref[enum:step5_roadmap]{5}-\hyperref[enum:step7_roadmap]{7} proceed identically as in the single-threshold converse, both in their description and analysis. We therefore omit the details here and refer the reader to Section \ref{sec:roadmap} for their descriptions, and Sections \ref{sec:step1_single} and \ref{sec:step5_single}-\ref{sec:step7_single} for their analyses.
    \item The description of Step \hyperref[enum:step4_roadmap]{4} is the same in both cases; however, the subsequent analyses differ. Thus, we include a detailed analysis of this step for the selectable-threshold case in Section \ref{sec:step4_multi}, while not repeating its description.
    \item Adapting Steps \hyperref[enum:step2_roadmap]{2} and \hyperref[enum:step3_roadmap]{3} to the selectable-threshold setting requires modifications of both their descriptions and subsequent analyses.
\end{itemize}
We now outline these modifications made to Steps \hyperref[enum:step2_roadmap]{2} and \hyperref[enum:step3_roadmap]{3}:
\begin{itemize}
    \item[2.] Fix $i \in \mathcal{I}$ and $\mathbf{s} = (s_{1}, \dots, s_{\ell})\in \mathcal{S}_i \coloneqq \{0,\dots, d_{i,1}\} \times \cdots \times \{0,\dots, d_{i,\ell}\}$, where for each $r \in [\ell]$, $d_{i,r}$ is the number of tests with threshold $\gamma_r$ that item $i$ appears in. We seek to lower bound $\Pr(L_{\mathbf{s},i})$, where $L_{\mathbf{s},i}$ is the event that the following both occur for all $r \in [\ell]$: (i) exactly $s_r$ of the tests with threshold $\gamma_r$ that $i$ is included in contain at most $\gamma_r - 2$ other defectives; and (ii) the remaining $d_{i,r} - s_r$ tests with threshold $\gamma_r$ that $i$ is included in contain at least $\gamma_r$ other defectives.
    \item[3.] In the selectable-threshold setting, the probability of $D_i$ (i.e., the event that $i$ is masked) can then be written as $\Pr(D_i) = \sum_{\mathbf{s} \in \mathcal{S}_i}\Pr(L_{\mathbf{s},i})$, analogous to its definition in Step \hyperref[enum:step3_roadmap]{3} in the single-threshold setting (which involved a sum of the probabilities of $L_{s,i}$ over $s \in \{0,\dots, d_i\}$). We can then use the previously derived lower bounds on $\Pr(L_{\mathbf{s},i})$ to lower bound $\Pr(D_i)$, similarly to the single-threshold setting.
\end{itemize}
The above changes can be seen as a direct generalization of Steps \hyperref[enum:step2_roadmap]{2} and \hyperref[enum:step3_roadmap]{3} to the selectable-threshold setting, so one would expect their respective analyses to exhibit a similar form to their single-threshold counterparts. Indeed, this is the case, with the details provided below.

\subsubsection{Step 2: Lower Bounding $\Pr(L_{\mathbf{s}, i})$} \label{sec:step2_multi}
As discussed above, the approach here is similar to that of Section \ref{sec:step2_single}, and can be seen as a direct generalization of the methods therein. Namely, for a given $i \in [n]$ and $\mathbf{s} = (s_1,\dots, s_\ell) \in \mathcal{S}_i$, we wish to derive a lower bound on the probability that among the tests with threshold $\gamma_r$ in which $i$ is included in, exactly $s_r$ tests contain at most $\gamma_r - 2$ other defectives, with the remaining $d_{i,r} -s_r$ tests containing at least $\gamma_r$ other defectives, for all $r \in [\ell]$. To formalize this, we first introduce some definitions, analogous to those given in Section \ref{sec:step2_single}:
\begin{itemize}
    \item We define $\bM^{(i)}_r \in \{0,1\}^{d_{i,r}\times (n-1)}$ similarly to $\bM^{(i)}$, except that $\bM_r^{(i)}$ only consists of tests that $i$ is included in \emph{with threshold} $\gamma_r$, rather than all tests $i$ is included in.  Thus, each row represents a test including item $i$ with threshold $\gamma_r$, with its $n-1$ entries indicating the presence/absence of the remaining items. 
    \item Similarly, the $t_r$-th defective weight of $\bM_r^{(i)}$, denoted by $U_{t_r}$, is the number of defective items (excluding $i$) that appear in test $t_r \in [d_{i,r}]$ with threshold $\gamma_r$.
    \item We define $\mathcal{A}_{s_r,i}$ as the collection of all subsets of $[d_{i,r}]$ of size $s_r$.
    \item For a given $r \in [\ell]$ and $A_r \in \mathcal{A}_{s_r,i}$, the events $L_{A_r}^-, L_{A_r}^+$ are defined analogously to \eqref{eq:L_plus_minus} as follows:
    \begin{equation}
        \label{eq:L_plus_minus_multi} L_{A_r}^- \coloneqq \bigcap_{t_r \in A_r} \{U_{t_r} \leq \gamma_r - 2 \}, \qquad   L_{A_r}^+ \coloneqq \bigcap_{t_r \in A_r^c} \{U_{t_r} \geq \gamma_r \},
    \end{equation}
    where $A_r^c \coloneqq [d_{i,r}] \setminus A_r$.  
    In particular, we note that, like their single-threshold counterparts, $L_{A_r}^-, L_{A_r}^+$ are intersections of decreasing and increasing events, respectively, for all $r \in [\ell]$.
    \item Furthermore, we define $\mathcal{A}_{\mathbf{s},i} = \mathcal{A}_{s_1,i} \times \cdots \times \mathcal{A}_{s_\ell,i}$, and let
    \begin{equation}
        \label{eq:L_bfA_plus_minus} L_{\mathbf{A}}^- \coloneqq \bigcap_{r \in [\ell]} L_{A_r}^-, \qquad L_{\mathbf{A}}^+ \coloneqq \bigcap_{r \in [\ell]} L_{A_r}^+
    \end{equation}
    for each $\mathbf{A} \in \mathcal{A}_{\mathbf{s}, i}$.
\end{itemize}
Given these definitions, we can thus express $\Pr(L_{\mathbf{s},i})$ as
\begin{align}
    \Pr(L_{\mathbf{s},i}) &= \Pr\left(\bigcup_{\mathbf{A} \in \mathcal{A}_{\mathbf{s},i}} \{L_{\mathbf{A}}^- \cap L_{\mathbf{A}}^+ \}\right) \\
    \label{eq:Pr_L_bfs_i} & =\sum_{\mathbf{A} \in \mathcal{A}_{\mathbf{s},i}} \Pr(L_{\mathbf{A}}^- \cap L_\mathbf{A}^+),
\end{align}
with \eqref{eq:Pr_L_bfs_i} following since the events $\{L_{\mathbf{A}}^- \cap L_{\mathbf{A}}^+\}_{\mathbf{A} \in \mathcal{A}_{\mathbf{s},i}}$ are mutually disjoint. As in Section \ref{sec:step2_single}, we first focus on bounding $\Pr(L_{\mathbf{A}}^- \cap L_{\mathbf{A}}^+) = \Pr(L_{\mathbf{A}}^-)\Pr(L_\mathbf{A}^+ \mid L_{\mathbf{A}}^-)$ for each $\mathbf{A} \in \mathcal{A}_{\mathbf{s},i}$.

{\bf Lower Bounding $\Pr(L_{\mathbf{A}}^-)$.} Since, as noted above, the events $L_{A_r}^-$ are intersections of decreasing events for all $r \in [\ell]$, the FKG inequality (Lemma \ref{lem:FKG}) gives
\begin{equation}
    \label{eq:L_bfA_minus_lb} \Pr(L_{\mathbf{A}}^-) \geq \prod_{r \in [\ell]} \prod_{t_r \in A_r} \Pr(U_{t_r} \leq \gamma_r - 2).
\end{equation}
Analogously to the single-threshold setting, each $U_{t_r}$ has distribution $U_{t_r} \sim {\rm Bin}(w_{t_r}, q)$, meaning \eqref{eq:L_bfA_minus_lb} can be written as
\begin{equation}
    \label{eq:Pr_L_bfA_minus_bino} \Pr(L_{\mathbf{A}}^-) \geq \prod_{r \in [\ell]} \prod_{t_r \in A_r} \sum_{j=0}^{\gamma_r - 2} {w_{t_r} \choose j} q^j(1-q)^{w_{t_r} - j}.
\end{equation}
Note that if $\gamma_r = 1$ for some $r \in [\ell]$, then we necessarily have $s_r = 0$, and the corresponding ``empty product'' $\prod_{t_r \in A_r} (\cdot)$ in \eqref{eq:Pr_L_bfA_minus_bino} is interpreted as equaling 1 (i.e., that $r$ value is effectively excluded).  See also Footnote \ref{foot:r1_note}.

{\bf Lower Bounding $\Pr(L_{\mathbf{A}}^+ \mid L_{\mathbf{A}}^-)$.} As in the single-threshold setting, the conditioning on $L_{\mathbf{A}}^-$ could lead to items no longer being defective independently, which complicates the analysis.  A similar approach can be adopted here, except that we need to avoid overlaps from \emph{all} sub-matrices indexed by $r' \in [\ell]$, rather than only the $r$ value matching the test under consideration.  Thus, when considering a test $t_r \in A_r^c$, we only consider the items that appear in none of the tests indexed by $A_{r'}$ across every $r' \in [\ell]$, where $\mathbf{A} = (A_1,\dotsc,A_{\ell})$.  We denote the number of such items as $z_{t_r, \mathbf{A}}$.  

In the single-threshold case, we lower bounded $z_{t_r, \mathbf{A}}$ in a later part of the analysis, but here we do so immediately to ensure that each $r \in [\ell]$ can be handled separately: We use Assumption \ref{asm:sparsity} to obtain $z_{t_r, \mathbf{A}} \ge z_{t_r} := w_{t_r} - c(d_{\max} - 1)$ with $d_{\max} = \max_i \sum_{r'} d_{i,r'}$ denoting the maximum \emph{total} degree.  We assume for convenience that $z_{t_r}$ is an integer; otherwise, rounding down suffices.

With the above modifications, we are now in a position to follow similar steps as the single-threshold case (i.e., argue the events are increasing, apply FKG, and substitute in the binomial probability), to obtain
\begin{align}
    \label{eq:Pr_L_bfA_plus_bino_0} \Pr(L_{\mathbf{A}}^+ \mid L_{\mathbf{A}}^-) &\geq \prod_{r \in [\ell]} \prod_{t_r \in A_r^c} \left(1 - \sum_{j=0}^{\gamma_r - 1} {z_{t_r,\mathbf{A}} \choose j}q^j(1-q)^{z_{t_r,\mathbf{A}} - j}\right) \\
    &\geq \prod_{r \in [\ell]} \prod_{t_r \in A_r^c} \left(1 - \sum_{j=0}^{\gamma_r - 1} {z_{t_r} \choose j}q^j(1-q)^{z_{t_r} - j}\right), \label{eq:Pr_L_bfA_plus_bino}
\end{align}
where the second step follows because the bracketed term is a binomial upper tail probability, and can only decrease when we reduce the first integer binomial parameter.

{\bf Combining Terms.} Combining \eqref{eq:Pr_L_bfA_minus_bino} and \eqref{eq:Pr_L_bfA_plus_bino} gives
\begin{align}
     \Pr(L_{\mathbf{A}}^- \cap L_{\mathbf{A}}^+) &\geq \left[\prod_{r \in [\ell]} \prod_{t_r \in A_r} \sum_{j=0}^{\gamma_r - 2} {w_{t_r} \choose j} q^j(1-q)^{w_{t_r} - j}\right] \nonumber \\
    \label{eq:Pr_L_plus_minus_bfA_1}& \qquad \times \left[\prod_{r \in [\ell]} \prod_{t_r \in A_r^c} \left(1 - \sum_{j=0}^{\gamma_r - 1} {z_{t_r} \choose j}q^j(1-q)^{z_{t_r} - j}\right)\right] \\
    &= \prod_{r \in [\ell]} \Bigg\{\Bigg[\prod_{t_r \in A_r} \sum_{j=0}^{\gamma_r - 2} {w_{t_r} \choose j} q^j(1-q)^{w_{t_r} - j}\Bigg] \nonumber\\
    \label{eq:Pr_L_plus_minus_bfA_2} &\hspace{40pt}\times \Bigg[\prod_{t_r \in A_r^c} \Bigg(1 - \sum_{j=0}^{\gamma_r - 1} {z_{t_r} \choose j}q^j(1-q)^{z_{t_r} - j}\Bigg)\Bigg]\Bigg\}.
\end{align}

{\bf Summing Over $\mathbf{A} \in \mathcal{A}_{\mathbf{s},i}$.} Summing over all $\mathbf{A} =(A_1,\dots, A_\ell) \in \mathcal{A}_{\mathbf{s},i}$ and using the above lower bound on $\Pr(L_{\mathbf{A}}^- \cap L_{\mathbf{A}}^+)$, we obtain 
\begin{align}
    \Pr(L_{\mathbf{s},i}) &\geq \sum_{\mathbf{A} \in \mathcal{A}_{\mathbf{s},i}} \Bigg\{\prod_{r \in [\ell]} \Bigg\{\Bigg[\prod_{t_r \in A_r} \sum_{j=0}^{\gamma_r - 2} {w_{t_r} \choose j} q^j(1-q)^{w_{t_r} - j}\Bigg] \\
    &\hspace{80pt}\times \Bigg[\prod_{t_r \in A_r^c} \Bigg(1 - \sum_{j=0}^{\gamma_r - 1} {z_{t_r} \choose j}q^j(1-q)^{z_{t_r} - j}\Bigg)\Bigg]\Bigg\}\Bigg\}.
\end{align}
To simplify this, note that since $\mathcal{A}_{\mathbf{s},i} = \mathcal{A}_{s_1,i} \times \cdots \times \mathcal{A}_{s_\ell, i}$, the sum $\sum_{\mathbf{A} \in \mathcal{A}_{\mathbf{s},i}}$ can be equivalently written as $\sum_{A_1 \in \mathcal{A}_{s_1,i}}\cdots \sum_{A_\ell \in \mathcal{A}_{s_\ell, i}}$. Substituting this in and swapping the order of the summations and product (which is justified since the $r$-th term inside the product $\prod_{r \in [\ell]}$ depends only on the $r$-th coordinate of $\mathbf{A}$), we thus obtain
\begin{align}
    \Pr(L_{\mathbf{s},i}) &\geq \prod_{r \in [\ell]} \Bigg\{ \sum_{A_r \in \mathcal{A}_{s_r,i}} \Bigg\{\Bigg[\prod_{t_r \in A_r} \sum_{j=0}^{\gamma_r - 2} {w_{t_r} \choose j} q^j(1-q)^{w_{t_r} - j}\Bigg] \nonumber \\
    \label{eq:Pr_L_bfs_i_rearranged}& \hspace{90pt}\times \Bigg[\prod_{t_r \in A_r^c} \Bigg(1 - \sum_{j=0}^{\gamma_r - 1} {z_{t_r} \choose j}q^j(1-q)^{z_{t_r} - j}\Bigg)\Bigg]\Bigg\}\Bigg\}.
\end{align}
Observe that the terms inside the product correspond exactly to those in the lower bound on $\Pr(L_{s,i})$ given in \eqref{eq:Pr_L_si_lb_init}, with $s$ renamed to $s_r$ (alongside the other relevant quantities $\gamma, A, \mathcal{A}_{s,i}, z_{t,A}, d_i, t, w_t, z_t$ also being renamed). Thus, we can directly use the final lower bound on $\Pr(L_{s,i})$ given in \eqref{eq:final_L_si_lb} to lower bound \eqref{eq:Pr_L_bfs_i_rearranged} by
\begin{align}
    \Pr(L_{\mathbf{s},i}) &\geq \prod_{r \in [\ell]} \Bigg\{ {d_{i,r} \choose s_r} \Bigg[\prod_{t_r \in [d_{i,r}]} \sum_{j=0}^{\gamma_r - 2} {w_{t_r}\choose j}q^j(1-q)^{w_{t_r} - j}\Bigg]^{\frac{s_r}{d_{i,r}}} \nonumber \\
    \label{eq:Pr_bfs_i_lb_final} &\hspace{70pt}\times \Bigg[\prod_{t_r \in [d_{i,r}]}\Bigg(1 - \sum_{j=0}^{\gamma_r -1}{z_{t_r}\choose j}q^j(1-q)^{z_{t_r} - j} \Bigg)\Bigg]^{\frac{d_{i,r}-s_r}{d_{i,r}}}\Bigg\},
\end{align}
where we recall that $z_{t_r} = w_{t_r} - c(d_{\max} - 1)$.

\subsubsection{Step 3: Lower Bounding $\Pr(D_{i})$}
Given the lower bound on $\Pr(L_{\mathbf{s},i})$ in \eqref{eq:Pr_bfs_i_lb_final}, we proceed to lower bound $\Pr(D_{i})$ in a similar manner as the single-threshold case. In particular, we use the relationship between the $D_i$ and $\{L_{\mathbf{s},i}\}_{\mathbf{s} \in \mathbf{S}_i}$ mentioned in Section \ref{sec:modifications} to lower bound $\Pr(D_i)$ by
\begin{align}
     \Pr(D_i) &\geq \sum_{\mathbf{s} \in \mathcal{S}_i}\Bigg\{\prod_{r \in [\ell]} \Bigg\{ {d_{i,r} \choose s_r} \Bigg[\prod_{t_r \in [d_{i,r}]} \sum_{j=0}^{\gamma_r - 2} {w_{t_r}\choose j}q^j(1-q)^{w_{t_r} - j}\Bigg]^{\frac{s_r}{d_{i,r}}} \nonumber \\
    \label{eq:Pr_Di_selectable_init} &\hspace{70pt}\times \Bigg[\prod_{t_r \in [d_{i,r}]}\Bigg(1 - \sum_{j=0}^{\gamma_r -1}{z_{t_r}\choose j}q^j(1-q)^{z_{t_r} - j} \Bigg)\Bigg]^{\frac{d_{i,r}-s_r}{d_{i,r}}}\Bigg\} \Bigg\}.
\end{align}
Recalling that $\mathcal{S}_i = \{0,\dots, d_{i,1}\}\times \cdots \times \{0,\dots, d_{i,\ell}\}$, we can write $\sum_{\mathbf{s} \in \mathcal{S}_i}$ as $\sum_{s_1 \in \{0,\dots, d_{i,1}\}} \cdots \sum_{s_\ell \in \{0,\dots, d_{i,\ell}\}}$.  By similar reasoning as the previous step, we can swap the order of the summations and the product to obtain
\begin{align}
    \Pr(D_i) &\geq \prod_{r \in [\ell]} \Bigg\{ \sum_{s_r  = 0}^{d_{i,r}} \Bigg\{{d_{i,r}\choose s_r} \Bigg[\prod_{t_r \in [d_{i,r}]} \sum_{j=0}^{\gamma_r - 2} {w_{t_r}\choose j}q^j(1-q)^{w_{t_r} - j}\Bigg]^{\frac{s_r}{d_{i,r}}} \nonumber \\
     \label{eq:Pr_Di_selectable_simplified_1}&\hspace{70pt}\times \Bigg[\prod_{t_r \in [d_{i,r}]}\Bigg(1 - \sum_{j=0}^{\gamma_r -1}{z_{t_r}\choose j}q^j(1-q)^{z_{t_r} - j} \Bigg)\Bigg]^{\frac{d_{i,r}-s_r}{d_{i,r}}}\Bigg\} \Bigg\} \\
    &= \prod_{r \in [\ell]}\Bigg\{ \Bigg[\Bigg( \prod_{t_r \in [d_{i,r}]}\sum_{j=0}^{\gamma_r-2}{w_{t_r} \choose j}q^j (1-q)^{w_{t_r} - j}\Bigg)^{\frac{1}{d_{i,r}}} \nonumber \\
    \label{eq:Pr_Di_selectable_simplified_2} &\hspace{70pt}+ \Bigg(\prod_{t_r \in [d_{i,r}]}\Bigg(1 - \sum_{j=0}^{\gamma_r - 1}{z_{t_r} \choose j} q^j(1 - q)^{z_{t_r} - j}\Bigg)\Bigg)^{\frac{1}{d_{i,r}}}\Bigg]^{d_{i,r}} \Bigg\},
\end{align}
with \eqref{eq:Pr_Di_selectable_simplified_2} following from the binomial formula. To further lower bound this, we can lower bound both sums in \eqref{eq:Pr_Di_selectable_simplified_2} analogously to \eqref{eq:first_sum_lb}--\eqref{eq:second_sum_lb} for each $r \in [\ell]$, which gives
\begin{align}
    \Pr(D_i) &\geq \prod_{r \in [\ell]} \Bigg\{\Bigg[\sum_{j=0}^{\gamma_r - 2}{w_{r}^{\max} \choose j}q^j (1-q)^{w_{r}^{\max}-j} + 1 - \sum_{j=0}^{\gamma_r-1} {z_{r}^{\min} \choose j} q^j (1-q)^{z_{r}^{\min} - j}\Bigg]^{d_{i,r}}\Bigg\} \\
    &\geq \prod_{r \in [\ell]} \Bigg\{\Bigg[\sum_{j=0}^{\gamma_r- 2}{w_{r}^{\max} \choose j}q^j (1-q)^{w_{r}^{\max}-j} + 1 - \sum_{j=0}^{\gamma_r-1} {z_{r}^{\min} \choose j} q^j (1-q)^{z_{r}^{\min} - j}\Bigg]^{d_{r}^{\max}}\Bigg\}, \label{eq:bound_by_max_d}
\end{align}
where $d_{r}^{\max} := \max_{i} d_{i,r}$ (and we later similarly use $d_{r}^{\min} := \min_{i} d_{i,r}$), $w_{r}^{\max} \coloneqq \max_{{t_r} \in \mathcal{T}_r} w_{t_r}$ (and similarly $w_{r}^{\min}$), and $z_{r}^{\min} \coloneqq \min_{t_r \in \mathcal{T}_r} z_{t_r} = w_{r}^{\min} - c(d_{\max} - 1)$, recalling that $\mathcal{T}_r \subseteq [T]$ is the set of test indices with threshold $\gamma_r$.  To obtain our final lower bound on $\Pr(D_i)$, we upper bound $d_{\max,r}$ in terms of $T$. Note that \eqref{eq:bound_by_max_d} uses that the argument to $[\cdot]$ is at most one, since it is a lower bound on a binomial probability.

An analogous double counting argument to that in the single-threshold case reveals that $\sum_{i=1}^nd_{i,r} = \sum_{t_r \in \mathcal{T}_r}w_{t_r}$ (recall that $d_{i,r}$ is the number of threshold-$\gamma_r$ tests $i$ is included in), which implies that $nd_{r}^{\min} \leq \lvert \mathcal{T}_r\rvert (w_{r}^{\max}+1)$. Assumption \ref{asm:sub-matrix_regularity} also implies that $nd_{r}^{\min} \geq \frac{n d_{r}^{\max}}{\zeta_r}$, which combined with the aforementioned double counting argument, implies that $d_{r}^{\max} \leq \frac{\zeta_r \lvert \mathcal{T}_r \rvert (w_{r}^{\max}+1)}{n}$ for each $r \in [\ell]$. Recalling that $\alpha_r \in [0,1]$ is defined such that $\lvert \mathcal{T}_r \rvert = \alpha_r T$ for each $r \in [\ell]$, and using the above bound on $d_{r}^{\max}$, we thus have that $\Pr(D_i)$ is lower bounded by
\begin{align}
    \Pr(D_i) &\geq \prod_{r \in [\ell]} \Bigg\{\Bigg[\sum_{j=0}^{\gamma_r - 2}{w_{r}^{\max} \choose j}q^j (1-q)^{w_{r}^{\max}-j} \nonumber \\
    \label{eq:Pr_Di_lb_selectable_final} &\hspace{70pt}+ 1 - \sum_{j=0}^{\gamma_r-1} {z_{r}^{\min} \choose j} q^j (1-q)^{z_{r}^{\min} - j}\Bigg]^{\frac{\zeta_r \alpha_r T (w_{r}^{\max}+1)}{n}}\Bigg\}.
\end{align}

\subsubsection{Lower Bounding the Number of Masked Items, i.i.d.~Prior} \label{sec:step4_multi}

Given the lower bound on $\Pr(D_i)$ derived above, we can follow the same steps as in Section \ref{sec:step4_single} to obtain a high probability lower bound on the number of masked items.  Specifically, we can use the lower bound on $\lvert \mathcal{I}\rvert$ and the fact that its masking events are mutually independent (Lemma \ref{lem:distance4_packing}), alongside \eqref{eq:Pr_Di_lb_selectable_final}, to argue that
\begin{equation}
   \label{eq:stochastic_dom_selectable} {\rm Bin}\left(n^{1-3\xi}, \prod_{r \in [\ell]}\Xi_r^{\frac{\zeta_r \alpha_r T (w_{r}^{\max}+1)}{n}}\right) = {\rm Bin} \left(n^{1-3\xi}, \exp\left(\frac{T}{n}\sum_{r \in [\ell]}\zeta_r \alpha_r (w_{r}^{\max}+1)\log (\Xi_r)\right)\right) \sd M,
\end{equation}
where $\Xi_r$ is the term inside the exponent in \eqref{eq:Pr_Di_lb_selectable_final} for each $r \in [\ell]$. The expected value of the binomial random variable is then
\begin{equation}
    \label{eq:ev_binom_selectable} n^{1-3\xi}\exp\left(\frac{T}{n}\sum_{r \in [\ell]}\zeta_r \alpha_r (w_{r}^{\max}+1)\log (\Xi_r)\right),
\end{equation}
which is at least $n^{2\xi}$ whenever
\begin{align}
    \label{eq:selectable_T_init} T& \leq \frac{(1-5\xi)n \log n}{\displaystyle-\sum_{r \in [\ell]}\zeta_r \alpha_r (w_{r}^{\max}+1) \log (\Xi_r)} \\
    &=  \frac{(1-5\xi)n \log n}{\displaystyle-\sum_{r \in [\ell]} \zeta_r \alpha_r (w_{r}^{\max}+1) \log \Bigg(\sum_{j=0}^{\gamma_r - 2}{w_{r}^{\max} \choose j}q^j (1-q)^{w_{r}^{\max}-j} + 1 - \sum_{j=0}^{\gamma_r-1} {z_{r}^{\min} \choose j} q^j (1-q)^{z_{r}^{\min} - j}\Bigg)}.
\end{align}
Under Assumption \ref{asm:sparsity}, $z_{r}^{\min}$ satisfies $z_{r} = w_{r}(1-o(1))$ for each $r \in [\ell]$.\footnote{Assumption \ref{asm:sub-matrix_regularity} ensures that $d_{i,r} = \Theta(\log n)$ for every $(i,r)$, and since the number of thresholds $\ell$ is constant, this implies that the maximum \emph{total} degree $d_{\max}$ in the definition of $z_{r}$ also scales as $\Theta(\log n)$.}  
Additionally, Assumption \ref{asm:sub-matrix_regularity} gives that $\zeta_r = 1+o(1)$ and $[w_{r}^{\min}, w_{r}^{\max}+1] = [\Delta_r (n/k)(1-o(1)), \Delta_r(n/k)(1+o(1))]$ for each $r \in [\ell]$. Using these and the Poisson approximation to the binomial (Lemma \ref{lem:poisson_approx}), the condition on $T$ can be written as
\begin{align}
    \label{eq:selectable_T_simplified} T &\leq \frac{n \log n}{\displaystyle-\sum_{r \in [\ell]}\alpha_r \Delta_r(n/k) \log\left(\sum_{j=0}^{\gamma_r -2} \frac{\Delta_r^j e^{-\Delta_r}}{j!} + 1 - \sum_{j=0}^{\gamma_r - 1}\frac{\Delta_r^j e^{-\Delta_r}}{j!}\right)}(1-5\xi)(1-o(1)) \\
    \label{eq:selectable_T_simplified} &= \frac{k \log k}{\displaystyle-\sum_{r \in [\ell]}\alpha_r \Delta_r \log \left(1 - \frac{\Delta_r^{\gamma_r - 1}e^{-\Delta_r}}{(\gamma_r - 1)!}\right)}(1-\varepsilon),
\end{align}
where in \eqref{eq:selectable_T_simplified} we have used $n = k^{\frac{1}{1-\xi}}$ and introduced $\varepsilon > 0$ that is arbitrarily small as a function of arbitrarily small $\xi > 0$ (and the $1-o(1)$ term). Taking the minimum over $(\Delta_1, \dots, \Delta_\ell) \in \bbR_+^\ell$, it follows that \eqref{eq:ev_binom_selectable} is at least $n^{2\xi}$ for all $(\Delta_1,\dots, \Delta_\ell)$ as long as
\begin{equation}
    \label{eq:selectable_T_optimised} T \leq \min_{(\Delta_1,\dots \Delta_\ell) \in \mathbb{R}_+^\ell} \frac{k \log k}{\displaystyle-\sum_{r \in [\ell]}\alpha_r \Delta_r \log \left(1 - \frac{\Delta_r^{\gamma_r - 1}e^{-\Delta_r}}{(\gamma_r - 1)!}\right)}(1-\varepsilon).
\end{equation}
To conclude this step, we can proceed in the same way as in the end of Section \ref{sec:step4_single} to show that the event $\mathcal{E} = \{M_0 \leq \frac{1}{4}n^{2\xi}\} \cup \{M_1 = 0\}$ (i.e., the event that there are at most $\frac{1}{4}n^{2\xi}$ masked non-defectives or no masked defectives) holds with probability $\exp(-\Omega(n^\xi))$.

\subsubsection{Wrapping Up}
Finally, as discussed in Section \ref{sec:modifications}, Steps \hyperref[enum:step5_roadmap]{5}-\hyperref[enum:step7_roadmap]{7} can be applied in the exact same way as in the single-threshold case, with their analyses remaining (almost)\footnote{The only difference between steps 5-7 here and steps 5-7 in the single-threshold setting is that the arguments surrounding \eqref{eq:w_t'_conc}--\eqref{eq:w_t'_min_max} need to be repeated for each $r \in [\ell]$, so we can establish the required concentration of $W_{t_r}'$ for each $t_{r} \in \mathcal{T}_r$ and $r \in [\ell]$ via a union bound (under Assumption \ref{asm:sub-matrix_regularity}). Since this straightforward change is the only one across the three steps, we omit the details.} entirely unchanged (see Sections \ref{sec:step5_single}-\ref{sec:step7_single} for these analyses). Since these steps are unchanged, we do not repeat the details, and only note that these steps jointly show that $\Pr(\hat{S} = S) = o(1)$ under the combinatorial prior for any decoder $\hat{S}$, sparsity parameter $\theta$, and test design $\bX$ satisfying Assumptions \ref{asm:sparsity}-\ref{asm:sub-matrix_regularity}, as long as \eqref{eq:selectable_T_optimised} holds. Thus, we can conclude that in order for $\Pr(\hat{S} \neq S) \nrightarrow 1$, it is necessary that
\begin{equation}
    \label{eq:selectable_T_final} T \geq \min_{(\Delta_1,\dots \Delta_\ell) \in \bbR_+^\ell} \frac{k \log k}{\displaystyle-\sum_{r \in [\ell]}\alpha_r \Delta_r \log \left(1 - \frac{\Delta_r^{\gamma_r - 1}e^{-\Delta_r}}{(\gamma_r - 1)!}\right)}(1-o(1)).
\end{equation}
This completes the proof of Theorem \ref{thm: multi}.

\subsection{Proof of Lemma~\ref{cor: bestdenom} (Maximization of $f$)} \label{app:pf_max_f}

Throughout the proof, we denote 
\begin{equation}
    A(\Delta,\gamma)=\dfrac{\Delta^{\gamma-1}}{\Gamma(\gamma)}e^{-\Delta}, \quad \text{ and thus } \quad h(\Delta,\gamma) = -\Delta \log\big( 1-A(\Delta,\gamma) \big), \label{eq:def_A} 
\end{equation}
where $\Gamma(\cdot)$ is the gamma function.  When $\gamma$ is integer-valued (as is the case in Lemma~\ref{cor: bestdenom}) we have the simplification $\Gamma(\gamma) = (\gamma-1)!$, but later in the proof we will consider general real-valued $\gamma > 0$ in \eqref{eq:def_A} in order to permit differentiation with respect to $\gamma$.

By the definitions $\Delta^*_{r} \coloneqq \arg\max_{\Delta> 0} h(\Delta,\gamma_r)$ and $r^* \coloneqq \arg\max_{r \in [\ell]} \gamma_r$, there exist choices $\boldsymbol \alpha^*$ and $\mathbf \Delta^*$ such that $f(\boldsymbol \alpha^*,\mathbf \Delta^*,\boldsymbol \gamma)= h(\Delta^*_{r^*},\gamma_{r^*})$.  To establish Lemma \ref{cor: bestdenom}, it suffices to show that 
    \begin{equation}
        f(\boldsymbol \alpha,\boldsymbol{\Delta},\boldsymbol \gamma)\le h(\Delta^*_{r^*},\gamma_{r^*})\quad\forall \boldsymbol \alpha,\boldsymbol{\Delta}.
    \end{equation}
    Consider an arbitrary tuple $(\boldsymbol \alpha,\boldsymbol{\Delta},\boldsymbol \gamma)$. Denoting $r_0=\arg\max_{r\in[\ell]}h(\Delta_r,\gamma_r)$, we have 
    \begin{equation}
        f(\boldsymbol \alpha,\boldsymbol{\Delta},\boldsymbol \gamma)\le h(\Delta_{r_0},\gamma_{r_0}),
    \end{equation}
    since $f$ is an average of $h(\Delta_r,\gamma_r)$ values and is thus upper bounded by the maximum.  
    We will establish that 
    \begin{equation}
        h(\Delta_{r_0},\gamma_{r_0})\le h(\Delta_{r_0}^*,\gamma_{r_0})\le h(\Delta_{r^*}^*,\gamma_{r^*}).
    \end{equation}
    The first inequality follows directly from the definition of $\Delta_r^*$. We proceed to prove the second inequality, starting with the following lemma limiting the range of $\Delta_{r}^*$ for $\gamma_r \ge 2$, and similarly for non-integer thresholds exceeding 2.  (The case $\gamma_r = 1$ was already handled in standard GT \cite{coja2020optimal}, and will be treated separately later.)
    
    \begin{lemma}\label{lem: boundxbygamma}
        For any real-valued threshold $\gamma \ge 2$ (possibly non-integer), the quantity $\Delta^* = \Delta^*(\gamma)=\arg\max_{\Delta> 0} h(\Delta,\gamma)$ satisfies $\gamma \ge \Delta^*> \gamma-1/3$.
    \end{lemma}
    \begin{proof}
    Since $\Delta^*=\arg\max_{\Delta> 0} h(\Delta,\gamma)$, we have
    $\dfrac{\partial}{\partial \Delta}h(\Delta,\gamma)\Big|_{\Delta=\Delta^*}=0$ (it is easy to verify that $h$ is unimodal with a single peak).  
    Differentiating \(h(\Delta,\gamma)=-\Delta\log(1-A(\Delta,\gamma))\)
    with respect to \(\Delta\), applying the product rule, and using $\frac{\partial}{\partial \Delta} A(\Delta,\gamma) = A(\Delta,\gamma)\big(\frac{\gamma-1}{\Delta}-1\big)$, we obtain
    \begin{equation}
        \frac{\partial}{\partial\Delta} h(\Delta,\gamma) =
        -\log(1-A(\Delta,\gamma)) + \frac{A(\Delta,\gamma)}{1-A(\Delta,\gamma)} (\gamma-1-\Delta).
    \end{equation}
    Setting this derivative to zero at \(\Delta=\Delta^*\) and writing
    \(\bar A=A(\Delta^*,\gamma)\) then gives
    \begin{equation}
        \Delta^*-\gamma+1 = -\frac{1-\bar A}{\bar A}\log(1-\bar A). \label{eq:barA_term}
    \end{equation}
    To establish the desired claim, it remains to show that the right-hand side lies in $[2/3,1]$.  To do so, we proceed with some useful observations:
    \begin{itemize}
        \item The quantity $\bar{A}$ always lies in $(0,1)$.  For integer-valued $\gamma$, this follows directly from $A(\Delta,\gamma)$ being a Poisson probability term.  For general real-valued $\gamma \ge 2$, the lower bound $\bar{A} > 0$ is still immediate, but we defer a stronger upper bound $\bar{A} \le \frac{1}{\sqrt{2\pi}}$ to the paragraph below. 
        \item The quantity $-\frac{1-\bar A}{\bar A}\log(1-\bar A)$ decreases monotonically from $1$ (at $\bar{A} = 0$) to $0$ (at $\bar{A} = 1$); the endpoints are easily checked, and monotonicity follows since the derivative $\frac{\bar{A}+\log(1-\bar{A})}{\bar{A}^2}$ is negative.
        \item The quantity $\Delta^{\gamma-1} e^{-\Delta}$ is maximized (with respect to $\Delta$) at $\gamma - 1$, and thus $(\Delta^{*})^{\gamma-1}e^{-\Delta^*}\le(\gamma-1)^{\gamma-1}e^{-\gamma+1} $.
    \end{itemize}
    The first two dot points immediately imply that \eqref{eq:barA_term} is upper bounded by 1 as desired.  To establish the lower bound of $2/3$, we observe from the third dot point above that $\bar{A}\le \dfrac{(\gamma-1)^{\gamma-1}e^{-\gamma+1}}{\Gamma(\gamma)}\le \frac{1}{\sqrt{2\pi}}$ for $\gamma\ge 2$, which can be verified via the Stirling-like lower bound $\Gamma(1+x) \ge \sqrt{2 \pi x}\big(\frac{x}{e}\big)^x$ for $x \ge 1$ \cite{Batir2017}.\footnote{This is a simplified form of their lower bound that replaces $x^2 + x/3 + a_*$ by $x^2$ in their notation (since $x,a_* > 0$). Their constraint $x \ge 1$ translates to our $\gamma \ge 2$ via the relation $x = \gamma - 1$.}  Combined with the second dot point above, this gives
    \begin{equation}
        -\dfrac{1-\bar{A}}{\bar{A}}\log(1-\bar{A})\ge -\dfrac{1-(2\pi)^{-1/2}}{(2\pi)^{-1/2}}\log(1-(2\pi)^{-1/2})\approx 0.767,
    \end{equation}
    which is strictly higher than the desired lower bound of $2/3$.
    \end{proof}

    Next, we state another lemma stating the conditions under which $h(\Delta^*(\gamma),\gamma)$ is increasing in $\gamma$, again treating $\gamma$ as a general real number (not necessarily an integer).
    
    \begin{lemma}
        Under the definition $\Delta^* = \Delta^*(\gamma) = \arg\max_{\Delta> 0} h(\Delta,\gamma)$, we have the equivalence 
        \begin{align}\label{eq: equiv}
            \frac{d}{d\gamma}\, h(\Delta^*,\gamma)>0
            \Longleftrightarrow
            \log \Delta^*>\Psi(\gamma),
        \end{align}
        where $\Psi(\cdot)$ is the Digamma function, i.e., the derivative of $\log \Gamma(\cdot)$.
    \end{lemma}
    \begin{proof}
     We first calculate $\frac{d}{d \gamma}h(\Delta^*, \gamma)$.  For intuition, suppose first that $\Delta^*(\gamma)$ is differentiable. Then,
    \begin{equation}
        \frac{d}{d\gamma}\, h(\Delta^*,\gamma)
        = \underbrace{\frac{\partial}{\partial \Delta}
        h(\Delta,\gamma)\bigg|_{\Delta=\Delta^*}}_{=0
        \text{ by optimality}}\,
        \frac{d\Delta^*}{d\gamma} + \frac{\partial}{\partial\gamma}
        h(\Delta^*,\gamma)
        = \frac{\partial}{\partial\gamma}
        h(\Delta^*,\gamma).
    \end{equation}
    Even without establishing differentiability of $\Delta^*(\gamma)$, the same conclusion follows from Danskin's theorem \cite{Bertsekas1999}, since $h$ is continuously differentiable and the optimizer $\Delta^*(\gamma)$ is unique (as noted earlier) and confined to a compact interval (see Lemma \ref{lem: boundxbygamma}). 
    Since $h(\Delta,\gamma)=-\Delta\log(1-A(\Delta,\gamma))$ basic differentiation then gives
    \begin{equation}
    \frac{\partial }{\partial \gamma}h(\Delta,\gamma)
    =
    \Delta\cdot \frac{\partial }{\partial \gamma}A(\Delta,\gamma)\cdot \frac{1}{1-A(\Delta,\gamma)}.
    \end{equation}
    Thus, $\frac{\partial }{\partial \gamma}h(\Delta,\gamma)>0$ if and only if $\frac{\partial }{\partial \gamma}A(\Delta,\gamma)>0$ (since $A(\Delta, \gamma) \in (0,1)$ as stated following \eqref{eq:barA_term}). 

    To understand $\frac{\partial }{\partial \gamma}A(\Delta,\gamma)$, it is more convenient to work with the logarithm, namely, $\log A(\Delta,\gamma)=(\gamma-1)\log \Delta-\log\Gamma(\gamma)-\Delta$. Taking the partial derivative with respect to $\gamma$, we obtain
        \begin{equation}
        \frac{\partial}{\partial \gamma}\log A(\Delta,\gamma)=\log \Delta-\Psi(\gamma).
        \end{equation}
        Writing $A(\Delta, \gamma) = \exp(\log(A(\Delta,\gamma)))$ and applying the chain rule for differentiation, it follows that
        \begin{equation}
        \frac{\partial }{\partial \gamma}A(\Delta,\gamma)=A(\Delta,\gamma)\bigl(\log \Delta-\Psi(\gamma)\bigr).
        \end{equation}
    Again using $A(\Delta,\gamma) \in (0,1)$, this establishes \eqref{eq: equiv}. 
\end{proof}

Now, note that due to Lemma \ref{lem: boundxbygamma}, we have $\Delta^*>\gamma-\frac13$, and hence
    \begin{equation}
    \log \Delta^*>\log\Bigl(\gamma-\frac13\Bigr) \ge \log(\gamma)-\dfrac{1}{2\gamma}\ge \Psi(\gamma)\quad \forall \gamma\ge 2,
    \end{equation}
where the middle inequality uses $\log(x-1/3) \ge \log(x) - 1/(2x)$ for $x \ge 2$ (which can be verified analytically or graphically), and the last inequality can be found in \cite[Eq.~6.3.21]{abramowitz}. Thus, by \eqref{eq: equiv}, we have
\begin{equation}
\frac{d}{d\gamma}\, h(\Delta^*,\gamma)>0\quad\text{for all }\gamma\ge 2,
\end{equation}
which establishes that these values are monotonically increasing for $\gamma \ge 2$.  
Moreover $h(\Delta^*,\gamma)$ is strictly higher for $\gamma = 2$ than for $\gamma = 1$: For $\gamma=1$, its value is $(\log 2)^2<0.5$ \cite{coja2020optimal}, whereas for $\gamma = 2$ its value is at least $h(3/2,2) = -\frac{3}{2}\log\big(1-\frac{3}{2} e^{-3/2}\big) >0.6$.  Thus, $\gamma = 1$ can never give the highest value as long as there is a higher integer-valued threshold.  
Since $r^*$ is defined such that $\gamma_{r*}$ is the highest threshold, it follows that $h(\Delta_{r_0}^*,\gamma_{r_0})\le h(\Delta_{r^*}^*,\gamma_{r^*})
$, which completes the proof.

\subsection{Proof of Lemma~\ref{cor: denomconv} (The Large-$\gamma$ Limit)} \label{app:pf_large_conv}

Let
    \begin{equation}
    A_\gamma(\Delta)\coloneqq\frac{\Delta^{\gamma-1}}{(\gamma-1)!}e^{-\Delta},\qquad 
    G_\gamma(\Delta)\coloneqq -\Delta\log\bigl(1-A_\gamma(\Delta)\bigr),\qquad \Delta\in[\gamma-1/3,\gamma],
    \end{equation}
and write $\Delta=\gamma-\varepsilon$ with $\varepsilon\in[0,1/3]$. By Stirling's formula, we have
    \begin{equation}
    (\gamma-1)! = \sqrt{2\pi(\gamma-1)}\,(\gamma-1)^{\gamma-1}e^{-(\gamma-1)}\bigl(1+o(1)\bigr),
    \end{equation}
and hence
    \begin{equation}
    A_\gamma(\gamma-\varepsilon)
    =\frac{(\gamma-\varepsilon)^{\gamma-1}e^{-(\gamma-\varepsilon)}}{\sqrt{2\pi(\gamma-1)}(\gamma-1)^{\gamma-1}e^{-(\gamma-1)}}\bigl(1+o(1)\bigr)
    =\frac{e^{-(1-\varepsilon)}}{\sqrt{2\pi(\gamma-1)}}\Bigl(1+\frac{1-\varepsilon}{\gamma-1}\Bigr)^{\gamma-1}\bigl(1+o(1)\bigr). \label{eq:minus_eps}
    \end{equation}
Since $\varepsilon\in[0,1/3]$ is bounded, we have uniformly in this range of $\varepsilon$ that
    \begin{equation}
    \Bigl(1+\frac{1-\varepsilon}{\gamma-1}\Bigr)^{\gamma-1}
    =\exp\Bigl((\gamma-1)\log\Bigl(1+\frac{1-\varepsilon}{\gamma-1}\Bigr)\Bigr)
    =\exp\bigl((1-\varepsilon)+o(1)\bigr)=e^{1-\varepsilon}\bigl(1+o(1)\bigr).
    \end{equation}
Hence, the exponential factors cancel in \eqref{eq:minus_eps}, giving
    \begin{equation}
    \sup_{\Delta\in[\gamma-1/3,\gamma]}\left|A_\gamma(\Delta)-\frac{1}{\sqrt{2\pi\gamma}}\right|=o\left(\frac{1}{\sqrt{\gamma}}\right),
    \end{equation}
    and in particular
    \begin{equation}
    \sup_{\Delta\in[\gamma-1/3,\gamma]}A_\gamma(\Delta)=O(\gamma^{-1/2}).
    \end{equation}
Using $\,-\log(1-u)=u+O(u^2)$ as $u\to0$, we obtain uniformly for $\Delta\in[\gamma-1/3,\gamma]$ that
    \begin{equation} G_\gamma(\Delta)=\Delta\Bigl(A_\gamma(\Delta)+O(A_\gamma(\Delta)^2)\Bigr)=\Delta\cdot A_\gamma(\Delta)+O\Bigl(\Delta\cdot A_\gamma(\Delta)^2\Bigr).
    \end{equation}
Now, $\Delta\in[\gamma-1/3,\gamma]$ implies $\Delta=\gamma-O(1)$, while $A_\gamma(\Delta)=O(\gamma^{-1/2})$, so
    \begin{equation}
    \Delta \cdot A_\gamma(\Delta)^2 = O(\gamma\cdot \gamma^{-1})=O(1),
    \quad\text{and}\quad
    \Delta\cdot A_\gamma(\Delta)=\bigl(\gamma-O(1)\bigr)\Bigl(\frac{1}{\sqrt{2\pi\gamma}}+o(\gamma^{-1/2})\Bigr)
    =\sqrt{\frac{\gamma}{2\pi}}\bigl(1+o(1)\bigr),
    \end{equation}
uniformly over $\Delta\in[\gamma-1/3,\gamma]$. Multiplying by $\sqrt{2\pi/\gamma}$ gives
    \begin{equation}
    \sup_{\Delta\in[\gamma-1/3,\gamma]}\left|\,G_\gamma(\Delta)\sqrt{\frac{2\pi}{\gamma}}-1\,\right| \to 0,
    \end{equation}
which establishes the first part of the lemma. The second part follows directly from Lemma~\ref{lem: boundxbygamma}. 

\section{Verification of Assumptions for the NCC Design} \label{app:verify}

In this appendix, we verify that Assumptions \ref{asm:regularity} and \ref{asm:sparsity} hold with probability $1-o(1)$ under the NCC design (Definition \ref{def:NCC}) with $L = \frac{\nu T}{k}$ for any fixed parameter $\nu > 0$ and any number of tests satisfying $T = \Theta(k \log n)$.  For convenience, we write $T = C k \log n$ with $C = \Theta(1)$.

Assumption \ref{asm:regularity} can be verified (almost) directly from the findings of \cite{coja2020information}:
\begin{itemize}
    \item The first step in \cite[App.~D.2]{coja2020information} is precisely to establish that each item is placed in $L - O(1)$ distinct tests.  Since $L \to \infty$, this is a strictly stronger statement than our requirement of $L(1-o(1))$.
    \item \cite[Lemma 2.4]{coja2020information} states that with probability $1-o(1)$, all row weights behave as $\frac{\nu n}{k} + O\big( \sqrt{\frac{\nu n}{k}} \log n \big)$ \emph{when repeated placements (i.e., ``multi-edges'') are counted}.  When $\nu = \Theta(1)$ and $k = \Theta(n^{\theta})$, this simplifies to the desired behavior $\frac{\nu n}{k}(1+o(1))$.

    It remains to establish that removing repeated placements changes each row weight by only $o(n/k)$. For a fixed test, the expected number of items placed in that test at least twice is $O(n/k^2)$, since for each item this probability is at most $\binom{L}{2}/T^2=O(k^{-2})$. A Chernoff bound, together with a union bound over all tests and the $O(1)$ bound on the total number of repeated placements per item from
    \cite[Appendix D.2]{coja2020information}, then yields a uniform $o(n/k)$ bound on the number of repeated placements per test. Since this argument is largely similar to \eqref{eq:both_prob}--\eqref{eq:c_dense} below, we omit the remaining details.
    
\end{itemize}
In the following, we verify Assumption \ref{asm:sparsity}. 
Fix two distinct tests $t_1,t_2\in[T]$, and let $V_{t_1,t_2}$ denote the number of items appearing in both tests.  For a given item, the probability of appearing in both specified tests is at most
\begin{equation}
    \Pr(i \text{ in both } t_1 \text{ and } t_2) \le \left(\frac{L}{T}\right)^2 = \frac{\nu^2}{k^2}. \label{eq:both_prob}
\end{equation}
Since items are placed independently, we have the stochastic domination
\begin{equation}
    V_{t_1,t_2} \sd \mathrm{Bin}\left(n,\frac{\nu^2}{k^2}\right). \label{eq:overlap_domination}
\end{equation}
The mean of the dominating binomial variable is
\begin{equation}
    \mu_V = \nu^2\frac{n}{k^2} = \Theta(n^{1-2\theta}). \label{eq:overlap_mean}
\end{equation}
We now split into two regimes.

{\bf The case $\theta < 1/2$.} Here $\mu_V=n^{\Omega(1)}$.  By the Chernoff bound, for any fixed pair $(t_1,t_2)$, we have
\begin{equation}
    \Pr\left( V_{t_1,t_2} \ge 2\mu_V \right) \le \exp(-\Omega(\mu_V)).
\end{equation}
Since there are at most $T^2=O(k^2\log^2 n)$ pairs of tests, and $\mu_V=n^{\Omega(1)}$, a union bound gives
\begin{equation}
    \max_{t_1\ne t_2} V_{t_1,t_2} = O\left(\frac{n}{k^2}\right)
    = O(n^{1-2\theta}) \label{eq:c_sparse}
\end{equation}
with probability $1-o(1)$.  Thus, Assumption \ref{asm:sparsity} holds with a polynomial ``margin'', in the sense that $n^{1-2\theta}$ is smaller than $\frac{n}{k \log n}$ by a factor of $O\big(\frac{\log n}{n^{\theta}}\big)$.

{\bf The case $\theta \ge 1/2$.} Here $\mu_V=O(1)$.  Fix $A > 0$ and let
\begin{equation}
    a_n = \frac{A\,\log n}{\log\log n}
\end{equation}
Using the standard binomial tail bound
\begin{equation}
    \Pr(V \ge a)
    \le \left(\frac{e \mathbb{E}[V]}{a}\right)^{a}, \label{eq:bin_tail_large}
\end{equation}
we obtain for sufficiently large $A$ that
\begin{equation}
    \Pr( V_{t_1,t_2}\ge a_n ) \le O(n^{-3}).
\end{equation}
Since $T^2=O(k^2\log^2 n)\le O(n^2\log^2 n)$, a union bound yields with probability $1-o(1)$ that
\begin{equation}
    \max_{t_1\ne t_2} V_{t_1,t_2}
    = O\left(\frac{\log n}{\log\log n}\right). \label{eq:c_dense}
\end{equation}
This is again significantly smaller than the required $o\big( \frac{n}{k \log n} \big)$ in Assumption \ref{asm:sparsity}; this time the bound is sub-logarithmic in $n$ when even a (small enough) polynomial upper bound would have sufficed.

\bibliography{ref}
\bibliographystyle{ieeetr}
\end{document}